\documentclass[sigconf,nonacm,balance=false]{acmart}

\usepackage{tikz}
\usetikzlibrary{shapes, arrows.meta, positioning,fit, calc,patterns, patterns.meta}
\usepackage[english]{babel}
\usepackage[utf8]{inputenc}

\usepackage{tabularray}
\usepackage{supertabular}
\usepackage{booktabs}
\usepackage{graphicx}

\usepackage[shortlabels]{enumitem}
\usepackage{mathtools}
\usepackage{caption}
\usepackage{float}
\usepackage{subcaption}
\usepackage{comment}
\usepackage{xcolor}
\usepackage{xspace}
\usepackage{url}
\usepackage{algorithm}
\usepackage[noend]{algpseudocode}
\usepackage{xparse}
\algrenewcommand\alglinenumber[1]{\footnotesize #1}

\definecolor{algcomment}{RGB}{100,120,160}

\makeatletter
\newenvironment{breakablealgorithm}
  {%
   \par\vspace{\intextsep}%
   \begin{center}
     \refstepcounter{algorithm}%
     \hrule height.8pt depth0pt \kern2pt%
     \renewcommand{\caption}[2][\relax]{%
       {\raggedright\textbf{\ALG@name~\thealgorithm} ##2\par}%
       \ifx\relax##1\relax
         \addcontentsline{loa}{algorithm}{\protect\numberline{\thealgorithm}##2}%
       \else
         \addcontentsline{loa}{algorithm}{\protect\numberline{\thealgorithm}##1}%
       \fi
       \kern2pt\hrule\kern2pt
     }
  }{%
     \kern2pt\hrule\relax%
   \end{center}
   \par\vspace{\intextsep}%
  }
\makeatother

\newtheorem{theorem}{Theorem}
\newtheorem{lemma}{Lemma}
\newtheorem{corollary}{Corollary}
\newtheorem{definition}{Definition}

\newtheorem*{theorem*}{Theorem}

\newcommand{\E}{\mathbb{E}}

\newcommand{\protocolname}{\textsc{AetherWeave}\xspace}

\DeclareMathOperator*{\argmax}{arg\,max}

\newcommand{\sk}{\ensuremath{sk}\xspace}
\newcommand{\stakesk}{\ensuremath{StakeSk}\xspace}
\newcommand{\stakeID}{\ensuremath{StakeID}\xspace}
\newcommand{\netpk}{\ensuremath{NetPk}\xspace}
\newcommand{\netsk}{\ensuremath{NetSk}\xspace}
\newcommand{\sig}{\ensuremath{{\sigma}}\xspace}

\newcommand{\Rstake}{\ensuremath{\mathcal{R}_\mathsf{stk}}\xspace}
\newcommand{\Rshare}{\ensuremath{\mathcal{R}_\mathsf{shr}}\xspace}

\newcommand{\netrecord}{\ensuremath{\textsc{NetRec}}\xspace}
\newcommand{\peerrecord}{\ensuremath{\textsc{PeerRec}}\xspace}

\newcommand{\request}{\ensuremath{\textsc{Request}}\xspace}
\newcommand{\response}{\ensuremath{\textsc{Response}}\xspace}
\newcommand{\slashproof}{\ensuremath{\textsc{SlashProof}}\xspace}
\newcommand{\commitrecord}{\ensuremath{\textsc{CommitRec}}\xspace}

\newcommand{\ptable}{\ensuremath{T_{gsp}}\xspace}
\newcommand{\otable}{\ensuremath{T_{priv}}\xspace}

\newcommand{\stake}{\ensuremath{{stake}}\xspace}
\newcommand{\stakeroot}{\ensuremath{StakeCom}\xspace}
\newcommand{\timestamp}{\ensuremath{ts}\xspace}
\newcommand{\nonce}{\ensuremath{\nu}\xspace}
\newcommand{\nonceoverlay}{\ensuremath{\eta}\xspace}
\newcommand{\reqcommit}{\ensuremath{ReqCom}\xspace}
\newcommand{\slashshare}{\ensuremath{share}\xspace}
\newcommand{\shareproof}{\ensuremath{\pi_{share}}\xspace}
\newcommand{\stakeproof}{\ensuremath{\pi_{stake}}\xspace}
\newcommand{\ind}{\ensuremath{ind}\xspace}
\newcommand{\indproof}{\ensuremath{\pi_{\ind}}\xspace}
\newcommand{\roundnum}{\ensuremath{round}\xspace}
\newcommand{\denylist}{\ensuremath{DenyLST}\xspace}
\newcommand{\addr}{\ensuremath{ADDR}\xspace}

\newcommand{\stakefreeze}{\ensuremath{\Delta_{freeze}}\xspace}
\newcommand{\stakewithdraw}{\ensuremath{\Delta_{withdraw}}\xspace}
\newcommand{\recordexp}{\ensuremath{\Delta_{exp}}\xspace}

\algnewcommand{\Sample}[2]{\State #1 \sim #2}
\newcommand{\require}[0]{{\bf require} }

\newcommand{\keygen}{\ensuremath{KeyGen}\xspace}
\newcommand{\Sign}{\ensuremath{Sign}\xspace}
\newcommand{\CheckSig}{\ensuremath{CheckSig}\xspace}
\newcommand{\hashid}{\ensuremath{H_{id}}\xspace}
\newcommand{\hashstake}{\ensuremath{H_{stake}}\xspace}
\newcommand{\hashshare}{\ensuremath{H_{share}}\xspace}
\newcommand{\veccommit}{\ensuremath{VecCommit}\xspace}
\newcommand{\vecopen}{\ensuremath{VecOpen}\xspace}
\newcommand{\vecverify}{\ensuremath{VecVerify}\xspace}
\NewDocumentCommand{\PRNG}{o o}{%
  \IfNoValueTF{#1}{%
    \ensuremath{\mathsf{PRNG}}\xspace%
  }{%
    \ensuremath{\mathsf{PRNG}_{#1}(#2)}%
  }%
}

\newcommand{\negl}{\ensuremath{\mathsf{negl}}\xspace}
\newcommand{\stgame}{\ensuremath{\mathsf{Game}^{\mathsf{stake\text{-}anon}}_{\mathcal{A}}(\lambda)}\xspace}
\newcommand{\cgame}{\ensuremath{\mathsf{Game}^{\mathsf{conn\text{-}priv}}_{\mathcal{A}}(\lambda)}\xspace}
\newcommand{\stadv}{\ensuremath{\mathsf{Adv}^{\mathsf{stake\text{-}anon}}_{\mathcal{A}}(\lambda)}\xspace}
\newcommand{\cadv}{\ensuremath{\mathsf{Adv}^{\mathsf{conn\text{-}priv}}_{\mathcal{A}}(\lambda)}\xspace}
\newcommand{\Adv}{\ensuremath{\mathsf{Adv}}\xspace}
\newcommand{\Advset}{\ensuremath{\mathcal{A}}\xspace}
\newcommand{\Honset}{\ensuremath{\mathcal{H}}\xspace}
\newcommand{\pconn}{\ensuremath{p_{c}}\xspace}
\newcommand{\quality}{\ensuremath{q}\xspace}
\newcommand{\visibility}{\ensuremath{v}\xspace}

\newcommand{\monthyear}{
    \ifcase \month \or January\or February\or March\or %
    April\or May \or June\or July\or August\or September\or October\or
    November\or December
    \fi\, \number \year}

\usepackage[breakable,most]{tcolorbox}
\newtcolorbox{enumbox}[2][]{
    breakable,
    colback=white,
    colframe=black,
    boxrule=0.8pt,
    colbacktitle=white,
    coltitle=black,
    enhanced,
    left=2mm,
    right=2mm,
    top=2mm,
    bottom=2mm,
    fonttitle=\bfseries,
    title={\strut #2},
    #1
}

\newlength{\awplotwidth}
\newlength{\awplotiiwidth}
\newlength{\awplotheight}
\usepackage{pgfplots}
\usepgfplotslibrary{fillbetween}
\usepgfplotslibrary{colorbrewer}
\pgfplotsset{compat=1.18}
\pgfplotscreateplotcyclelist{bright}{
    {plotA, thick},
    {plotB, thick},
    {plotC, thick},
    {plotD, thick},
    {plotE, thick},
    {plotF, thick},
}
\usepackage[capitalise,nameinlink]{cleveref}

\crefname{theorem}{Theorem}{Theorems}
\crefname{lemma}{Lemma}{Lemmas}
\crefname{corollary}{Corollary}{Corollaries}
\crefname{definition}{Definition}{Definitions}
\crefname{proposition}{Proposition}{Propositions}
\crefname{algorithm}{Alg.}{Algs.}
\Crefname{algorithm}{Algorithm}{Algorithms}
\newcommand{\algline}[2]{Alg.~\ref{#1}:\ref{#2}}
\newcommand{\alglines}[3]{Alg.~\ref{#1}:\ref{#2}--\ref{#3}}
\crefname{appendix}{Appendix}{Appendices}
\Crefname{appendix}{Appendix}{Appendices}
\pgfplotsset{
    awplot/.style={
        grid=major,
        major grid style={solid,draw=gray!50},
        tick label style={font=\scriptsize, /pgf/number format/fixed},
        label style={font=\footnotesize},
        legend style={
            draw=none,
            fill=none,
            font=\scriptsize,
            cells={anchor=west}
        },
        axis line style={draw=black, line width=0.5pt},
        tick align=outside,
        scaled y ticks=false,
        cycle list name=bright,
    }
}
\definecolor{plotA}{HTML}{4477AA}  %
\definecolor{plotB}{HTML}{EE7733}  %
\definecolor{plotC}{HTML}{228833}  %
\definecolor{plotD}{HTML}{CC3311}  %
\definecolor{plotE}{HTML}{AA3377}  %
\definecolor{plotF}{HTML}{CCBB44}  %

\definecolor{SciBlue}{HTML}{332288}
\definecolor{SciCyan}{HTML}{88CCEE}
\definecolor{SciTeal}{HTML}{44AA99}
\definecolor{SciGreen}{HTML}{117733}
\definecolor{SciOlive}{HTML}{999933}
\definecolor{SciSand}{HTML}{DDCC77}
\definecolor{SciRose}{HTML}{CC6677}
\definecolor{SciWine}{HTML}{882255}
\definecolor{SciPurple}{HTML}{AA4499}
\definecolor{SciGray}{HTML}{DDDDDD}

\newcommand{\nf}{\mathsf{NoFlag}} 
\newcommand{\dc}{\mathsf{Discon}} 
\newcommand{\hv}{\mathsf{Heavy}} \newcommand{\HV}{\textsc{Heavy}}
\newcommand{\Bad}{\textsc{Bad}} \newcommand{\flag}{\mathsf{flag}}

\usepackage{adjustbox}

\begin{document}

\title{AetherWeave: Sybil-Resistant Robust Peer Discovery with Stake}

\setcopyright{acmlicensed}

\author{Kaya Alpturer}
\affiliation{%
  \institution{Princeton University}
  \city{}
  \country{}
}
\email{kalpturer@princeton.edu}

\author{Constantine Doumanidis}
\affiliation{%
  \institution{Princeton University}
  \city{}
  \country{}
}
\email{doumanidis@princeton.edu}

\author{Aviv Zohar}
\affiliation{%
  \institution{The Hebrew University}
  \city{}
  \country{}
}
\email{avivz@cs.huji.ac.il}

\begin{abstract} Peer-discovery protocols within P2P networks are often vulnerable: because
creating network identities is essentially free, adversaries can eclipse honest
nodes or partition the overlay. This threat is especially acute for blockchains,
whose security depends on resilient peer connectivity. We present
{\protocolname}, a stake-backed peer-discovery protocol that ties network
participation to deposited stake, raising the cost of large-scale attacks. 
We prove that, with high probability, either the
honest overlay remains connected or a $(1{-}\delta)$-fraction of nodes in every
smaller component raise an attack-detection flag---even against a very powerful 
adversary. To our knowledge, {\protocolname} is the first peer-discovery
protocol to simultaneously provide Sybil resistance and privacy: nodes prove
they hold valid stake without revealing which deposit they own, and gossiping
does not expose peer-table contents. A cryptographic commitment scheme
rate-limits discovery requests per round; exceeding the limit yields a publicly
verifiable misbehavior proof that triggers on-chain slashing. Beyond deposit and
slashing, the protocol requires no on-chain interaction, with per-node
communication scaling as $O(s\sqrt{n})$. We validate our design through a
mean-field analysis with closed-form convergence bounds, extensive adversarial
simulations, and an end-to-end prototype built by forking Prysm, a leading
Ethereum consensus client.
 \end{abstract}

\maketitle
\bibliographystyle{ACM-Reference-Format}

\section{Introduction}
Peer-to-Peer (P2P) networks form the backbone of many decentralized systems, yet their security suffers from a lack of a solid identity system.
Because creating network identities is essentially free, an attacker who
controls many IP addresses can flood the peer-discovery process with adversarial
entries, eclipsing victim nodes or partitioning the network into disconnected
components~\cite{heilman2015eclipsebitcoin,marcus2018eclipseethereum,douceur2002sybil}.
Existing defenses (IP subnet bans, rate-limiting, reputation
tracking~\cite{yu2006sybilguard,yu2008sybillimit}) either inconvenience a
well-resourced attacker or require trust assumptions that conflict with the
open-participation model of modern P2P systems.

A natural idea is to tie network participation to economic stake. In
proof-of-stake blockchains, validators already lock substantial deposits, so
repurposing this existing collateral for peer discovery imposes no additional
cost on honest participants; more broadly, any P2P system where participants can
post collateral benefits from this approach.
Coretti et al.~\cite{coretti2022scuttlebutt} take a significant step and establish a theoretical foundation
for this direction, showing that a stake-weighted random graph can sustain
Byzantine-resilient gossip. However, their analysis assumes a known
mapping from stake deposits to addresses and thus abstracts away the \emph{peer-discovery problem} itself: how
nodes discover the address of one another, handle address churn, and are able to bootstrap the overlay
in the first place. Privacy---the ability to participate without linking one's network
presence to a specific on-chain deposit---is also left unaddressed.

In this paper we present \protocolname, a \textbf{stake-backed peer-discovery
protocol} that bridges this gap. \protocolname leverages stake---locked
funds---to constrain the action space of would-be attackers. Requiring stake
deposits introduces a large cost to gaining over-representation in the system,
and thereby forms the basis of honest and robust network formation. Stake acts
both as a surrogate for identity and as a punishment mechanism: malicious
entities can have their stake slashed as punishment for misbehavior, ejecting
them from the system at least until they invest additional funds.

A key challenge in tying peer discovery to economic stake is preserving
privacy: naively, each node's network address would be linkable to its on-chain
deposit, creating a surveillance risk that conflicts with the goals of open
participation. This tension is deepened by the need for slashing---punishing
misbehavior requires eventually identifying the offender's deposit, yet honest
nodes must remain anonymous. \protocolname resolves this by keeping stake
identity and network identity unlinkable for well-behaved participants: a node
proves it \emph{has} valid stake without revealing \emph{which} deposit it
owns, and gossiping about peers does not reveal a node's overlay connections.
Anonymity is broken only for attackers: misbehavior produces a cryptographic
proof that ties the offense to a specific deposit, enabling on-chain slashing
without requiring ongoing blockchain interaction. In particular, nodes that
request too many peer addresses---even when spreading requests across multiple
peers---are promptly detected.

While we rely on a blockchain to provide staking
capability, \protocolname is not restricted to the blockchain
domain: any P2P system that can support deposit-based participation can adopt
our protocols. Nodes without stake can still join on a best-effort basis, while
staked nodes form the secure core of the network. Application domains include
\emph{blockchain peer discovery}, where Bitcoin and Ethereum currently rely on
ad-hoc defenses against eclipse and Sybil
attacks~\cite{heilman2015eclipsebitcoin,marcus2018eclipseethereum} (our
prototype demonstrates augmenting Ethereum's peer discovery with stake-backed
guarantees); \emph{decentralized content delivery}, where
IPFS~\cite{benet2014ipfs} is vulnerable to Sybil-based content
eclipsing~\cite{cholez2024ipfs} and already integrates with deposit-based
storage markets (Filecoin~\cite{protocollabs2017filecoin}); and
\emph{anonymous communication}, where the Tor
network~\cite{dingledine2004tor} has suffered real-world Sybil
attacks~\cite{winter2016sybiltor} and \protocolname's stake-anonymity property
would raise the cost of such attacks without compromising anonymity.

An important limitation must be acknowledged: no protocol can absolutely
prevent partitioning. A node that bootstraps from a dishonest peer has no
external reference point---every record it receives is mediated by an
adversary that can suppress honest entries at will. More broadly, a
sufficiently powerful network adversary can selectively drop messages between
honest nodes, and no local algorithm can distinguish a genuine absence of
peers from adversarial suppression. \protocolname addresses this with a
two-tiered guarantee: against weaker adversaries, partitioning is
\emph{prevented}---the honest overlay remains connected with high
probability; against stronger adversaries where prevention may fail, every
successful eclipse or partition causes affected nodes to raise an
attack-detection flag, enabling corrective action (e.g., switching
bootstrap contacts or alerting the operator).

\subsection{Main contributions}
\protocolname is a round-based protocol in which nodes discover peers through stake-weighted gossip. We make the following contributions:

\begin{description}[nosep, leftmargin=0pt, style=unboxed]
	\item[Partition detection.] We prove that, with high probability,
	      either the honest overlay is fully connected or every smaller
	      component has at least a $(1{-}\delta)$-fraction of its nodes raising
	      an attack-detection flag (\cref{thm:overlay}). This guarantee holds
	      even against an omniscient adversary that knows every honest node's
	      private seeds and can manipulate message delivery in the gossip phase. A node can
	      bootstrap from any single honest peer and reliably detect that it is not eclipsed
	      (\cref{sec:overlay-resistance,sec:detecting-partitions}).
	\item[Spam prevention via slashing.] We introduce a cryptographic
	      commitment scheme that rate-limits peer-discovery requests per round.
	      Nodes that exceed the limit produce a publicly verifiable proof of
	      misbehavior, enabling on-chain slashing of their stake. Detection
	      occurs within $O(1)$ rounds with overwhelming probability
	      (\cref{sec:too-many-reqs,sec:security}), and becomes increasingly more likely if nodes exceed their budget by larger amounts.
	\item[Privacy.] To our knowledge, \protocolname is the first
	      peer-dis\-covery protocol to simultaneously achieve Sybil
	      resistance and privacy. Network identity is decoupled from
	      stake ownership: the protocol reveals that a node \emph{has}
	      stake, but not \emph{which} deposit it owns. Gossiping does
	      not expose a node's peer table to other participants
	      (\cref{sec:security}).
	\item[Efficiency.] The protocol requires no on-chain interaction
	      except for deposit and slashing. Per-node communication scales as
	      $O(s\sqrt{n})$ per round, where $n$ is the network size and $s$ is
	      a small constant (\cref{sec:mean-field,sec:eval,sec:prototype}).
\end{description}

We validate these properties through a mean-field analysis that yields
closed-form convergence bounds for table quality and node visibility
(\cref{sec:mean-field}), extensive simulations (\cref{sec:eval}), and experiments based on an
end-to-end prototype built by forking Prysm, a leading Ethereum consensus
client (\cref{sec:prototype}). The prototype includes an on-chain staking
contract, zero-knowledge proof circuits, and a modified libp2p stack. Our
experiments confirm that all core guarantees hold even under scaled-down
parameters.

We begin with a brief overview in \cref{sec:overview} before formalizing the model and protocol in \cref{sec:model,sec:protocol}.

\section{Brief Overview of \protocolname} \label{sec:overview}

The \protocolname protocol operates in a network of $n$ staked nodes. Each node
identifies itself by a public key \netpk bound to its on-chain stake, and
maintains two peer tables, each of size $s\sqrt{n}$ (for a constant $s >
1$). The first, \ptable (``gossip''), is used for peer-record gossip; the
second, \otable (``private''), supplies the entries from which overlay
connections are later drawn. The two
tables are populated using different private seeds, so gossip responses (served
from \ptable) reveal nothing about a node's overlay neighbors (drawn from
\otable).

\paragraph{Refreshing Peer Tables}
Nodes periodically refresh their peer tables to account for churn.
Each round, a node selects a pseudorandom \emph{slice} of the identifier space
by sampling fresh seeds, then sends requests to every peer in \ptable. Each
responder filters its own table and returns only the records matching the
requester's slice, keeping responses small ($s^2$ records in expectation).
Intuitively, $s\sqrt{n}$ peers each holding $s\sqrt{n}$ uniformly sampled
records collectively cover a $1 - e^{-s^2}$ fraction of the network; as $s$
grows the coverage nears perfection, but per-node communication increases, so
$s$ trades off reliability against bandwidth.

Each returned peer record includes the \netpk of the node, a \emph{recent}
cryptographic proof of stake (via a vector commitment maintained by a smart
contract), and a signed network address \addr.

We assume that the adversary is economically bounded and controls an $\alpha$
fraction of the total stake. Analytical models and simulations show that for
$s^2(1-\alpha)>1$, peer records propagate effectively and honest nodes maintain
diverse peer tables, resulting in a robust, well-connected network.

\paragraph{Detecting Eclipse Attacks}
Because record selection is deterministically defined by the requester's seeds,
adversaries cannot \emph{inflate} their representation---irrelevant records are
filtered out. The only remaining avenue is \emph{suppression} of honest
records, which produces a statistical signal: if the number of unique records a
node collects falls significantly below $s\sqrt{n}$, the node raises an
attack-detection flag (\cref{sec:eclipse-detect}).

\paragraph{Overlay Construction}
Once \otable has been populated through gossip, each node independently includes
each peer as an overlay neighbor with probability $\pconn = c \cdot \log n /
(s\sqrt{n})$ for a suitable constant $c > 1$, yielding $\Theta(\log n)$
connections per node---sufficient for the honest overlay to be connected with
high probability (\cref{thm:overlay}).

\paragraph{Privacy}
The protocol preserves privacy through two mechanisms.
First, \emph{unlinkable identifiers}: each \netpk is deterministically
derived from the same secret controlling the on-chain stake but remains
unlinkable to it; stake proofs use Zero Knowledge Proofs (ZKPs) and hence do not break this
unlinkability. Second, \emph{private peer requests}: the seeds $\nonce,
\nonceoverlay$ remain secret, and nodes use private information retrieval
(via trusted execution environments) to query peers’ records without
revealing the retrieved values.

\paragraph{Punishing Excessive Requests}
Responding to heartbeat requests involves computational work, making excessive
querying a potential DoS vector. \protocolname enforces a global per-round limit
of $s\sqrt{n}$ requests per node. Each request carries a commitment to the full
batch of intended recipients for that round, together with a cryptographic share
of the requester's slashing secret. A node that issues requests under two
distinct batch commitments in the same round reveals enough information to
reconstruct its slashing secret, enabling anyone to burn its stake on-chain.
Thus, nodes remain unlinkable to their stake while honest, but lose this
protection if they exceed their quota. This mechanism enforces a global request
quota without a trusted aggregator---detection requires observing just two
conflicting commitments from the same node in a single round.

\section{Model and Notation} \label{sec:model} We consider a P2P system with $n$
nodes, indexed by $[n] = \{1, \ldots, n\}$. The system comprises two components:
the network $\mathcal{N}$ and the blockchain $\mathcal{B}$.

\paragraph{Network model}
Nodes communicate over $\mathcal{N}$, which is capable of relaying messages
between any pair of peers. No guarantees are made regarding message delivery
(messages may be delayed or dropped). Furthermore, communication at the network
layer is not assumed to be private and adversaries may learn the contents of
messages. Hence, our protocol establishes authenticated and private
communication over channels. Each node $i$ is assigned a network address
$\addr_i$, which may change over time due to churn. An address $\addr_i$ may
represent a public IP address or an endpoint in an anonymizing network such as
Tor. The adversary may obtain an unbounded number of such addresses at will.
Since {\addr}s are only semi-persistent, the system employs continuous gossip to
maintain a table of connectable peers.

\paragraph{Blockchain model}
The blockchain $\mathcal{B}$ tracks stake ownership and enforces slashing
conditions via a smart contract. Each node possesses a stake identifier
$\stakeID$ that is associated with one unit of stake on $\mathcal{B}$. The
adversary is \emph{economically bounded}, controlling a fraction $\alpha$ of the
total staked funds. Protocol participants can periodically query a smart
contract on $\mathcal{B}$ to obtain a commitment $\stakeroot$ to the current
stake allocation. To avoid the need to constantly refresh $\stakeroot$, or to
remember many versions, the stake allocation in $\mathcal{B}$ is allowed to
change only once per epoch (for example, once a week). To allow slashing
penalties time to occur, withdrawals are delayed somewhat after the end of the
epoch in case misbehavior occurred just before it ended.

\paragraph{The Peer Discovery Problem}
The goal of the protocol is for each node to maintain a peer table $\otable^t$
listing the network addresses of other nodes at time $t$. Overlay connections
are established by randomly sampling entries from $\otable^t$.

\begin{definition}[eclipsed set] \label{def:eclipsed}
A set of honest nodes is \textbf{eclipsed} if every peer they select is either
adversarial or within the same set, effectively isolating them from the
remaining honest nodes. (A single node whose selected peers are all adversarial
is the special case.)
\end{definition}

The peer-discovery problem centers around the challenge of disseminating
information about known peers to others, especially to new nodes that are just
joining the system. Since available network addresses are subject to churn,
information stored in peer tables must be constantly refreshed. The objective is
to maintain peer tables in which the proportion of stale and adversarial entries
remains small, thereby reducing the risk of eclipse attacks.

\paragraph{Cryptographic primitives.}
The protocol relies on standard cryptographic building blocks: an
EUF-CMA signature scheme
$(\keygen,\allowbreak\Sign,\allowbreak\CheckSig)$,
non-interactive zero-knowledge proofs, collision-resistant hash
functions $\hashstake$, $\hashid$, $\hashshare$ (modeled as random
oracles), Merkle-tree vector commitments
$(\veccommit,\allowbreak\vecopen,\allowbreak\vecverify)$, and a
pseudorandom number generator $\PRNG$. Full definitions are given in
\cref{appendix:primitives}; notation is summarized in
\cref{appendix:symbols}.

\subsection{Attacker Model}
We consider a computationally bounded (PPT) adversary that
controls an $\alpha$ fraction of the total stake ($\alpha<1$) and may
coordinate all accounts it controls. We write $\Honset$ and $\Advset$ for the
honest and adversarial node sets, respectively. During peer discovery, $\mathcal{A}$ can
eavesdrop on, delay, reorder, drop, and replay all messages, including
selectively denying connectivity between honest nodes. The network therefore
provides no confidentiality, integrity, or delivery guarantees beyond what is
achieved cryptographically. Once an honest node has obtained another honest
node's authenticated address through discovery, we assume the two can establish
a direct connection; the adversary cannot permanently block connections between
honest nodes that already know each other's addresses (see
\cref{sec:discussion} for a discussion of this assumption).

\section{Protocol Description} \label{sec:protocol}
We now describe each building block of the protocol in detail, beginning
with the keys and identifiers used by each node.

\subsection{Keys and Identifiers}
Each node holds a master secret key \sk from which two identifiers are
derived. First, a network identifier \netpk with its associated private key
$\netsk$, used to sign network messages. Second, a staking identifier
$\stakeID$ whose pre-image $\stakesk$ is revealed when the node's stake is
slashed.

Formally, let $\sk \xleftarrow{\$} \{0,1\}^{\lambda}$ denote a uniformly sampled
master secret, and derive
$(\netsk, \netpk) \gets \keygen(sk)$,
$\stakesk\gets \hashstake(sk)$, and
$\stakeID\gets \hashid(\stakesk)$
(See \cref{fig:derive}).

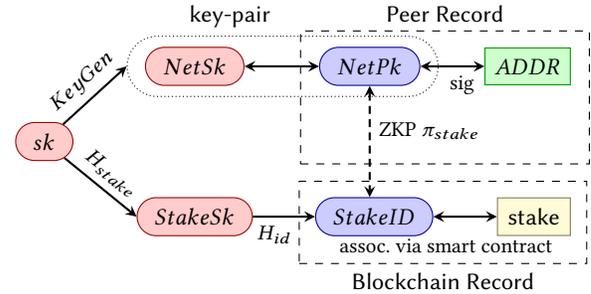
\begin{figure}[ht]
    \centering
    \begin{tikzpicture}[
      box/.style={draw, rectangle, inner xsep=4pt, inner ysep=4pt, align=center},
      rbox/.style={draw, rounded rectangle, inner xsep=6pt, inner ysep=4pt},
      every node/.style={font=\sffamily}
    ]

    \node[rbox, anchor=center,fill=red!20, draw = red!60!black] (sk) at (-2,0) {\sk};
    \node[rbox, anchor=center,fill=red!20, draw = red!60!black] (netsk) at (0,1) {\netsk};
    \node[rbox, anchor=center,fill=red!20, draw = red!60!black] (stakesk) at (0,-1) {\stakesk};
    \node[rbox, anchor=west, fill=blue!20, draw = blue!60!black] (stakeid) at (1.6,-1) {\stakeID};
    \node[rbox, anchor=east, fill=blue!20, draw = blue!60!black] (netpk) at (3,1) {\netpk};
    \node[box, anchor=east, fill=green!20, draw = green!60!black] (ip) at (5,1) {\addr};
    \node[box, anchor=east, fill=yellow!20, draw = yellow!40!black] (stake) at (5,-1) {stake};

    \node[draw, densely dotted, rbox, fit=(netsk)(netpk), inner sep=4pt, label={[xshift=-20pt]above:key-pair}](keys) {};

    \node (extra) at ($(stake)!0.5!(stakeid)+(0,-0.3)$) {}; %
    \node[draw, dashed, box, fit=(stakeid)(stake)(extra), inner sep=6pt, label=below:Blockchain Record](blockchain) {};

    \draw[-{stealth}, thick] (sk) -- (stakesk.west) node[midway,sloped, above,font=\small] {\hashstake};
    \draw[-{stealth}, thick] (stakesk) -- (stakeid) node[pos=0.35,sloped, below,font=\small ] {\hashid};

    \draw[-{stealth}, thick] (sk) -- (keys.west)  node[midway,sloped, above,font=\small ] {\keygen};
    \draw[{stealth}-{stealth}, thick] (netpk) -- (ip) node[pos=0.65,sloped, below,font=\small ] {sig};
    \draw[{stealth}-{stealth}, thick] (netsk) -- (netpk) node[midway,sloped, below,font=\small ] {};
    \draw[{stealth}-{stealth}, thick] (stakeid) -- (stake) node[pos=0.2,sloped, below=5pt,font=\small ] {assoc. via smart contract};
    \draw[{stealth}-{stealth}, thick, densely dashed]
  (netpk.south) -- (netpk.south |- stakeid.north)
  node[pos=0.4, right, font=\small](zkp){ZKP $\pi_\stake$};

    \node[draw, dashed, box, fit=(netpk)(ip)(zkp), inner xsep=7pt, inner ysep=6pt, label=above:Peer Record](network) {};

\end{tikzpicture}

\caption{Key derivation: \netpk and \stakeID are unlinkable identifiers
derived from a master secret \sk. A zero-knowledge proof $\stakeproof$
attests that \netpk is associated with \emph{some} stake allocation
without revealing which one.}
\label{fig:derive}
\end{figure}

\subsection{Stake and the Smart Contract}
\label{sec:stake-contract}
An on-chain smart contract allows nodes to associate each stake
allocation with the identifier \stakeID. Stake can be burned (slashed) by
presenting the pre-image \stakesk;
\cref{sec:too-many-reqs} shows how \stakesk is revealed when nodes misbehave.
For simplicity, we assume each stake deposit is a fixed amount.
(It is straightforward to scale protocol parameters proportionally
to stake and thus accommodate heterogeneous allocations.)

The contract publishes a vector commitment over all active (deposited,
unslashed) \stakeID{s} and provides membership proofs for individual
identifiers (pseudocode in \cref{alg:smartcontract},
\cref{appendix:pseudocode}).
Crucially, the contract performs no zero-knowledge proof verification,
as all proof checks happen off-chain among protocol participants.
The remaining on-chain costs come from Merkle-tree updates when stakes
are deposited, withdrawn, or slashed, which remain cheap
(see \cref{sec:prototype} for gas measurements).

Blockchain time is partitioned into \emph{epochs}; the active commitment
changes only at epoch boundaries, avoiding frequent on-chain reads.
Each new commitment incorporates newly added deposits and removes withdrawn
stake. Changes are frozen \stakefreeze time units before the epoch boundary to
give nodes time to query the next commitment, and withdrawals are held for
\stakewithdraw time units after the epoch ends to allow reporting of violations committed near the epoch boundary and prevent
attackers from withdrawing before their stake can be slashed.
\Cref{fig:stake_timing} depicts this timeline.

\begin{figure}[!ht]
    \centering
        \begin{tikzpicture}[x=1cm,y=1cm]
  \def\axisLen{8}        %
  \def\barHalf{0.2}      %
  \def\tickHalf{0.3}      %

  \def\openDelta{0.5}     %

  \tikzset{
    tl/axis/.style      = {line width=0.8pt,-{latex}},
    tl/div/.style       = {draw=black,line width=0.8pt},
    tl/weakdiv/.style       = {draw=gray,line width=0.5pt},
    tl/label/.style     = {font=\footnotesize,
                           align=center, execute at begin node=\strut},
    tl/ticklabel/.style = {font=\footnotesize},
    tl/regionfill/.style= {fill=blue!10, draw=blue!35},
    tl/eventarrow/.style= {black, thin, -{stealth}}
  }

  \newcommand{\Tick}[2]{%
    \draw[tl/div] (#1,-\tickHalf) -- (#1,\tickHalf);
    \node[tl/ticklabel,anchor=north] at (#1,-0.5) {#2};
  }

  \newcommand{\WeakTick}[1]{%
    \draw[tl/weakdiv] (#1,-\barHalf) -- (#1,\barHalf);
  }

  \newcommand{\EventX}[4]{%
    \draw[tl/eventarrow] (#1+#2,#3) node[tl/label,anchor=south]{#4} -- (#1,\tickHalf);
  }

\newcommand{\ShadeRange}[5][tl/regionfill]{%
  \path[#1] (#2,-\barHalf) rectangle (#3,\barHalf);
  \ifstrequal{#4}{above}{%
    \node[font=\footnotesize,anchor=south] at ({(#2+#3)/2},\barHalf) {#5};
  }{%
    \ifstrequal{#4}{below}{%
      \node[font=\footnotesize,anchor=north] at ({(#2+#3)/2},-\barHalf) {#5};
    }{%
      \node[font=\footnotesize,anchor=#4] at ({(#2+#3)/2},0) {#5};
    }%
  }%
}

  \newcommand{\ShadeOpen}[5][tl/regionfill]{%
    \path[#1] (#2,-\barHalf) rectangle (#3-\openDelta,\barHalf);
    \path[#1] (#3-\openDelta,-\barHalf) -- (#3-\openDelta,\barHalf) -- (#3,0) -- cycle;
    \ifstrequal{#4}{above}{%
    \node[font=\footnotesize,anchor=south] at ({(#2+#3)/2},\barHalf) {#5};
  }{%
    \ifstrequal{#4}{below}{%
      \node[font=\footnotesize,anchor=north] at ({(#2+#3)/2},-\barHalf) {#5};
    }{%
      \node[font=\footnotesize,anchor=#4] at ({(#2+#3)/2},0) {#5};
    }%
  }
}

\newcommand{\ShadeOpenMinusInfRange}[5][tl/regionfill]{%
  \path[#1, even odd rule]
    (#2,-\barHalf) rectangle (#3,\barHalf)
    (#2,-\barHalf) -- (#2,\barHalf) -- (#2+\openDelta,0) -- cycle;
      \ifstrequal{#4}{above}{%
    \node[font=\footnotesize,anchor=south] at ({(#2+#3)/2},\barHalf) {#5};
  }{%
    \ifstrequal{#4}{below}{%
      \node[font=\footnotesize,anchor=north] at ({(#2+#3)/2},-\barHalf) {#5};
    }{%
      \node[font=\footnotesize,anchor=#4] at ({(#2+#3)/2},0) {#5};
    }%
  }
}

  \ShadeOpenMinusInfRange[fill=blue!30] {0}{2}{below}{\begin{tabular}{c} Deposit \&\\ withdrawal requests\\  are accepted \end{tabular}}
  \ShadeRange[fill=orange!40]{2}{4}{above}{\begin{tabular}{c}All stake \\ is frozen \end{tabular}}
  \ShadeRange[fill=magenta!50]{4}{6}{above}{\begin{tabular}{c} Withdrawal \\ delay \end{tabular}}

  \ShadeOpen[fill=blue!30]{6}{8}{below}{\begin{tabular}{c} withdrawals \\ released\end{tabular}}

  \Tick{4}{\begin{tabular}{c} Epoch {i} starts: \\ New stake activated, \\ Withdrawn stake deactivated\end{tabular}}
  \WeakTick{2}
  \WeakTick{6}

  \EventX{2}{0}{1.2}{$t_i-\stakefreeze$}
  \EventX{6}{0}{1.2}{$t_i+\stakewithdraw$}

  \draw[tl/axis] (0.5,0) -- (\axisLen,0);
  \node [font=\small, anchor=north] at (\axisLen, 0){t};
\end{tikzpicture}
    \caption{Stake deposit and withdrawal timeline. Deposits and withdrawals
    are frozen before the epoch boundary so nodes can obtain the next
    commitment. Withdrawals are further delayed to allow slashing for
    late-epoch misbehavior.}
    \label{fig:stake_timing}
\end{figure}
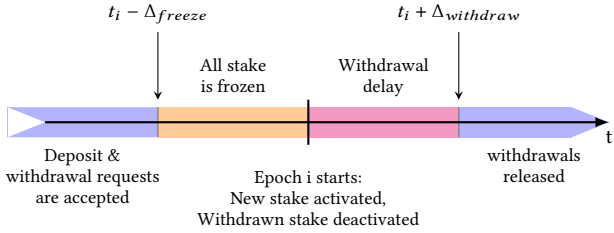

\subsection{Network Records, Peer Records, and \netpk}
Each node is identified on the network by its public key $\netpk$ (with
corresponding private key $\netsk$). A node generates and signs a \netrecord
that attests to both its stake (via a ZKP) and its network address, binding
them to \netpk:
$$ \netrecord = \langle \netpk, \stakeroot, \stakeproof, \langle \addr,
\timestamp\rangle_\sig \rangle$$
Here \stakeroot is the current epoch's vector commitment
(\cref{sec:stake-contract}), \stakeproof is a ZKP proving membership in this
commitment, and $\timestamp$ marks when the record was signed.

Upon receiving a \netrecord, a node verifies the signature, the ZKP
$\stakeproof$, and that $\stakeroot$ is a recent epoch commitment
(\alglines{alg:receive_new_record}{line:rec-verify-stake}{line:rec-stakeroot}).
Newer records supersede older ones, and records older than \recordexp are
rejected on arrival or evicted from storage
(\algline{alg:receive_new_record}{line:rec-merge};
\algline{alg:handle_response}{line:handle-evict-expired}).
To bound table sizes, each table is capped at $(1+\epsilon)\cdot s\sqrt{n}$
entries; when this cap is exceeded, the record with the highest $\PRNG$ score
(i.e., the one least likely to belong in the node's slice) is evicted first.

A \peerrecord augments the \netrecord with per-round commitment records
used to detect excessive requests (see \cref{sec:too-many-reqs}):
$$\peerrecord = \langle \netrecord, [(\commitrecord,\roundnum)_i]_{i=1}^m
\rangle $$

Each node maintains two tables: \ptable and \otable that are collections of
\peerrecord{s}.

\subsection{The Heartbeat: Gossiping about Peer Records}
Peer-record gossip proceeds in rounds (with many rounds per epoch).
During each round, a node sends a heartbeat request to every peer in its
\ptable and receives records that populate both \ptable and \otable.

To prevent attackers from inflating their representation in responses,
record selection is determined by the \emph{requester}'s seeds
$(\nonce,\nonceoverlay)$: a record with identifier $\netpk_j$ is included
iff
\begin{equation*}
  \PRNG[\nonce][\netpk_j] < \tfrac{s}{\sqrt{n}}
  \;\;\text{or}\;\;
  \PRNG[\nonceoverlay][\netpk_j] < \tfrac{s}{\sqrt{n}},
\end{equation*}
where $\PRNG$ maps each (seed, identifier) pair uniformly to $[0,1)$.
This yields a pseudorandom slice of approximately $s\sqrt{n}$ records per
seed. To reduce communication overhead, the requester includes these seeds
in its request so the \emph{responder} evaluates the $\PRNG$ locally and
returns only the selected records.

\label{rem:tee-privacy}
In practice, the response computation
(\alglines{alg:send_response}{line:filter-start}{line:filter-end})
runs inside a TEE on the responder side, so the requester's seeds remain hidden
and the response is padded to a fixed length.
Without TEEs, seed privacy can still be achieved by having the responder
transmit its full table so the requester filters locally, at the cost of
higher communication (\cref{appendix:pseudocode}).

\smallskip\noindent\textbf{Sending a heartbeat request.}
The heartbeat proceeds in four steps
(full pseudocode in \cref{alg:heartbeat,alg:send_response},
\cref{appendix:pseudocode}):
\begin{enumerate}[leftmargin=*]
  \item \emph{Build identity.} The node creates a fresh \netrecord
  containing its signed network address and a ZKP $\stakeproof$ (relation
  \Rstake) that certifies $\netpk$ is backed by some stake in the current
  epoch commitment $\stakeroot$, without revealing which stake
  (\alglines{alg:heartbeat}{line:hb-stakeproof}{line:hb-netrecord}).

  \item \emph{Sample seeds.} The node draws fresh random nonces
  $(\nonce,\nonceoverlay)$---one for the gossip slice, one for the overlay
  slice---and signs them together with the current round number
  (\algline{alg:heartbeat}{line:hb-seeds}).

  \item \emph{Commit to recipients.} The node computes a vector commitment
  $\reqcommit$ over the $s\sqrt{n}$ peers it intends to contact this round
  (\algline{alg:heartbeat}{line:hb-commit}). It also derives a slashing share
  $\slashshare$ that binds $\reqcommit$ to its secret key
  (\cref{sec:too-many-reqs}), along with a ZKP of well-formedness. These are
  bundled into a \commitrecord attached to every request.

  \item \emph{Send requests.} For each peer at position $\ind$ in $\ptable$,
  the node sends a $\request$ containing: the signed nonces, the
  \commitrecord together with an opening proof for position $\ind$, and the
  node's \netrecord
  (\alglines{alg:heartbeat}{line:hb-send-loop}{line:hb-send}).
\end{enumerate}

\smallskip\noindent\textbf{Responding to a request.}
Upon receiving a request (\cref{alg:send_response}), the responder:
(i)~validates and stores the sender's \netrecord;
(ii)~verifies freshness (current round), nonce signature,
the vector commitment opening---confirming that the responder appears in the
requester's declared recipient set (\cref{sec:too-many-reqs})---and the ZK
share proof for~\Rshare;
(iii)~enforces per-sender rate limiting; and
(iv)~evaluates the requester's seeds $(\nonce,\nonceoverlay)$ against
each record in its $\ptable$, returning those whose $\PRNG$ score falls
below $s/\sqrt{n}$.
ZK relations appear in \cref{appendix:primitives}; pseudocode in
\cref{appendix:pseudocode}.

Since each node holds the private key corresponding to its \netpk, all
messages between peers can be end-to-end encrypted.

\paragraph{Bootstrapping.}
The heartbeat also serves as a bootstrap mechanism. A new node adds one
or more bootstrap contacts to its peer table and begins participating in
heartbeat rounds, thereby populating its tables and propagating its own
record.
Concretely, whenever a node sends a request, the responder may add the
requester's record to its own table
(\algline{alg:send_response}{line:resp-recv}), so a new node's record
can later diffuse through the network via subsequent requests to the responder.
The mean-field analysis confirms that even a node starting with zero visibility
spreads its record exponentially fast (\cref{sec:mf-visibility}); we validate
this with bootstrap simulations in \cref{sec:eval}.

\paragraph{Epoch transitions.}
Once the \stakefreeze freeze period has elapsed and the new commitment is
finalized on-chain, each node executes the epoch transition
(\cref{alg:epoch_transition} in \cref{appendix:pseudocode}).
A protocol parameter $d \geq 1$ controls how many past epoch commitments
remain accepted: the set $\mathit{AccComs}$ always contains the $d$ most
recent commitments. Setting $d=1$ requires all records to reference the
current epoch; larger values provide a grace period for peers to regenerate
their proofs. During the transition each node also regenerates its own
$\stakeproof$ against the new commitment so that subsequent \netrecord{s}
reference $\stakeroot_{\mathit{new}}$.

\subsection{Detecting Eclipse attacks}
\label{sec:eclipse-detect}

\protocolname is designed so that a node bootstrapping from even a single
connection can \emph{detect} when it is likely eclipsed, whether because that
connection is malicious or the network is otherwise partitioned.

\paragraph{Why eclipse bias is detectable.}
For any responder, the set of records that should be returned is
\emph{determined} by the requester-chosen seeds $(\nonce,\nonceoverlay)$: a
record with identifier $\netpk_j$ is included iff
$\PRNG[\nonce][\netpk_j] < s/\sqrt{n}$ or
$\PRNG[\nonceoverlay][\netpk_j] < s/\sqrt{n}$. A malicious responder therefore
cannot \emph{inflate} adversarial representation---records outside the PRNG
threshold are discarded. The only remaining avenue is \emph{suppression}:
withholding honest records.
This yields a statistical signal: each node expects to collect
$\Theta(s\sqrt{n})$ distinct valid records per round, concentrated around
the mean by standard Chernoff bounds. Under an eclipse the count stagnates
well below this expectation, as the node is unable to retrieve large portions
of its slice.

\paragraph{A practical test.}
Let $U_t$ denote the number of \emph{new} distinct $\netpk$ values added to
$\ptable^t$ in round $t$. If $|U_t|< \theta$, the node
warns the operator that it is not well connected.
The threshold $\theta$ must balance false positives against
detection sensitivity; \cref{sec:overlay-construction,sec:detecting-partitions}
provide precise guidance.

\paragraph{Global guarantees.}
If an adversary partitions the network, nodes on the smaller side trigger the
flag. Supplying enough records to suppress the flag creates cross-partition
connections that undermine the partition.
\Cref{sec:overlay-resistance} formalizes this.

\subsection{Overlay Construction from \otable}
\label{sec:overlay-construction}
Once \otable has been populated via the heartbeat mechanism, each node
independently includes each peer in its \otable as an overlay neighbor with
probability $\pconn$, yielding an expected degree of
$\pconn \cdot s\sqrt{n}$. Setting $\pconn = c \cdot \log n / (s\sqrt{n})$
for a suitable constant $c > 1$ gives $\Theta(\log n)$ connections per node,
which suffices for connectivity with high probability by standard random-graph
results, provided the adversarial fraction in \otable remains below a constant
threshold (ensured by the mean-field analysis of \cref{sec:mean-field}).
\Cref{thm:overlay} formalizes this guarantee.

Overlay connections are drawn exclusively from \otable, selected using
the private seed \nonceoverlay, while peer-discovery responses are served from
\ptable using \nonce
(\algline{alg:receive_new_record}{line:rec-insert-otable};
\alglines{alg:send_response}{line:filter-start}{line:filter-end}).
Because responding to requests never reveals
entries from \otable, an adversary observing or participating in the gossip
protocol learns nothing about which peers a node has chosen as overlay
neighbors.

\subsection{Detecting Too Many Requests and Slashing}
\label{sec:too-many-reqs}

A central DoS vector in \protocolname is flooding the network with heartbeat
requests, since responding requires running the PRNG over the full table and
sending a response typically larger than the request. Even if each responder
rate-limits per sender
(\algline{alg:send_response}{line:resp-rate-limit}), an attacker can spread
requests across many responders. We therefore enforce a \emph{global} per-round
request quota and provide a compact \emph{cryptographic proof} that a requester
exceeded it, enabling on-chain slashing. If slashing is not desired, the same
mechanism can simply drop excessive requests without serving them.

\paragraph{Per-round quota via unique batch commitments.}
Each request carries a
$\commitrecord=\langle \reqcommit, \slashshare, \shareproof\rangle$
where $\reqcommit$ is a vector commitment to the intended recipients for
that round (\algline{alg:heartbeat}{line:hb-commit}), together with an
opening $\langle \ind,\indproof\rangle$ proving that the responder's $\netpk$
appears at position $\ind$ in the committed list.
Responders verify this opening,
check that the request references the current round, and verify the ZK share
proof $\shareproof$
(\alglines{alg:send_response}{line:resp-verify-opening}{line:resp-verify-share}),
thereby accepting only requests that are \emph{consistent} with the
requester's declared batch and cryptographically bound to a single commitment
per round.

We now enforce the policy:
\begin{quote}
\emph{A staked identity $\netpk$ may issue requests under at most one batch commitment $\reqcommit$ per round.}
\end{quote}
This is exactly what \cref{alg:send_response} enforces: peers store
$[(\commitrecord_i,\allowbreak\roundnum_i)]$ inside each \peerrecord,
and upon observing two distinct commitments for the same $\netpk$
and $\roundnum$ they construct a \slashproof
(\alglines{alg:receive_new_record}{line:rec-dup-detect}{line:rec-submit-slash}).

\paragraph{Why two commitments imply slashability.}
Each share is an affine function of the commitment (Definition~\Rshare):
\[
  \slashshare = \hashshare(\sk,\roundnum) \cdot \reqcommit + \stakesk.
\]
Two valid shares under distinct commitments $\reqcommit_1\neq\reqcommit_2$ in
the same round yield two linear equations in two unknowns
($\hashshare(\sk,\roundnum)$ and $\stakesk$), from which any observer can
recover $\stakesk$ and slash the offender's stake on-chain
(\algline{alg:smartcontract}{line:sc-slash}).
The full derivation appears in \cref{appendix:pseudocode}.

\paragraph{Constructing and propagating \slashproof.}
The public slashing evidence is
\begin{multline*}
\slashproof = \langle \netpk,\; \stakeroot,\; \roundnum,\\
\commitrecord_1,\; \commitrecord_2 \rangle,
\end{multline*}
where the two \commitrecord{s} carry distinct $\reqcommit$ values for
the same $\netpk$ and $\roundnum$.
Any peer holding both records verifies
$\text{ZKVerify}_{\Rshare}$ for each and confirms
$\reqcommit_1\neq \reqcommit_2$
(\algline{alg:verify_slash}{line:vs-distinct}).
Verified \slashproof{s} are gossiped in $\response$ and added to a
local \denylist; deny-listed identities are deprioritized or ignored until the
epoch commitment leaves $\mathit{AccComs}$
(\algline{alg:send_response}{line:resp-denylist};
\algline{alg:epoch_transition}{line:et-purge-deny}).

Thus, a requester cannot contact more than $s\sqrt{n}$ distinct peers per round
without creating an additional commitment, making it slashable---achieving a
\emph{globally enforceable} request quota without resorting to a trusted aggregator or the collection of many Shamir shares across responders.

\section{Analysis of Network Health} \label{sec:mean-field}
We analyze the health of the network by studying (a) \emph{table quality}: how
well the tables of nodes cover the honest portion of the network, and (b)
\emph{node visibility}: how well peer records propagate to other honest nodes.
We present both analyses in \cref{sec:mf-quality} and
\cref{sec:mf-visibility}, respectively, and summarize the results in
\cref{fig:mean-field-plots}.
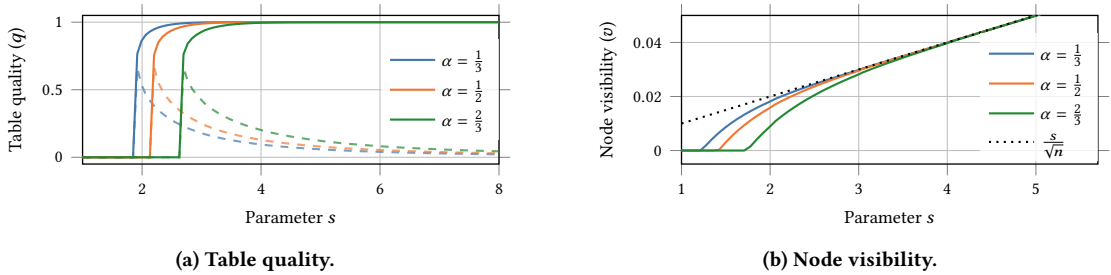
\begin{figure*}
    \centering
    \begin{subfigure}[b]{\awplotiiwidth}
        \centering
        \begin{tikzpicture}
            \begin{axis}[
                awplot,
                xlabel={Parameter $s$},
                ylabel={Table quality ($q$)},
                xmin=1, xmax=8,
                ymin=-0.05, ymax=1.05,
                width=\awplotiiwidth,
                height=\awplotheight,
                legend style={
                    at={(1.0,0.5)},
                    anchor=east,
                    draw=none,
                    fill=white,
                    fill opacity=0.8,
                    text opacity=1,
                    cells={anchor=west},
                    font=\scriptsize,
                },
            ]

                \addplot[color=plotA, thick, mark=none]
                    table[x=s, y=q_high] {figures/mean-field-data/alpha_1_3.txt};
                \label{plt:qh-a13}
                \addlegendentry{$\alpha=\frac{1}{3}$}
                \addplot[color=plotA, thick, dashed, opacity=0.7, forget plot]
                    table[x=s, y=q_thresh] {figures/mean-field-data/alpha_1_3.txt};
                \label{plt:qt-a13}

                \addplot[color=plotB, thick, mark=none]
                    table[x=s, y=q_high] {figures/mean-field-data/alpha_1_2.txt};
                \label{plt:qh-a12}
                \addlegendentry{$\alpha=\frac{1}{2}$}
                \addplot[color=plotB, thick, dashed, opacity=0.7, forget plot]
                    table[x=s, y=q_thresh] {figures/mean-field-data/alpha_1_2.txt};
                \label{plt:qt-a12}

                \addplot[color=plotC, thick, mark=none]
                    table[x=s, y=q_high] {figures/mean-field-data/alpha_2_3.txt};
                \label{plt:qh-a23}
                \addlegendentry{$\alpha=\frac{2}{3}$}
                \addplot[color=plotC, thick, dashed, opacity=0.7, forget plot]
                    table[x=s, y=q_thresh] {figures/mean-field-data/alpha_2_3.txt};
                \label{plt:qt-a23}
            \end{axis}
        \end{tikzpicture}
        \caption{Table quality.}
        \label{fig:mf-quality-plot}
    \end{subfigure}
    \hspace{2em}
    \begin{subfigure}[b]{\awplotiiwidth}
        \centering
        \begin{tikzpicture}
            \begin{axis}[
                awplot,
                xlabel={Parameter $s$},
                ylabel={Node visibility ($v$)},
                xmin=1, xmax=5.7,
                ymin=-0.005, ymax=0.05,
                width=\awplotiiwidth,
                height=\awplotheight,
                legend style={
                    at={(1.0,0.45)},
                    anchor=east,
                    draw=none,
                    fill=white,
                    fill opacity=0.8,
                    text opacity=1,
                    cells={anchor=west},
                    font=\scriptsize,
                },
            ]
                \addplot[color=plotA, thick, mark=none]
                    table[x=s, y=v_high] {figures/mean-field-data/alpha_1_3.txt};
                \label{plt:vh-a13}
                \addlegendentry{$\alpha=\frac{1}{3}$}

                \addplot[color=plotB, thick, mark=none]
                    table[x=s, y=v_high] {figures/mean-field-data/alpha_1_2.txt};
                \label{plt:vh-a12}
                \addlegendentry{$\alpha=\frac{1}{2}$}

                \addplot[color=plotC, thick, mark=none]
                    table[x=s, y=v_high] {figures/mean-field-data/alpha_2_3.txt};
                \label{plt:vh-a23}
                \addlegendentry{$\alpha=\frac{2}{3}$}

                \addplot[black, thick, dotted]
                    table[x=s, y=s_over_sqrt_n] {figures/mean-field-data/s_over_sqrt_n.txt};
                \label{plt:s-over-sqrt-n}
                \addlegendentry{$\frac{s}{\sqrt{n}}$}
            \end{axis}
        \end{tikzpicture}
        \caption{Node visibility.}
        \label{fig:mf-visibility-plot}
    \end{subfigure}
    \caption{Mean-field behavior with $n=10{,}000$ and table size $s\sqrt{n}$ for fixed
    adversary stake $\alpha\in\{1/3,1/2,2/3\}$.
    (a)~\emph{Table quality} $q$ vs.\ $s$: solid curves show the stable equilibrium
    $q_{\mathrm{high}}$; dashed curves show the unstable threshold
    $q_{\mathrm{thresh}}$ that separates initial conditions converging to the
    healthy equilibrium from those collapsing toward $q\approx 0$.
    (b)~\emph{Visibility} $v$ vs.\ $s$: solid curves show the stable visibility
    $v_{\mathrm{high}}$; the dotted line plots $\frac{s}{\sqrt{n}}$, the expected
    visibility in an honest network. We suggest $s=4$ based on these results.}
    \label{fig:mean-field-plots}
\end{figure*}

We use a mean-field approach: we track a representative honest node and
derive a deterministic recurrence describing how each quantity evolves
round by round, treating nodes as independent in the large-network limit.
This approximation is justified when the network is large
enough that individual correlations average out. We validate these predictions against simulations in \cref{sec:eval}.
The recurrence's \emph{stable} fixed points are the system's equilibria and \emph{unstable} fixed points mark basin-of-attraction boundaries.

The main challenges in maintaining table quality and node visibility are that
adversarial or offline nodes may not respond to requests with fresh peer
records, leading to poor sampling of the network. We analyze network health
while assuming that an adversary controlling an $\alpha$ fraction of the stake
is \emph{silent}. A silent adversary represents the worst-case scenario for
network health; a more active adversary that relays peer records would only
improve network health.

\subsection{Table Quality} \label{sec:mf-quality}
We introduce the notion of \emph{table quality} to measure how well a node's
peer table covers the slice that the node is supposed to retrieve
from the honest portion of the network. Recall that a node's nonce selects each
node with probability $\frac{s}{\sqrt{n}}$ for its table
(\algline{alg:receive_new_record}{line:rec-insert-ptable}), hence targeting $s\sqrt{n} (1 -
\alpha)$ honest in expectation.
\begin{definition}[table quality]
		The \textbf{table quality} $\quality_r$ in round $r$ represents the
    honest fraction of the slice that is present in a node's
    peer table at the end of round $r$.
\end{definition}

\paragraph{Mean field analysis of table quality.}
We show that the table quality remains high in expectation even in the presence
of a silent adversary. Let the adversary control an $\alpha$ fraction of the
total stake and be \emph{silent}. We derive the difference equation by analyzing
the probability that a specific honest node $j$ in $i$'s slice is \emph{not}
in $i$'s peer table after one heartbeat
(\cref{alg:heartbeat}).
We consider two ways that $i$ may
fail to receive $j$'s record:
\begin{itemize}
    \item \emph{Requests fail to retrieve $j$'s record:} Node $i$ sends
    $s\sqrt{n}(1-\alpha)\quality_r$ requests to honest peers. Each peer holds
    $j$'s record with probability $\frac{s\quality_r}{\sqrt{n}}$ (it samples $j$
    with probability $\frac{s}{\sqrt{n}}$, of which a $\quality_r$ fraction are
    obtained). The probability none of the replies contain $j$ is $(1 -
    \frac{s\quality_r}{\sqrt{n}})^{s\sqrt{n}\quality_r(1-\alpha)}$.
    \item \emph{$j$ does not request from $i$:} With probability
    $(1 - \frac{s\quality_r}{\sqrt{n}})$, node $j$ fails to sample $i$ for its
    requests, accounting for the ``$+1$'' exponent term.
\end{itemize}
Combining these two failure modes, the probability that $j$'s record is
\emph{not} in $i$'s table after one heartbeat is the product of the two
independent failure probabilities, yielding the following difference equation:
\[
		\quality_{r+1}
		= 1 -
		\bigg(
		\bigg( 1 - \frac{s}{\sqrt{n}}\quality_r \bigg)^{s\sqrt{n}\quality_r(1-\alpha) + 1} \bigg)
\]

When $s > 1$ the difference equation has two positive fixed points: a
\emph{stable} high equilibrium $\quality_{\mathrm{high}} \approx 1$ (solid
curves in \cref{fig:mf-quality-plot}) and an \emph{unstable} threshold
$\quality_{\mathrm{thresh}}$ (dashed curves). If the initial table quality
exceeds $\quality_{\mathrm{thresh}}$, the system converges to
$\quality_{\mathrm{high}}$; otherwise it degrades toward $\quality \approx 0$.
\cref{fig:mf-quality-plot} shows results for various $s$ and $\alpha$.

\subsection{Node Visibility} \label{sec:mf-visibility}
We now analyze the \emph{node visibility} property, which measures how well a
given honest node's peer record propagates to other honest nodes in the network,
in the presence of a silent adversary.

\begin{definition}[node visibility]
    The \textbf{node visibility} $\visibility_r$ of a node at round $r$
    represents the fraction of honest nodes that have the node's peer record
    in their peer table at the end of round $r$.
\end{definition}

\paragraph{Mean field analysis of node visibility.}
Node visibility settles to $s/\sqrt{n}$, the expected fraction of honest nodes
holding a given record. For node $j$ to have $i$'s record: (1)~$i$ must pass
$j$'s PRNG threshold (probability $s/\sqrt{n}$), and (2)~at least one of
$j$'s $s\sqrt{n}(1-\alpha)$ honest respondents must hold $i$'s record (each
does so with probability $\visibility_r$). This yields:
\[
    \visibility_{r+1} = \frac{s}{\sqrt{n}}
    \left(1 - (1 - \visibility_r)^{s\sqrt{n}(1-\alpha)}\right)
\]
Note that we ignore \emph{injection}---the mechanism by which a responder adds
the requester's record to its own table
(\algline{alg:send_response}{line:resp-recv})---as it only helps
visibility.\footnote{A refined mean-field analysis shows that the injection term
shifts the equation intercept such that at $\visibility_0 = 0$, we still get
growth. This is essential for bootstrapping: a new node with zero visibility can
still spread its record because every request it sends gives the responder a
chance to store it.}

\paragraph{$R_0$ threshold.}
The basic reproduction number $R_0 = s^2(1-\alpha)$ governs whether a new
node's record spreads.\footnote{Borrowing from epidemiology: $R_0 > 1$ means
each node spreads the record to more than one peer on average.}
Suppose $\visibility_r$ is small. Using the approximation
$1-(1-\visibility)^x \approx x\visibility$ for small $\visibility$, the
difference equation linearizes to:
\[
    \visibility_{r+1}
    \approx \frac{s}{\sqrt{n}} \cdot s\sqrt{n}(1-\alpha) \visibility_r
    = s^2 (1-\alpha)\visibility_r.
\]
Growth occurs when $R_0 = s^2(1-\alpha) > 1$, i.e., $s > 1/\sqrt{1-\alpha}$.
Above this threshold, visibility converges to a stable fixed point close to
$s/\sqrt{n}$.
These guarantees are preconditions for \cref{thm:overlay}.
In particular, a network at the high equilibrium is $\gamma$-healthy
(\cref{def:network_states}) with $\gamma \approx \quality_{\mathrm{high}}$,
bridging the mean-field predictions to the security analysis of
\cref{sec:security}.
\cref{fig:mf-visibility-plot} shows the results for different $s$ and $\alpha$.

\section{Security and Privacy} \label{sec:security}
We now state the main security and privacy theorems of \protocolname.
\cref{tab:security-summary} summarizes the results, and deferred
proofs can be found in \cref{appendix:proofs}.
\begin{table}[t]
	\centering
	\small
	\begin{tabular}{@{} l c @{}}
		\toprule
		\textit{Security \& Privacy} & \textbf{Result(s)} \\
		\midrule
		Partition Detection & \Cref{lem:attack-soundness,lem:attack-completeness} \\
		Overlay Formation & \Cref{thm:overlay} \\
		Slashing Spam &
		\Cref{lem:slashing-soundness,lem:stake-secret-recovery,lem:slashing-completeness}
		\\
		Stake Anonymity & \Cref{lem:stake-anon} \\
		Connection Privacy & \Cref{lem:conn-privacy} \\
		\bottomrule
	\end{tabular}
	\caption{Summary of security and privacy properties.}
	\label{tab:security-summary}
\end{table}

\subsection{Detecting Partitions} \label{sec:detecting-partitions}
As a warmup, we show that the detection mechanism cleanly separates healthy
from compromised network states. In this subsection we do not model an active
adversary; we simply assume the network is either healthy or partitioned and
show that the flag reliably distinguishes the two cases.

Recall from the mean-field analysis (\cref{sec:mean-field}) that the protocol's
peer discovery mechanism drives the network toward an equilibrium where honest
nodes collect peer records from most of their slice. We formalize this as follows.

\begin{definition}[$\gamma$-healthy network] \label{def:network_states}
For $0 < \gamma < 1$, we say the network is \textbf{$\gamma$-healthy} if,
within a single heartbeat round in which all nodes behave honestly, every
honest node~$i$ receives peer records from at least a $\gamma$-fraction of all nodes in its slice.
\end{definition}

\paragraph{Parameter overview.}
Our security analysis uses three key fractions derived from the adversary
stake~$\alpha$:
\begin{itemize}
  \item The \emph{critical partition fraction} $\varphi = (1+\alpha)/2$. If the
  adversary controls an $\alpha$-fraction of nodes, then in any partition of the
  honest nodes the smaller side sees at most a $\varphi$-fraction of the full
  network.
  \item The \emph{healthy reachability fraction} $\gamma \in (\varphi, 1)$, a
  lower bound on the expected fraction of the network reachable by an honest
  node during normal operation (justified by the mean-field analysis of
  \cref{sec:mean-field}).
  \item The \emph{detection threshold} $\theta \in (\varphi, \gamma)$: a node
  raises the attack-detection flag when the number of unique peer records it
  receives falls below $\theta s\sqrt{n}$.
\end{itemize}
As long as $\varphi < \theta < \gamma$, the protocol can distinguish healthy
from compromised states: in a $\gamma$-healthy network each node sees at
least a $\gamma$-fraction of peers (by definition) and will not trigger the
flag (with high probability), while in a partition the smaller side sees at
most a $\varphi$-fraction and will trigger the flag (with high probability).
We set $\theta = (5+\alpha)/6$, which evenly splits the interval
$[\varphi, 1]$ (see \cref{appendix:proofs} for the derivation).
This choice is valid whenever $\gamma > (5+\alpha)/6$, which is the parameter
regime we target in practice; see \cref{sec:parameter-concrete} for concrete
parameterizations.

\subsubsection{Healthy and compromised network states}
In a $\gamma$-healthy network, the flag should not fire:
\newcommand{\stmtAttackSoundness}{%
  Assume the network is $\gamma$-healthy and the adversary is not attacking
  (i.e.\ the entire network behaves honestly) in round $r$.
  If $\theta < \gamma$, then every honest node raises the attack-detection
  flag with probability at most
  $\exp\left(-\frac{s\sqrt{n}(\gamma - \theta)^2}{2\gamma}\right)$.}
\begin{lemma}[attack soundness] \label{lem:attack-soundness}
  \stmtAttackSoundness
\end{lemma}

A \emph{partitioned} network is the complementary case: the honest nodes are
split into two disjoint eclipsed sets (\cref{def:eclipsed}) $A$ and $B$, each
isolated from the other (though both may still see adversarial nodes). Such a
network is $\gamma$-healthy only for a small value of $\gamma$, since each side
can discover only a limited fraction of the network.

We first establish that the smaller side of any partition is bounded in size,
then show the flag fires.

\newcommand{\stmtPartitionSize}{%
  Suppose the network is \emph{partitioned} into honest sets $(A,B)$ and let
  $V_A = A \cup \Advset$ and $V_B = B \cup \Advset$ denote the full views of
  each side (honest nodes plus adversarial nodes visible to that side). If,
  without loss of generality, $|V_A| \le |V_B|$, then $|V_A| \leq \varphi n$
  where $\varphi = (1+\alpha)/2$ is the critical fraction.}
\begin{lemma}[global partitioning $\Rightarrow$ small view]
  \label{lem:partition-size}
  \stmtPartitionSize
\end{lemma}

Since the smaller side sees at most a $\varphi$-fraction of the network, and
$\theta > \varphi$, the flag fires with high probability:
\newcommand{\stmtAttackCompleteness}{%
  Assume the network is partitioned in round $r$.
  If $\theta > \varphi$, then every honest node in the smaller
  partition raises the attack-detection flag with probability at least
  $1 - \exp\left(-\frac{s\sqrt{n}(\theta - \varphi)^2}{\theta + \varphi}\right)$.}
\begin{lemma}[attack completeness] \label{lem:attack-completeness}
  \stmtAttackCompleteness
\end{lemma}

Together, \cref{lem:attack-soundness,lem:attack-completeness} show that when
$\varphi < \theta < \gamma$, the detection mechanism separates healthy from
compromised states. We illustrate with $\alpha = 0.25$, $s = 4$, $\gamma = 0.90$.
We choose $\theta$ to balance the two error rates. Setting the Chernoff
exponents $(\gamma-\theta)^2/(2\gamma)$ and $(\theta-\varphi)^2/(\theta+\varphi)$
approximately equal gives $\theta \approx 0.75$.
\begin{table}[H]
    \centering
    \small
    \begin{tabular}{@{} r c c c @{}}
        \toprule
        $n$ & $s\sqrt{n}$ & False positive & False negative \\
        \midrule
        $10^4$ & $4.00\times 10^2$ & $5.13\times 10^{-4}$ & $7.56\times 10^{-4}$ \\
        $10^5$ & $1.26\times 10^3$ & $2.82\times 10^{-9}$ & $2.05\times 10^{-8}$ \\
        $10^6$ & $4.00\times 10^3$ & $3.28\times 10^{-25}$ & $1.18\times 10^{-22}$ \\
        \bottomrule
    \end{tabular}
    \caption{Exact binomial error probabilities for $\alpha=0.25$, $s=4$,
    $\gamma=0.90$, $\varphi=0.625$, and $\theta=0.75$.}
    \label{tab:concrete-attack-bounds}
\end{table}

\subsection{Overlay Resistance and Eclipse Detection} \label{sec:overlay-resistance}
The previous subsection showed that a partitioned network is reliably detected.
We now turn to the harder question: can a powerful adversary actually
\emph{create} such a partition in the overlay? We analyze resistance to
partitioning under an extremely strong adversary, which we formalize next.
Recall that each node constructs its overlay by independently including
each peer in $\otable$ with probability $\pconn$, yielding an expected
degree of $\pconn s \sqrt{n}$.

\begin{definition}[omniscient overlay adversary] \label{def:omniscient-overlay}
An \emph{omniscient overlay adversary} controls an $\alpha$-fraction of nodes
and has the following capabilities during overlay construction:
\begin{enumerate}[(i)]
  \item it knows every honest node's private nonce $\nonceoverlay$ and therefore
    knows the exact slice each node selects;
  \item it can completely determine the contents of every honest node's
    $\otable$, choosing which honest records to deliver and which to withhold;
  \item it optimizes its strategy against the full network topology.
\end{enumerate}
The only constraint is that the adversary must respect each node's nonce
$\nonceoverlay$: it cannot alter the pseudorandom selection that determines
which peers fall in a node's slice.
\end{definition}

Despite this worst-case power, the
attack-detection mechanism ensures the following global guarantee:

\begin{theorem*}[informal, see \cref{thm:overlay}] \label{thm:overlay-informal}
  In the presence of an omniscient overlay adversary (\cref{def:omniscient-overlay}),
  for appropriately chosen threshold $\theta > (1+\alpha)/2$ and selection
  probability $\pconn$, with high probability, one of the following holds:
  \begin{enumerate}[(a)]
    \item the overlay forms a single connected component among all honest nodes, or
    \item the overlay is disconnected, and every smaller component has at least
    a $(1-\delta)$-fraction of its honest nodes raising the attack-detection
    flag.
  \end{enumerate}
\end{theorem*}
\colorlet{AEdge}{blue!60!black}
\colorlet{AFill}{blue!12}
\colorlet{BEdge}{teal!60!black}
\colorlet{BFill}{teal!12}
\colorlet{AdvEdge}{orange!70!black}
\colorlet{AdvFill}{orange!12}
\colorlet{NodeEdge}{black!55}
\colorlet{FlagRed}{red!85!black}
\colorlet{OuterEdge}{black!45}
\colorlet{CaseBlue}{blue!75!black}
\colorlet{DeltaEdge}{red!55!black}
\colorlet{DeltaFill}{red!10}

\tikzset{
  overlay/outerbox/.style={rounded corners=1.8mm, draw=OuterEdge, line width=0.65pt},
  overlay/advbox/.style={rounded corners=1.5mm, draw=AdvEdge, fill=AdvFill, line width=0.70pt},
  overlay/vnode/.style={circle, draw=NodeEdge, fill=white, line width=0.50pt, minimum size=2.20mm, inner sep=0pt},
  overlay/anode/.style={circle, draw=AdvEdge, fill=orange!8, line width=0.50pt, minimum size=2.00mm, inner sep=0pt},
  overlay/cutline/.style={densely dashed, draw=OuterEdge, line width=0.55pt},
  overlay/localA/.style={draw=AEdge, line width=0.45pt},
  overlay/localB/.style={draw=BEdge, line width=0.45pt},
  overlay/cross/.style={->, draw=CaseBlue, line width=0.90pt,
    preaction={draw,white,line width=1.55pt,opacity=0.95},
    >= {Latex[length=1.60mm,width=1.10mm]},
    shorten <=1.05pt, shorten >=0pt},
  overlay/advlink/.style={->, draw=AdvEdge, line width=0.70pt,
    preaction={draw,white,line width=1.35pt,opacity=0.93},
    >= {Latex[length=1.35mm,width=1.00mm]},
    shorten <=1.00pt, shorten >=0pt},
  overlay/pillA/.style={font=\scriptsize\bfseries\sffamily, text=AEdge, fill=white, draw=blue!30, rounded corners=0.9mm, inner xsep=2.8pt, inner ysep=1.0pt},
  overlay/pillB/.style={font=\scriptsize\bfseries\sffamily, text=BEdge, fill=white, draw=teal!30, rounded corners=0.9mm, inner xsep=2.8pt, inner ysep=1.0pt},
  overlay/pillAdv/.style={font=\scriptsize\bfseries\sffamily, text=AdvEdge, fill=white, draw=orange!30, rounded corners=0.9mm, inner xsep=3.2pt, inner ysep=1.0pt}
}

\newcommand{\overlayflag}[1]{%
  \draw[FlagRed, line width=0.28pt]
    ($(#1.north east)+(0.010,0.0)$) -- ++(0,0.09);
  \fill[FlagRed]
    ($(#1.north east)+(0.010,0.068)$) -- ++(0.16,0.045) -- ++(-0.16,0.055) -- cycle;
}

\newcommand{\overlaypanelsetup}[1]{%
  \fill[AFill, rounded corners=1.45mm] (0.08,0.08) rectangle (1.26,1.56);
  \fill[BFill, rounded corners=1.45mm] (1.60,0.08) rectangle (3.68,1.56);
  \draw[overlay/outerbox] (0,0) rectangle (3.76,1.64);
  \draw[overlay/cutline] (1.44,0.10) -- (1.44,1.54);
  \draw[overlay/advbox] (0.30,1.92) rectangle (3.46,2.78);
  \node[overlay/pillA,anchor=west] at (0.18,1.49) {$A$};
  \node[overlay/pillB,anchor=west] at (1.72,1.49) {$B$};
  \node[overlay/pillAdv,anchor=west] at (0.48,2.58) {$\Adv$};
  \coordinate (#1a1) at (0.26,0.94);
  \coordinate (#1a2) at (0.70,1.22);
  \coordinate (#1a3) at (1.10,0.76);
  \coordinate (#1a4) at (0.36,0.36);
  \coordinate (#1a5) at (0.84,0.30);
  \coordinate (#1b1) at (1.78,1.16);
  \coordinate (#1b2) at (2.20,0.96);
  \coordinate (#1b3) at (2.60,1.18);
  \coordinate (#1b4) at (3.10,0.94);
  \coordinate (#1b5) at (1.90,0.58);
  \coordinate (#1b6) at (2.40,0.70);
  \coordinate (#1b7) at (2.86,0.54);
  \coordinate (#1b8) at (2.06,0.22);
  \coordinate (#1b9) at (2.54,0.30);
  \coordinate (#1b10) at (3.18,0.26);
  \coordinate (#1u1) at (0.82,2.35);
  \coordinate (#1u2) at (1.56,2.35);
  \coordinate (#1u3) at (2.30,2.35);
  \coordinate (#1u4) at (3.04,2.35);
}
\newcommand{\overlaypaneltopnodes}[1]{%
  \node[overlay/vnode] (#1a1v) at (#1a1) {};
  \node[overlay/vnode] (#1a2v) at (#1a2) {};
  \node[overlay/vnode] (#1a3v) at (#1a3) {};
  \node[overlay/vnode] (#1a4v) at (#1a4) {};
  \node[overlay/vnode] (#1a5v) at (#1a5) {};
  \node[overlay/vnode] (#1b1v) at (#1b1) {};
  \node[overlay/vnode] (#1b2v) at (#1b2) {};
  \node[overlay/vnode] (#1b3v) at (#1b3) {};
  \node[overlay/vnode] (#1b4v) at (#1b4) {};
  \node[overlay/vnode] (#1b5v) at (#1b5) {};
  \node[overlay/vnode] (#1b6v) at (#1b6) {};
  \node[overlay/vnode] (#1b7v) at (#1b7) {};
  \node[overlay/vnode] (#1b8v) at (#1b8) {};
  \node[overlay/vnode] (#1b9v) at (#1b9) {};
  \node[overlay/vnode] (#1b10v) at (#1b10) {};
  \node[overlay/anode] (#1u1v) at (#1u1) {};
  \node[overlay/anode] (#1u2v) at (#1u2) {};
  \node[overlay/anode] (#1u3v) at (#1u3) {};
  \node[overlay/anode] (#1u4v) at (#1u4) {};
}

\newcommand{\overlaypanelleft}{%
  \begin{tikzpicture}[x=1cm,y=1cm,scale=1.00,transform shape,font=\footnotesize\sffamily,line cap=round,line join=round]
    \overlaypanelsetup{L}
    \fill[DeltaFill, rounded corners=1.2mm] ($(La5)+(-0.22,-0.22)$) rectangle (1.24,0.52);
    \draw[DeltaEdge, rounded corners=1.2mm, line width=0.50pt, densely dashed]
      ($(La5)+(-0.22,-0.22)$) rectangle (1.24,0.52);
    \node[font=\scriptsize, text=DeltaEdge] at (1.06,0.36) {$\delta$};
    \draw[overlay/localA] (La2) -- (La1);
    \draw[overlay/localA] (La2) to[bend right=6] (La3);
    \draw[overlay/localA] (La1) -- (La4);
    \draw[overlay/localA] (La3) -- (La5);
    \draw[overlay/localA] (La4) to[bend right=6] (La5);  %
    \draw[overlay/localB] (Lb1) -- (Lb2);
    \draw[overlay/localB] (Lb2) to[bend left=6] (Lb3);
    \draw[overlay/localB] (Lb3) -- (Lb4);
    \draw[overlay/localB] (Lb2) -- (Lb6);
    \draw[overlay/localB] (Lb6) to[bend right=6] (Lb5);
    \draw[overlay/localB] (Lb6) -- (Lb7);
    \draw[overlay/localB] (Lb5) -- (Lb8);
    \draw[overlay/localB] (Lb8) to[bend right=6] (Lb9);
    \draw[overlay/localB] (Lb7) -- (Lb10);
    \draw[overlay/localB] (Lb3) to[bend right=8] (Lb7);  %
    \draw[overlay/localB] (Lb9) to[bend right=6] (Lb10); %
    \draw[overlay/localB] (Lb1) to[bend right=10] (Lb5); %
    \overlaypaneltopnodes{L}
    \draw[overlay/advlink] (La2v.north) to[out=84,in=-96] (Lu1v.south);
    \draw[overlay/advlink] (La4v.north) to[out=40,in=-148] (Lu2v.south west);
    \draw[overlay/advlink] (Lb4v.north) to[out=88,in=-44] (Lu3v.south east);
    \draw[overlay/advlink] (Lb7v.north) to[out=82,in=-80] (Lu4v.south);
    \overlayflag{La1v}
    \overlayflag{La2v}
    \overlayflag{La4v}
    \overlayflag{La3v}
  \end{tikzpicture}%
}

\newcommand{\overlaypanelright}{%
  \begin{tikzpicture}[x=1cm,y=1cm,scale=1.00,transform shape,font=\footnotesize\sffamily,line cap=round,line join=round]
    \overlaypanelsetup{R}
    \draw[overlay/localA] (Ra2) -- (Ra1);
    \draw[overlay/localA] (Ra2) to[bend right=6] (Ra3);
    \draw[overlay/localA] (Ra1) -- (Ra4);
    \draw[overlay/localA] (Ra3) -- (Ra5);
    \draw[overlay/localA] (Ra4) to[bend right=6] (Ra5);  %
    \draw[overlay/localB] (Rb1) -- (Rb2);
    \draw[overlay/localB] (Rb2) to[bend left=6] (Rb3);
    \draw[overlay/localB] (Rb3) -- (Rb4);
    \draw[overlay/localB] (Rb2) -- (Rb6);
    \draw[overlay/localB] (Rb6) to[bend right=6] (Rb5);
    \draw[overlay/localB] (Rb6) -- (Rb7);
    \draw[overlay/localB] (Rb5) -- (Rb8);
    \draw[overlay/localB] (Rb8) to[bend right=6] (Rb9);
    \draw[overlay/localB] (Rb7) -- (Rb10);
    \draw[overlay/localB] (Rb3) to[bend right=8] (Rb7);
    \draw[overlay/localB] (Rb9) to[bend right=6] (Rb10);
    \draw[overlay/localB] (Rb1) to[bend right=10] (Rb5);
    \overlaypaneltopnodes{R}
    \draw[overlay/advlink] (Ra2v.north) to[out=84,in=-96] (Ru1v.south);
    \draw[overlay/advlink] (Ra5v.north) to[out=72,in=-148] (Ru2v.south west);
    \draw[overlay/advlink] (Rb4v.north) to[out=88,in=-44] (Ru3v.south east);
    \draw[overlay/advlink] (Rb7v.north) to[out=82,in=-80] (Ru4v.south);
    \draw[overlay/cross] (Ra3v.east) to[bend right=8] (Rb9v.west);
    \draw[overlay/cross] (Ra4v.east) to[bend right=10] (Rb4v.west);
    \draw[overlay/cross] (Ra2v.east) to[bend left=12] (Rb2v.west);
  \end{tikzpicture}%
}

\begin{figure}[t]
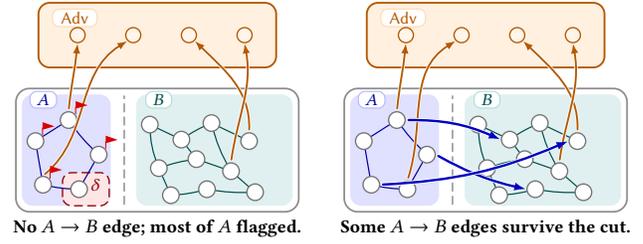

  \centering
  \captionsetup[subfigure]{labelformat=empty, justification=centering, font=footnotesize, skip=2pt}
  \begin{subfigure}[t]{0.485\linewidth}
    \centering
    \overlaypanelleft
    \caption{No \(A\to B\) edge; most of \(A\) flagged.}
  \end{subfigure}\hfill
  \begin{subfigure}[t]{0.485\linewidth}
    \centering
    \overlaypanelright
    \caption{Some \(A\to B\) edges survive the cut.}
  \end{subfigure}
  \caption{Overlay resistance across an honest
  cut \((A,B)\), with \(A\) the smaller side. Either the adversary removes all
  cross-cut overlay edges and most nodes in \(A\) raise the attack-detection
  flag, or some sampled \(A\to B\) edges survive, so the overlay is not
  partitioned across the cut.}
  \label{fig:overlay}
\end{figure}

The proof (see \cref{fig:overlay}) quantifies over all cuts $(A,B)$ of the
honest nodes, where $A$ is the smaller side, and decomposes the bad event into
two parts: (1)~the \emph{heavy} event---too many nodes in $A$ have an unusually
large share of non-$B$ entries; and (2)~the \emph{disconnected} event---no node
in $A$ connects to $B$ in the overlay. Since nodes in $A$ expect
$|B|\cdot s/\sqrt{n}$ entries from $B$ in their slices, the adversary must fill
most tables to avoid triggering flags---but then many nodes hold $B$-entries,
and for large enough $\pconn$ some form cross-cut overlay connections.

The attacker is free to choose the worst-case cut; knowing a node's slice, it
can always place the entire slice on one side, so some fraction of nodes are
trivially kept from flagging. The tolerance $\delta$ captures this unavoidable
fraction. However, even a successful attack that was not detected is
soon noticed: in the next round nodes sample fresh slices, and the partition is
detected with high probability.

\newcommand{\stmtOverlay}{%
  In the presence of an omniscient overlay adversary (\cref{def:omniscient-overlay}),
  let $\delta \in (0,1)$ be a fixed constant, and suppose the
  expected overlay degree satisfies
  \[
    \pconn s\sqrt{n} \geq C n^\kappa \log n,
  \]
  for some $\kappa \in [0, 1/2)$ and a constant $C > 2/(\delta\varepsilon)$,
  where $\varepsilon = (1-\alpha)/6$.
  Then with probability at least
  \[
    1 - e^{-\Theta(\sqrt n)} - n^{-\Theta(n^\kappa)},
  \]
  every cut $(A,B)$ that disconnects honest nodes in the overlay
  gets detected by the attack-detection mechanism:
  at least a $(1-\delta)$-fraction of nodes in $A$
  raises the attack-detection flag.
}
\begin{theorem}\label{thm:overlay}
  \stmtOverlay
\end{theorem}

When $\kappa > 0$ the failure probability is negligible. The boundary case
$\kappa = 0$ gives polynomial decay $n^{-\Theta(1)}$, and the required
degree $\Omega(\log n)$ is tight up to constants by classical results on
random-graph connectivity.

\paragraph{Concrete overlay parameters.} \label{sec:parameter-concrete}

To see how the bound behaves at practical network sizes (below the asymptotic
regime), we run a Monte Carlo simulation of the cut attack
($n=10{,}000$, $s=4$, $\alpha=0.5$, $\delta=0.25$).
\Cref{fig:cut-attack-sim}(a) plots the per-cut success probability
as a function of cut size~$k$ for several overlay degrees~$m$; the dashed line
shows $\binom{|\Honset|}{k}$ (the number of cuts, relevant for a union bound).
For $m=50$ the decay dominates the combinatorial growth.
Since small cuts dominate the union bound (larger cuts contribute negligibly),
\cref{fig:cut-attack-sim}(b) focuses on isolating a single node ($k=1$) and
shows how the success probability decreases as $\theta$ increases, confirming
effective detection across a wide parameter range.
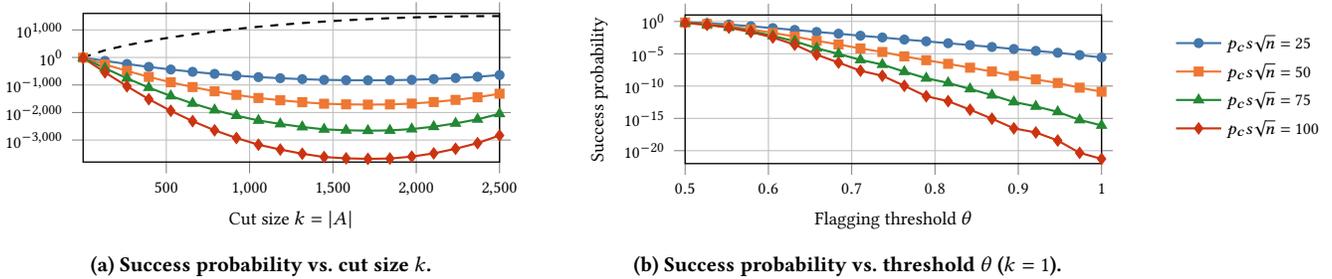
\begin{figure*}
    \centering
    \begin{subfigure}[b]{\awplotiiwidth}
        \centering
        \begin{tikzpicture}
            \begin{axis}[
                awplot,
                xlabel={Cut size $k = |A|$},
                xmin=1, xmax=2500,
                ymin=-3800, ymax=1600,
                width=\awplotiiwidth,
                height=\awplotheight,
                ytick={-3000,-2000,-1000,0,1000},
                yticklabel={$10^{\pgfmathprintnumber[fixed,precision=0]{\tick}}$},
            ]
                \addplot[color=plotA, thick, mark=*, mark size=1.5pt]
                    table[x=k, y=m25] {figures/cut-attack-data/cut_attack_k.dat};
                \addplot[color=plotB, thick, mark=square*, mark size=1.5pt]
                    table[x=k, y=m50] {figures/cut-attack-data/cut_attack_k.dat};
                \addplot[color=plotC, thick, mark=triangle*, mark size=1.8pt]
                    table[x=k, y=m75] {figures/cut-attack-data/cut_attack_k.dat};
                \addplot[color=plotD, thick, mark=diamond*, mark size=1.8pt]
                    table[x=k, y=m100] {figures/cut-attack-data/cut_attack_k.dat};
                \addplot[black, thick, dashed, mark=none]
                    table[x=k, y=logbinom] {figures/cut-attack-data/binom_ref.dat};
            \end{axis}
        \end{tikzpicture}
        \caption{Success probability vs.\ cut size $k$.}
        \label{fig:cut-attack-k}
    \end{subfigure}
    \hfill
    \begin{subfigure}[b]{\awplotiiwidth}
        \centering
        \begin{tikzpicture}
            \begin{axis}[
                awplot,
                xlabel={Flagging threshold $\theta$},
                ylabel={Success probability},
                xmin=0.5, xmax=1.0,
                ymin=-22, ymax=1,
                width=\awplotiiwidth,
                height=\awplotheight,
                ytick={-20,-15,-10,-5,0},
                yticklabel={$10^{\pgfmathprintnumber[fixed,precision=0]{\tick}}$},
                legend to name=cutlegend,
                legend columns=1,
                legend cell align=left,
            ]
                \addplot[color=plotA, thick, mark=*, mark size=1.5pt]
                    table[x=theta, y=m25] {figures/cut-attack-data/cut_attack_theta.dat};
                \addlegendentry{$\pconn s\sqrt{n}=25$}

                \addplot[color=plotB, thick, mark=square*, mark size=1.5pt]
                    table[x=theta, y=m50] {figures/cut-attack-data/cut_attack_theta.dat};
                \addlegendentry{$\pconn s\sqrt{n}=50$}

                \addplot[color=plotC, thick, mark=triangle*, mark size=1.8pt]
                    table[x=theta, y=m75] {figures/cut-attack-data/cut_attack_theta.dat};
                \addlegendentry{$\pconn s\sqrt{n}=75$}

                \addplot[color=plotD, thick, mark=diamond*, mark size=1.8pt]
                    table[x=theta, y=m100] {figures/cut-attack-data/cut_attack_theta.dat};
                \addlegendentry{$\pconn s\sqrt{n}=100$}

                \addplot[black, thick, dashed, mark=none, forget plot] coordinates {(0,0)};
                \addlegendentry{$\binom{|\Honset|}{k}$}
            \end{axis}
        \end{tikzpicture}
        \caption{Success probability vs.\ threshold $\theta$ ($k=1$).}
        \label{fig:cut-attack-theta}
    \end{subfigure}
    \hfill
    \raisebox{5.5em}{\ref{cutlegend}}
    \caption{Monte Carlo simulation of the cut attack ($n{=}10{,}000$,
    $s{=}4$, $\alpha{=}0.5$, $\delta{=}0.25$, $5{,}000$ trials).
    (a)~Attack success probability as a function of cut size~$k$ with
    $\theta{=}0.9$.  The dashed line shows the number of cuts
    $\binom{|\Honset|}{k}$; the attack must overcome this combinatorial
    factor for a union bound.
    (b)~Success probability for isolating a single node ($k{=}1$) as a
    function of the flagging threshold~$\theta$.  Higher~$\theta$ and
    larger overlay degrees~$\pconn s \sqrt{n}$ both sharply reduce the adversary's chances.}
    \label{fig:cut-attack-sim}
\end{figure*}

\subsection{Spam Prevention via Slashing}
\protocolname limits the number of requests a node can send in a round to $s\sqrt{n}$
through the commitment mechanism (\cref{sec:too-many-reqs}; \algline{alg:heartbeat}{line:hb-commit}). Without such a
limit, an adversarial node could spam the network with an arbitrarily large
number of $\request$ messages, leading to denial-of-service attacks or attempts
to overrepresent itself in peer tables. We introduce a slashing condition that
rate-limits each node's requests and penalizes nodes that violate this
condition.

\begin{definition}[slashing condition]
    A node violates the slashing condition in round $r$ if it sends requests to
    more than $s\sqrt{n}$ distinct peers in that round.
\end{definition}

If an adversarial node violates the slashing condition, there will be two
(potentially distinct) nodes that hear about commitments that, when combined,
reveal cryptographic evidence for slashing
(\alglines{alg:receive_new_record}{line:rec-dup-detect}{line:rec-recover}).
As we have seen in
\cref{sec:mean-field}, in a healthy network state, an honest node will
hear both commitments with high probability, leading to slashing of the
violator's stake.
\newcommand{\stmtStakeSecretRecovery}{%
    If a node
    $i$ violates the slashing condition in round $r$, then there exist two valid
    requests sent by $i$ in round $r$ such that the contents allow efficient
    recovery of $i$'s stake secret $\stakesk_i$.}
\begin{lemma}[stake secret recovery] \label{lem:stake-secret-recovery}
  \stmtStakeSecretRecovery
\end{lemma}

Now we argue that the slashing mechanism satisfies soundness (honest nodes are
not slashed) and completeness (violators are slashed with high probability) in a
healthy network state.
\newcommand{\stmtSlashingSoundness}{%
    Let $i$ be an
    honest node that sends at most $s\sqrt{n}$ requests in round $r$. Then, the
    probability that a computationally-bounded adversary learns $i$'s stake
    secret $\stakesk_i$ is at most $\negl(\lambda)$.}
\begin{lemma}[slashing soundness] \label{lem:slashing-soundness}
  \stmtSlashingSoundness
\end{lemma}

\newcommand{\stmtSlashingCompleteness}{%
    Assume that the network is in a healthy
    state (i.e.\ $\gamma$-healthy). If a node $i$ violates the
    slashing condition in round $r$ by sending requests under two distinct
    commitments, then $i$'s stake secret $\stakesk_i$ is recovered by some
    honest node with probability at least
    $1 - \exp\left(-\Omega(s^2)\right)$.}
\begin{corollary}[slashing completeness]
    \label{lem:slashing-completeness}
  \stmtSlashingCompleteness\footnote{In our simulations, detection occurs with overwhelming probability within a single round (\cref{fig:slashing}).}
\end{corollary}

\subsection{Stake and Connection Privacy}
Lastly, we analyze the privacy properties of the \protocolname protocol. We
achieve two key privacy properties:
\begin{itemize}
	\item \emph{Stake Anonymity:} A node's stake identity $\stakeID$ cannot be
	linked to its network identity $\netpk$ from the message the node sends in a
	round, provided the per-round request limit is respected.
	\item \emph{Connection privacy:} An adversary cannot determine which subset of
	peers a given node retains for constructing its overlay table $\otable$ in a
	given round.
\end{itemize}
\subsubsection{Stake Anonymity}
We define stake anonymity via an indistinguishability game: observing a round
of requests from $\netpk$ should not reveal which $\stakeID$ backs it as long as the requester did not generate more than one commitment to his set of requests.

\begin{definition}[Stake Anonymity]
Let $\stgame$ be defined as follows for some fixed round $r$:
\begin{enumerate}
	\item The challenger samples keys for $n$ nodes $\{\sk\}$ and computes the
  corresponding $\stakesk$, $\stakeID$, and $(\netsk,\netpk)$.
	\item The adversary $\mathcal{A}$ is given the set of stake identities
	$\{\stakeID\}$ and selects two distinct ones $\stakeID_0$, $\stakeID_1$.
	\item The challenger samples a bit $b \leftarrow \{0,1\}$ and selects
  $\stakeID_b$ with corresponding $\sk_b$, $\stakesk_b$, and $\netpk_b$.
	\item The challenger runs one execution of round $r$ for a node with secret
	key $\sk_b$ among $n$ nodes, generating all messages sent by that node in that
	round.
	\item $\mathcal{A}$ observes the transcript of messages and outputs a bit
	$b'$.
\end{enumerate}
Define the advantage
\(
  \stadv := \left| \Pr[b' = b] - \frac{1}{2} \right|.
\)\\
\protocolname satisfies \textbf{stake anonymity} if for all computationally
bounded adversaries $\mathcal{A}$: $\stadv \leq \negl(\lambda)$.
\end{definition}

\newcommand{\stmtStakeAnon}{%
  If the protocol limits
	requests to $s\sqrt{n}$ per round, the \protocolname protocol satisfies stake
	anonymity.}
\begin{lemma}[stake anonymity] \label{lem:stake-anon}
  \stmtStakeAnon
\end{lemma}
\subsubsection{Connection Privacy}
We define connection privacy via an indistinguishability game between an
adversary and a challenger. In this section we analyze with the worst-case
assumption that the adversary can see all gossip tables while the challenger
node constructs its overlay table.

\begin{definition}[Connection Privacy]
Fix a round $r$ and a node $i$ in a protocol execution with $n$
nodes. Let $\cgame$ be defined as follows:
\begin{enumerate}
	\item The adversary $\mathcal{A}$ outputs two candidate overlay nonces
  $\nonceoverlay_0, \nonceoverlay_1$ that induce two distinct peer subsets $P_0
  \neq P_1$ from the set of peer records accessible by node $i$ in round $r$.
	\item The challenger samples bit $b \leftarrow \{0,1\}$ and runs one honest
	execution of round $r$ for node $i$ selecting exactly $P_b$ for the overlay
	table $\otable$.
	\item $\mathcal{A}$ observes the full transcript of messages exchanged in
	round $r$ in addition to the contents of all gossip tables
	$\ptable$.\footnote{Note that we assume the peer retrieval by node $i$ is done
	by full-table transmission and client-side filtering. Alternatively, TEEs can
	be used to reduce the communication overhead while still achieving connection
	privacy.} and outputs $b'$.
\end{enumerate}
Define the advantage
\(
  \cadv := \left| \Pr[b' = b] - \frac{1}{2} \right|.
\)\\
\protocolname satisfies \textbf{connection privacy} if for all adversaries
$\mathcal{A}$, $\cadv = 0$.
\end{definition}
The guarantee is information-theoretic (it holds against unbounded
adversaries) and is achieved under either full-table transmission
(each responder sends its entire $\ptable$ and the requester filters
locally) or TEE-based filtering (\cref{rem:tee-privacy}).
Server-side filtering without a TEE reveals the selection.

\newcommand{\stmtConnPrivacy}{%
  \protocolname
	satisfies connection privacy under full-table transmission, where each node
	response returns a complete $\ptable$ to requesting nodes, and nodes perform
	selection locally without revealing their nonce value $\nonceoverlay$.}
\begin{lemma}[connection privacy] \label{lem:conn-privacy}
  \stmtConnPrivacy
\end{lemma}

\section{Evaluation}
\label{sec:eval}

To evaluate our protocol, we implemented a custom event-driven
simulator\footnote{Source code: \url{https://github.com/CedArctic/aetherweave_simulator}.}
as well as a full prototype implementation (\cref{sec:prototype}) and
conducted experiments across various scenarios. This section presents the simulation results,
demonstrating the protocol's stability under churn
and its security against adversarial manipulation.

\paragraph{Simulator implementation.}
We built a discrete-event simulator in Python (${\sim}2\text{k}$ LOC).
The simulator maintains a heap-based priority queue of events---heartbeats,
message deliveries, churn events, and metric snapshots---and advances
simulated time only when processing the next event.
Cryptographic operations (commitments, signatures, VRFs) are replaced by fast
hashing (\texttt{xxhash3}); this abstraction is sound because our analysis
depends on the statistical properties of peer selection, not on cryptographic
hardness.
The network is modeled as fully connected with uniform random message delays
($20$--$250$\,ms).
Each data point aggregates multiple independent runs with distinct random
seeds for statistical confidence.

\subsection{Network health}
We first establish the baseline performance of the protocol in maintaining
network connectivity and accommodating new participants.

\subsubsection{Churn resistance.}
We measure \emph{record correctness}: the fraction of peer records in honest
nodes' peer tables whose stored network address matches the peer's current
address. We model churn as a random process with a fixed rate $c$, where $c$
is the number of nodes that change their network address per round (each event
corresponds to a node updating its network address or leaving and re-joining).
We observe that record correctness degrades only marginally under high churn
rates (\cref{fig:churn}).

\subsubsection{Table quality.}
We define the \emph{table quality} of a node as the fraction of its
slice---the set of honest nodes selected by its nonce---that is
present in its peer table with a correct network address. We report the
average table quality across all honest nodes in the presence of churn and a
fixed fraction of \emph{silent} (non-responding) nodes, i.e., nodes that do not
respond to peer-discovery requests. As shown in \cref{fig:table-quality}, the
protocol maintains high-quality peer tables even under adverse conditions,
with only a slight decrease under high churn rates and a
reasonable fraction of silent nodes.
\begin{figure*}
    \centering
    \begin{subfigure}[b]{\awplotwidth}
        \centering
        \begin{tikzpicture}
            \begin{axis}[
                awplot,
                xlabel={Rounds},
                ylabel={Nodes storing record},
                xmin=0, xmax=20,
                ymin=0, ymax=440,
                width=\awplotwidth,
                height=\awplotheight,
                legend pos=south east,
                legend style={fill=white, fill opacity=0.8, text opacity=1},
            ]
                \addplot[dashed, thick, gray, domain=0:25] {400};
                \addlegendentry{Target $s\sqrt{n}$};
                \label{plt:target-s}

                \addplot [color=plotA, thick, mark=none]
                    table[x=round, y=appearances]
                    {large_experiments/data/bootstrap.txt};
                \addlegendentry{Joining Node};
                \label{plt:bootstrap-line}
            \end{axis}
        \end{tikzpicture}
        \caption{Bootstrapping}
        \label{fig:bootstrap}
    \end{subfigure}
    \hfill
    \begin{subfigure}[b]{\awplotwidth}
        \centering
        \begin{tikzpicture}
            \begin{axis}[
                awplot,
                xlabel={Rounds},
                ylabel={Correct records (\%)},
                xmin=0, xmax=50,
                ymin=98, ymax=100.15,
                width=\awplotwidth,
                height=\awplotheight,
                legend pos=south west,
                legend columns=2,
                legend style={fill=white, fill opacity=0.8, text opacity=1},
            ]
                \addplot [thick, gray, dashed, domain=0:50] {100};
                \addlegendentry{No Churn};

                \addplot [color=plotA, thick, mark=none]
                    table[x=time, y expr=\thisrow{success_rate}*100]
                    {large_experiments/data/churn_10.txt};
                \addlegendentry{$c=10$};
                \label{plt:churn-low}

                \addplot [color=plotB, thick, mark=none]
                    table[x=time, y expr=\thisrow{success_rate}*100]
                    {large_experiments/data/churn_50.txt};
                \addlegendentry{$c=50$};
                \label{plt:churn-med}

                \addplot [color=plotC, thick, mark=none]
                    table[x=time, y expr=\thisrow{success_rate}*100]
                    {large_experiments/data/churn_100.txt};
                \addlegendentry{$c=100$};
                \label{plt:churn-high}
            \end{axis}
        \end{tikzpicture}
        \caption{Churn resistance}
        \label{fig:churn}
    \end{subfigure}
    \hfill
    \begin{subfigure}[b]{\awplotwidth}
        \centering
        \begin{tikzpicture}
            \begin{axis}[
                awplot,
                xlabel={Silent node fraction (\%)},
                ylabel={Avg.\ table quality (\%)},
                xmin=0, xmax=90,
                ymin=60, ymax=103,
                xtick={0, 20, 40, 60, 80},
                width=\awplotwidth,
                height=\awplotheight,
                legend pos=south west,
                legend style={fill=white, fill opacity=0.8, text opacity=1},
            ]

                \addplot [color=plotA, thick, mark=*, mark size=1.5pt]
                    table[x expr=\thisrow{alpha}*100, y expr=\thisrow{quality}*100]
                    {large_experiments/data/table_quality_churn_10.txt};
                \addlegendentry{$c=10$};
                \label{plt:low-churn-line}

                \addplot [color=plotB, thick, mark=triangle*, mark size=1.5pt]
                    table[x expr=\thisrow{alpha}*100, y expr=\thisrow{quality}*100]
                    {large_experiments/data/table_quality_churn_50.txt};
                \addlegendentry{$c=50$};
                \label{plt:med-churn-line}

                \addplot [color=plotC, thick, mark=o, mark size=1.5pt]
                    table[x expr=\thisrow{alpha}*100, y expr=\thisrow{quality}*100]
                    {large_experiments/data/table_quality_churn_100.txt};
                \addlegendentry{$c=100$};
                \label{plt:high-churn-line}
            \end{axis}
        \end{tikzpicture}
        \caption{Table quality with churn}
        \label{fig:table-quality}
    \end{subfigure}

    \caption{Network health simulations ($n = 10{,}000$, $s\sqrt{n} = 400$).
    (a)~A joining node reaches its target representation within a few rounds.
    (b)~Record correctness remains above 99\% even under high churn ($c$ = address changes per round).
    (c)~Table quality degrades gracefully with silent (non-responding) nodes across churn rates.}
    \label{fig:dynamics}
\end{figure*}
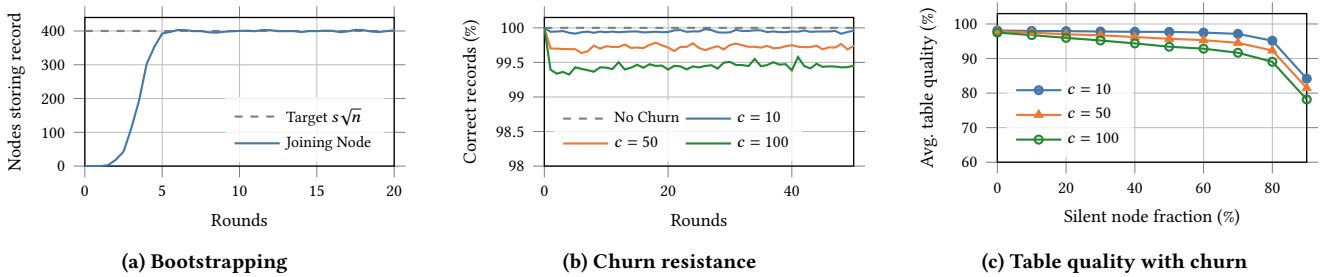

\subsubsection{Bootstrapping.}
Next, we simulate the process of a new node joining and bootstrapping into the
system (\cref{fig:bootstrap}). We track the number of existing
participants that list the joining node in their peer tables after each
round. Consistent with our theoretical analysis, we observe that the node's
representation in the network converges rapidly to $s\sqrt{n}$.

\subsection{Adversarial Resilience}
\label{sec:eval:security}

We analyze the protocol's robustness against adversaries attempting to increase
their representation beyond their fair share (proportional to their stake). We
consider two distinct attack vectors: \emph{filtering} (censorship) and
\emph{oversampling} (protocol violation).

\subsubsection{The filtering attack.}
We first examine an adversary controlling $\alpha$ fraction of the total stake
who attempts to dominate the views of honest nodes. In this scenario,
adversarial nodes selectively "filter out" records from
honest nodes to amplify the relative frequency of adversarial records.

\cref{fig:filtering} demonstrates that even with significant stake, the
adversary fails to monopolize the network. Honest nodes remain
well-represented in the peer tables of their peers, preventing the adversary
from over-representing itself.

\subsubsection{The slashing mechanism.}
Finally, we evaluate the detectability of an adversary attempting to request
data from more than the permitted $s\sqrt{n}$ nodes per round. We simulate different
adversaries that send peer records to $k \cdot s\sqrt{n}$ targets per round for
$k = 2,3,4$ to accelerate their view construction.

As shown in \cref{fig:slashing}, detection is almost certain within one round.
Each curve corresponds to a different table-size scaling $s$ (where
$s\sqrt{n}$) and aggregates detection events across all three $k$ values,
showing that detection speed increases with $s$. The plot shows the cumulative
fraction of adversarial heartbeat rounds in which the adversary is detected.
This confirms that deviations from the protocol are detectable with high
probability and can be penalized via slashing.

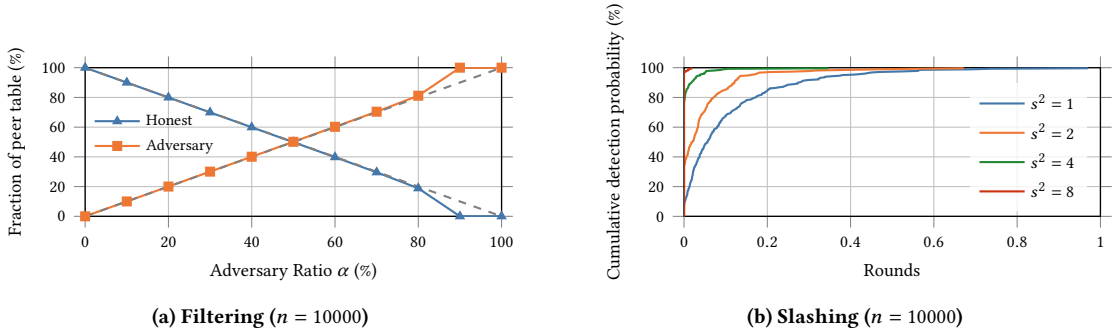
\begin{figure*}
    \centering
    \begin{subfigure}[b]{\awplotiiwidth}
        \centering
        \begin{tikzpicture}
            \begin{axis}[
                awplot,
                xlabel={Adversary Ratio $\alpha$ (\%)},
                ylabel={Fraction of peer table (\%)},
                xmin=0, xmax=100,
                ymin=0, ymax=100,
                width=\awplotiiwidth,
                height=\awplotheight,
                xtick={0, 20, 40, 60, 80, 100},
                legend style = {at={(0.0,0.35)}, anchor=south west}%
            ]
                \addplot [color=plotA, thick, mark size=1.5pt, mark=triangle*]
                    table[x expr=\thisrow{alpha}*100, y expr=\thisrow{representation}*100]
                    {large_experiments/data/filtering_honest.txt};
                \addlegendentry{Honest};
                \label{plt:filtering-honest}

                \addplot [color=plotB, thick, mark size=1.5pt, mark=square*]
                    table[x expr=\thisrow{alpha}*100, y expr=\thisrow{representation}*100]
                    {large_experiments/data/filtering_adversary.txt};
                \addlegendentry{Adversary};
                \label{plt:filtering-adversary}

                \addplot [color=gray, thick, dashed, domain=0:100] {100-x};
                \addplot [color=gray, thick, dashed, domain=0:100] {x};
            \end{axis}
        \end{tikzpicture}
        \caption{Filtering ($n=10000$)}
        \label{fig:filtering}
    \end{subfigure}
    \hspace{2em}
    \begin{subfigure}[b]{\awplotiiwidth}
        \centering
        \begin{tikzpicture}
            \begin{axis}[
                awplot,
                xlabel={Rounds},
                ylabel={Cumulative detection probability (\%)},
                xmin=0, xmax=1,
                ymin=0, ymax=100,
                width=\awplotiiwidth,
                height=\awplotheight,
                legend pos=south east,
                legend style={fill=white, fill opacity=0.8, text opacity=1},
            ]
                \addplot [color=plotA, thick]
                    table[x=time, y expr=\thisrow{cdf}*100]
                    {large_experiments/data/slashing_rs1.txt};
                \addlegendentry{$s^2=1$};

                \addplot [color=plotB, thick]
                    table[x=time, y expr=\thisrow{cdf}*100]
                    {large_experiments/data/slashing_rs2.txt};
                \addlegendentry{$s^2=2$};

                \addplot [color=plotC, thick]
                    table[x=time, y expr=\thisrow{cdf}*100]
                    {large_experiments/data/slashing_rs4.txt};
                \addlegendentry{$s^2=4$};

                \addplot [color=plotD, thick]
                    table[x=time, y expr=\thisrow{cdf}*100]
                    {large_experiments/data/slashing_rs8.txt};
                \addlegendentry{$s^2=8$};

                \label{plt:slashing}
            \end{axis}
        \end{tikzpicture}
        \caption{Slashing ($n=10000$)}
        \label{fig:slashing}
    \end{subfigure}
    \vspace{0.5em}
    \caption{Adversarial resilience with $n = 10000$ and $s\sqrt{n} = 400$.
    (a) \emph{Filtering}: the fraction of honest nodes' address-table entries
    occupied by each class. Honest nodes remain well-represented
    (\ref{plt:filtering-honest}) compared to the adversary
    (\ref{plt:filtering-adversary}), even as the adversary ratio $\alpha$ increases.
    (b) \emph{Slashing}: cumulative detection probability (CDF) over time for an
    adversary sending requests to $k \cdot s\sqrt{n}$ targets per round, aggregated over
    $k = 2,3,4$. Each curve uses a different table-size scaling $s\sqrt{n}$,
    parameterized by $s^2$. Detection approaches $100\%$ rapidly; when $s^2 \ge 1$
    (\ref{plt:slashing}), the adversary is caught with high probability within a
    single round.}
    \label{fig:security}
\end{figure*}

Having validated the protocol's properties in simulation, we now turn to a
prototype implementation that demonstrates its practicality on a real
blockchain platform.

\section{Prototype implementation} \label{sec:prototype}
To evaluate the practicality of \protocolname{}, we implement a full prototype
using three components: a staking smart contract on Ethereum,
ZKP circuits for stake and share verification, and a modified
consensus client that runs the protocol over libp2p.\footnote{Prototype source code: \url{https://github.com/CedArctic/aetherweave_prysm}. Staking contract: \url{https://github.com/CedArctic/aetherweave_contract}.}

Our prototype is built by forking Prysm~\cite{prysm}, a widely used
Go implementation of the Ethereum consensus client. We implement
\stakeproof and \shareproof as Circom~\cite{circom} circuits using the
Groth16~\cite{groth2016size} proof system with Baby Jubjub keys and
Poseidon~\cite{poseidon} hashes, and develop a Solidity staking contract
following the interface in \cref{alg:smartcontract}.
The client integrates \texttt{rapidsnark}~\cite{go-rapidsnark} for proof
generation and verification, and communicates with the staking contract via
standard JSON-RPC.
For full implementation details, see \cref{appendix:prototype}.

\paragraph{On-chain costs.}
Because ZKP verification is off-chain, the staking contract involves no heavy
operations. Cost is dominated by Poseidon rehashing in the sparse Merkle tree,
which grows linearly with traversal depth at ${\approx}46$k gas per level. At
$10{,}000$ stakers on a depth $32$ tree, we get estimated averages of
${\approx}873$k gas for \textsc{DepositAndStake}, ${\approx}668$k for
\textsc{UnStake}, and ${\approx}679$k for \textsc{Slash}. \textsc{ClaimFunds}
(which only clears storage mappings) costs a constant ${\approx}32$k gas. All
operations remain well within the block gas limit.

\paragraph{Experimental setup.}
\label{sec:experimental_setup}
We deploy a simulated Ethereum testnet using the
\texttt{ethereum-package} with the Kurtosis Docker orchestrator on a
server with an AMD EPYC 9354P processor and 512\,GB of RAM.
We set the table scaling constant to $s = 4$ and configure short 2-minute
protocol rounds.
Our infrastructure runs up to 100 full nodes concurrently;
considering that table size equals $s\sqrt{n}$, this allows us to
simulate deployments with network sizes of $n = 100$, $225$, $400$,
and $625$.
Where necessary, we benchmark individual functions and
estimate conservative execution times for larger network sizes
based on the expected number of invocations per round.

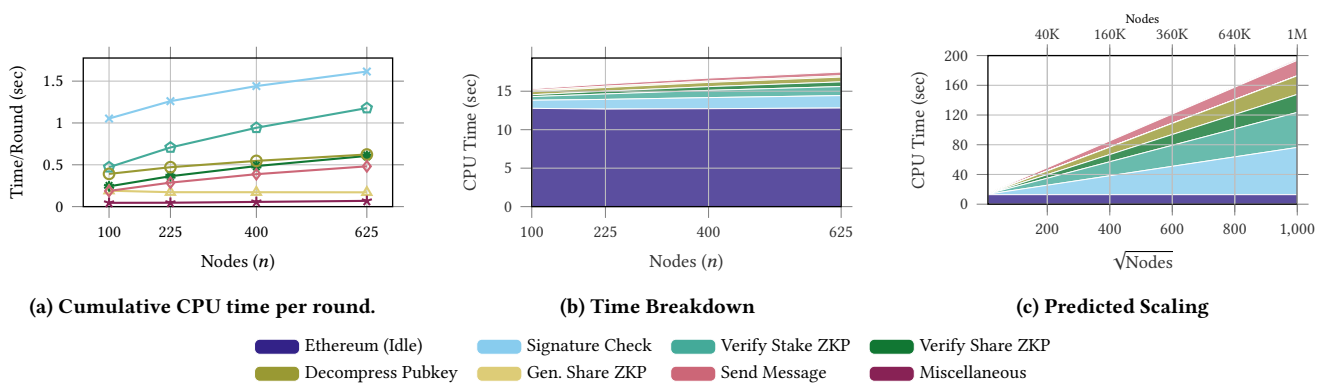
\begin{figure*}[t]
    \centering

    \begin{subfigure}[b]{\awplotwidth}
        \centering
        \begin{tikzpicture}
        \begin{axis}[
            awplot,
            width=\awplotwidth,
            height=\awplotheight,
            xlabel={Nodes ($n$)},
            ylabel={Time/Round (sec)},
            ymin=0,
            xtick=data,
            legend style={draw=none, fill=none, font=\tiny},
            mark size=2pt,
            cycle list={
                {fill=SciBlue, draw=white, line width=0.1pt},
                {fill=SciCyan, draw=white, line width=0.1pt},
                {fill=SciTeal, draw=white, line width=0.1pt},
                {fill=SciGreen, draw=white, line width=0.1pt},
                {fill=SciOlive, draw=white, line width=0.1pt},
                {fill=SciSand, draw=white, line width=0.1pt},
                {fill=SciRose, draw=white, line width=0.1pt},
                {fill=SciWine, draw=white, line width=0.1pt},
                {fill=SciPurple, draw=white, line width=0.1pt},
                {fill=SciGray, draw=white, line width=0.1pt}
            }
        ]
        \addplot[color=SciCyan, thick, mark=x] table[x=nodes, y=verifyAWMessage] {prototype/profiling/verifyAWMessage_scaling.dat};
        \addplot[color=SciTeal, thick, mark=pentagon] table[x=nodes, y=VerifyGroth16] {prototype/profiling/VerifyGroth16_Stake_predicted_small.dat};
        \addplot[color=SciGreen, thick, mark=10-pointed star] table[x=nodes, y=VerifyGroth16] {prototype/profiling/VerifyGroth16_Share_predicted_small.dat};
        \addplot[color=SciOlive, thick, mark=o] table[x=nodes, y=processMarshalledPubkey] {prototype/profiling/processMarshalledPubkey_scaling.dat};
        \addplot[color=SciSand, thick, mark=triangle] table[x=nodes, y=generateZKProof] {prototype/profiling/generateZKProof_scaling.dat};
        \addplot[color=SciRose, thick, mark=diamond] table[x=nodes, y=p2p.(*Service).Send] {prototype/profiling/p2p.Send_scaling.dat};
        \addplot[color=SciWine, thick, mark=star] table[x=nodes, y=aw_misc] {prototype/profiling/stacked_scaling.dat};
        \end{axis}
        \end{tikzpicture}
        \caption{Cumulative CPU time per round.}
        \label{fig:cpu_func_scaling}
    \end{subfigure}
    \hfill
    \begin{subfigure}[b]{\awplotwidth}
        \centering
        \begin{tikzpicture}
        \begin{axis}[
            awplot,
            width=\awplotwidth,
            height=\awplotheight,
            xlabel={Nodes ($n$)},
            ylabel={CPU Time (sec)},
            stack plots=y,
            area style,
            legend to name=commonlegend,
            legend columns=4,
            legend style={
                font=\footnotesize,
                rounded corners=2pt,
                /tikz/every even column/.style={column sep=5pt}
            },
            enlarge x limits=false,
            scaled x ticks=false,
            xtick=data,
            ymin=0,
            draw=black,
            fill opacity=0.8,
            cycle list={
                {fill=SciBlue, draw=white, line width=0.1pt},
                {fill=SciCyan, draw=white, line width=0.1pt},
                {fill=SciTeal, draw=white, line width=0.1pt},
                {fill=SciGreen, draw=white, line width=0.1pt},
                {fill=SciOlive, draw=white, line width=0.1pt},
                {fill=SciSand, draw=white, line width=0.1pt},
                {fill=SciRose, draw=white, line width=0.1pt},
                {fill=SciWine, draw=white, line width=0.1pt},
                {fill=SciPurple, draw=white, line width=0.1pt},
                {fill=SciGray, draw=white, line width=0.1pt}
            }
        ]
        \addplot table[x=nodes, y=Other] {prototype/profiling/stacked_scaling.dat} \closedcycle; \addlegendentry{Ethereum (Idle)}
        \addplot table[x=nodes, y=verifyAWMessage] {prototype/profiling/stacked_scaling.dat} \closedcycle; \addlegendentry{Signature Check}
        \addplot table[x=nodes, y=VerifyGroth16] {prototype/profiling/VerifyGroth16_Stake_predicted_small.dat} \closedcycle; \addlegendentry{Verify Stake ZKP}
        \addplot table[x=nodes, y=VerifyGroth16] {prototype/profiling/VerifyGroth16_Share_predicted_small.dat} \closedcycle; \addlegendentry{Verify Share ZKP}
        \addplot table[x=nodes, y=processMarshalledPubkey] {prototype/profiling/stacked_scaling.dat} \closedcycle; \addlegendentry{Decompress Pubkey}
        \addplot table[x=nodes, y=generateZKProof] {prototype/profiling/stacked_scaling.dat} \closedcycle; \addlegendentry{Gen. Share ZKP}
        \addplot table[x=nodes, y={p2p.(*Service).Send}] {prototype/profiling/stacked_scaling.dat} \closedcycle; \addlegendentry{Send Message}
        \addplot table[x=nodes, y={aw_misc}] {prototype/profiling/stacked_scaling.dat} \closedcycle; \addlegendentry{Miscellaneous}
        \end{axis}
        \end{tikzpicture}
        \caption{Time Breakdown}
        \label{fig:cpu_stacked}
    \end{subfigure}
    \hfill
    \begin{subfigure}[b]{\awplotwidth}
        \centering
        \begin{tikzpicture}
        \begin{axis}[
            awplot,
            width=\awplotwidth,
            height=\awplotheight,
            xlabel={$\sqrt{\text{Nodes}}$},
            xlabel style={yshift=2.2ex},
            ylabel={CPU Time (sec)},
            xmin=10, xmax=1000,
            ymin=0, ymax=200,
            ytick={0, 40, 80, 120, 160, 200},
            xtick={200, 400, 600, 800, 1000},
            stack plots=y,
            area style,
            draw=black,
            fill opacity=0.8,
            extra x ticks={200,400,600,800,1000},
            extra x tick labels={40K, 160K, 360K, 640K, 1M},
            extra x tick style={ticklabel pos=top, major tick length=3pt},
            cycle list={
                {fill=SciBlue, draw=white, line width=0.1pt},
                {fill=SciCyan, draw=white, line width=0.1pt},
                {fill=SciTeal, draw=white, line width=0.1pt},
                {fill=SciGreen, draw=white, line width=0.1pt},
                {fill=SciOlive, draw=white, line width=0.1pt},
                {fill=SciSand, draw=white, line width=0.1pt},
                {fill=SciRose, draw=white, line width=0.1pt},
                {fill=SciWine, draw=white, line width=0.1pt},
                {fill=SciPurple, draw=white, line width=0.1pt},
                {fill=SciGray, draw=white, line width=0.1pt}
            }
        ]
        \addplot table[x={nodes_sqrt}, y=Other] {prototype/profiling/stacked_predicted.dat} \closedcycle;
        \addplot table[x={nodes_sqrt}, y=verifyAWMessage] {prototype/profiling/stacked_predicted.dat} \closedcycle;
        \addplot table[x={nodes_sqrt}, y={VerifyGroth16_Stake}] {prototype/profiling/stacked_predicted.dat} \closedcycle;
        \addplot table[x={nodes_sqrt}, y={VerifyGroth16_Share}] {prototype/profiling/stacked_predicted.dat} \closedcycle;
        \addplot table[x={nodes_sqrt}, y=processMarshalledPubkey] {prototype/profiling/stacked_predicted.dat} \closedcycle;
        \addplot table[x={nodes_sqrt}, y=generateZKProof] {prototype/profiling/stacked_predicted.dat} \closedcycle;
        \addplot table[x={nodes_sqrt}, y={p2p.(*Service).Send}] {prototype/profiling/stacked_predicted.dat} \closedcycle;
        \addplot table[x={nodes_sqrt}, y={aw_misc}] {prototype/profiling/stacked_predicted.dat} \closedcycle;
        \end{axis}
        \node[font=\tiny] at (rel axis cs:0.5,1.25) {Nodes};
        \end{tikzpicture}
        \caption{Predicted Scaling}
        \label{fig:cpu_stacked_predicted}
    \end{subfigure}

    \vspace{0.2em}
    \ref{commonlegend} %
    \vspace{-1em}
    \caption{CPU performance profiling and scaling across major protocol functions.}
    \label{fig:performance_all}
\end{figure*}

\paragraph{Experimental results.}
\label{sec:experimental_results}
\Cref{fig:performance_all} presents the CPU performance profile of the
protocol.
The cost is dominated by Baby Jubjub signature checks, public key
decompression, and ZKP verification;
\cref{tbl:complexities} reports the per-operation execution times and
expected invocation counts.
The predicted scaling (\cref{fig:cpu_stacked_predicted}) shows
that even at one million nodes, the protocol requires only a few CPU
cores with short 2-minute rounds.
Network traffic overhead also grows linearly with~$\sqrt{n}$, consistent
with the theoretical communication complexity.
Detailed per-round CPU and network traffic profiles
(\cref{fig:cpu_period,fig:net_rx_period,fig:net_rx_scaling}) are
provided in \cref{appendix:prototype}.

Our results confirm that \protocolname introduces modest
computational and networking overhead while maintaining the expected
asymptotic scaling properties.

    \begin{table}[h]
        \centering
        \resizebox{\linewidth}{!}{
        \begin{tblr}{
          cells = {c, m},
          rowsep = 1pt,
          row{1} = {font=\bfseries},
          hline{1-2,8} = {-}{0.08em},
        }
        Function                & {Invocations\\per Round}                         & {CPU Time per \\Invocation (sec)}\\
        Generate Share ZKP         & $1$                                 & 0.171              \\
        Verify Share ZKP            & $(s + 1) \cdot \sqrt{n}$          & $6.06 \cdot 10^{-3}$ \\
        Verify Stake ZKP            & $(s + 1) \cdot \sqrt{n}$          & $5.89 \cdot 10^{-3}$ \\
        Send Message                & $2 \cdot s \cdot \sqrt{n}$          & $2.58 \cdot 10^{-3}$ \\
        Signature Check         & $s\cdot(s^2 + 1) \cdot \sqrt{n}$    & $9.35 \cdot 10^{-4}$ \\
        Decompress PubKey & $s\cdot(s^2 + 1) \cdot \sqrt{n}$    & $3.68 \cdot 10^{-4}$ \\
        \end{tblr}
        }
        \caption{Per-round invocations and CPU cost.}
        \label{tbl:complexities}
    \end{table}

\section{Related Work} \label{sec:related-work}
Decentralized peer discovery falls into two families:
\emph{structured tables} such as Chord~\cite{stoica2001chord},
Pastry~\cite{rowstron2001pastry}, and
Kademlia~\cite{maymounkov2002kademlia}; and \emph{unstructured
gossiping} protocols like Cyclon~\cite{voulgaris2005cyclon} and
HyParView~\cite{leitao2007hyparview}, which maintain randomized partial
views~\cite{jelasity2007peersampling}. Neither family addresses open,
adversarial settings where Sybil attacks~\cite{douceur2002sybil} allow a
single party to overwhelm honest participants.

\paragraph{Eclipse and Sybil defenses.}
Social-graph approaches such as SybilGuard~\cite{yu2006sybilguard} and
SybilLimit~\cite{yu2008sybillimit} bound Sybil penetration under
trust-graph assumptions; S/Kademlia~\cite{baumgart2007skademlia}
constrains identifier generation and adds redundant lookups. Practical
eclipse attacks have been demonstrated against both
Bitcoin~\cite{heilman2015eclipsebitcoin} and
Ethereum~\cite{marcus2018eclipseethereum} at low cost, while deployed
defenses remain ad-hoc (e.g., IP subnet bans).

\paragraph{Economic mechanisms and stake-weighted security.}
Resource-based defenses date to proof-of-work puzzles~\cite{dwork1992pricing};
modern proof-of-stake chains use locked stake to limit adversarial
influence and enable slashing~\cite{kiayias2017ouroboros,buterin2017casper}.
Coretti et al.~\cite{coretti2022scuttlebutt} propose a stake-weighted
random graph for the network layer but provide no concrete discovery
mechanism for a dynamic participant set. \protocolname{} applies stake as a
Sybil-resistance primitive at the \emph{networking layer}, combining
stake-based discovery with slashing while preserving stake-ownership
privacy.

\paragraph{Network privacy and topology inference.}
Topology-inference attacks can link Bitcoin clients to IP-level
information~\cite{biryukov2014deanonymisation} or reconstruct
significant portions of network
topology~\cite{delgadosegura2019txprobe}. \protocolname{} counters this by
mixing peer information so that peers shared with others are unlinked
from the overlay, making topology inference extremely costly.
Dandelion~\cite{venkatakrishnan2017dandelion,fanti2018dandelionpp}
addresses a complementary problem---transaction-origin anonymity---but does
not cover peer discovery or overlay formation.

\section{Discussion and Conclusion} \label{sec:discussion}

\paragraph{Limitations.}
Symmetric connectivity is assumed; nodes behind NATs or firewalls
require hole-punching or relay mechanisms not modeled here.
Requiring stake excludes capital-less nodes, creating a tension between
attack cost and accessibility.
The seed-privacy guarantee (\cref{rem:tee-privacy}) depends on TEE integrity;
side-channel attacks~\cite{kocher2019spectre,vanbulck2018foreshadow} weaken
this in practice, though a compromised TEE affects only slice privacy---not
correctness, slashing, or eclipse detection.
The protocol punishes misbehavior via slashing but does not reward
honest participation; a rational node could free-ride by declining to serve
responses, suggesting micropayments or non-responsiveness penalties as future
work.
Higher-layer protocols may leak topology information that undermines
connection privacy; similarly, an adversary that drops overlay
connections after peer discovery could partition nodes despite healthy gossip
tables---whether detection analogous to
\cref{sec:eclipse-detect} is possible at the overlay level remains open.

\paragraph{Extensions.}
Private information retrieval could replace TEEs for seed privacy
(\cref{rem:tee-privacy}), removing hardware trust at the cost of higher
computation.
Communication overhead can be reduced via partial querying: by a
coupon-collector argument, a node need only contact a subset of its peers to
obtain a representative slice.

\paragraph{Conclusion.}
We presented \protocolname, a stake-backed peer-discovery protocol for open P2P
networks. \protocolname{} combines economic stake with cryptographic techniques
to achieve eclipse detection, verifiable rate limiting via slashable
commitments, and stake-network unlinkability. Mean-field analysis and
simulations confirm rapid convergence to healthy equilibria and resilience under
adversarial churn. A prototype built on a production consensus client
demonstrates practical scalability.
\label{sec:conclusion}

\bibliography{ref.bib}


\begin{thebibliography}{34}


\ifx \showCODEN    \undefined \def \showCODEN     #1{\unskip}     \fi
\ifx \showISBNx    \undefined \def \showISBNx     #1{\unskip}     \fi
\ifx \showISBNxiii \undefined \def \showISBNxiii  #1{\unskip}     \fi
\ifx \showISSN     \undefined \def \showISSN      #1{\unskip}     \fi
\ifx \showLCCN     \undefined \def \showLCCN      #1{\unskip}     \fi
\ifx \shownote     \undefined \def \shownote      #1{#1}          \fi
\ifx \showarticletitle \undefined \def \showarticletitle #1{#1}   \fi
\ifx \showURL      \undefined \def \showURL       {\relax}        \fi
\providecommand\bibfield[2]{#2}
\providecommand\bibinfo[2]{#2}
\providecommand\natexlab[1]{#1}
\providecommand\showeprint[2][]{arXiv:#2}

\bibitem[Baumgart and Mies(2007)]%
        {baumgart2007skademlia}
\bibfield{author}{\bibinfo{person}{Ingmar Baumgart} {and}
  \bibinfo{person}{Sebastian Mies}.} \bibinfo{year}{2007}\natexlab{}.
\newblock \showarticletitle{{S/Kademlia}: A Practicable Approach Towards Secure
  Key-Based Routing}. In \bibinfo{booktitle}{\emph{Proceedings of the 13th
  International Conference on Parallel and Distributed Systems ({ICPADS})}}.
  \bibinfo{publisher}{IEEE}, \bibinfo{pages}{1--8}.
\newblock
\href{https://doi.org/10.1109/ICPADS.2007.4447808}{doi:\nolinkurl{10.1109/ICPADS.2007.4447808}}


\bibitem[Bell\'{e}s-Mu\~{n}oz et~al\mbox{.}(2023)]%
        {circom}
\bibfield{author}{\bibinfo{person}{Marta Bell\'{e}s-Mu\~{n}oz},
  \bibinfo{person}{Miguel Isabel}, \bibinfo{person}{Jose~Luis Mu\~{n}oz Tapia},
  \bibinfo{person}{Albert Rubio}, {and} \bibinfo{person}{Jordi Baylina}.}
  \bibinfo{year}{2023}\natexlab{}.
\newblock \showarticletitle{Circom: A Circuit Description Language for Building
  Zero-Knowledge Applications}.
\newblock \bibinfo{journal}{\emph{IEEE Transactions on Dependable and Secure
  Computing}} \bibinfo{volume}{20}, \bibinfo{number}{6} (\bibinfo{year}{2023}),
  \bibinfo{pages}{4733--4751}.
\newblock
\href{https://doi.org/10.1109/TDSC.2022.3232813}{doi:\nolinkurl{10.1109/TDSC.2022.3232813}}


\bibitem[Benet(2014)]%
        {benet2014ipfs}
\bibfield{author}{\bibinfo{person}{Juan Benet}.}
  \bibinfo{year}{2014}\natexlab{}.
\newblock \showarticletitle{{IPFS} -- Content Addressed, Versioned, {P2P} File
  System}.
\newblock \bibinfo{journal}{\emph{arXiv preprint arXiv:1407.3561}}
  (\bibinfo{year}{2014}).
\newblock


\bibitem[Biryukov et~al\mbox{.}(2014)]%
        {biryukov2014deanonymisation}
\bibfield{author}{\bibinfo{person}{Alex Biryukov}, \bibinfo{person}{Dmitry
  Khovratovich}, {and} \bibinfo{person}{Ivan Pustogarov}.}
  \bibinfo{year}{2014}\natexlab{}.
\newblock \showarticletitle{Deanonymisation of Clients in Bitcoin P2P Network}.
  In \bibinfo{booktitle}{\emph{Proceedings of the 2014 ACM SIGSAC Conference on
  Computer and Communications Security}} \emph{(\bibinfo{series}{CCS '14})}.
  \bibinfo{publisher}{ACM}, \bibinfo{pages}{15--29}.
\newblock
\href{https://doi.org/10.1145/2660267.2660379}{doi:\nolinkurl{10.1145/2660267.2660379}}


\bibitem[Bojja~Venkatakrishnan et~al\mbox{.}(2017)]%
        {venkatakrishnan2017dandelion}
\bibfield{author}{\bibinfo{person}{Shaileshh Bojja~Venkatakrishnan},
  \bibinfo{person}{Giulia Fanti}, {and} \bibinfo{person}{Pramod Viswanath}.}
  \bibinfo{year}{2017}\natexlab{}.
\newblock \showarticletitle{Dandelion: Redesigning the Bitcoin Network for
  Anonymity}.
\newblock \bibinfo{journal}{\emph{Proc. ACM Meas. Anal. Comput. Syst.}}
  \bibinfo{volume}{1}, \bibinfo{number}{1}, Article \bibinfo{articleno}{22}
  (\bibinfo{year}{2017}).
\newblock
\href{https://doi.org/10.1145/3084459}{doi:\nolinkurl{10.1145/3084459}}


\bibitem[Buterin and Griffith(2017)]%
        {buterin2017casper}
\bibfield{author}{\bibinfo{person}{Vitalik Buterin} {and}
  \bibinfo{person}{Virgil Griffith}.} \bibinfo{year}{2017}\natexlab{}.
\newblock \showarticletitle{Casper the Friendly Finality Gadget}.
\newblock \bibinfo{journal}{\emph{arXiv preprint arXiv:1710.09437}}
  (\bibinfo{year}{2017}).
\newblock
\urldef\tempurl%
\url{https://arxiv.org/abs/1710.09437}
\showURL{%
\tempurl}


\bibitem[Cholez and Ignat(2024)]%
        {cholez2024ipfs}
\bibfield{author}{\bibinfo{person}{Thibault Cholez} {and}
  \bibinfo{person}{Claudia-Lavinia Ignat}.} \bibinfo{year}{2024}\natexlab{}.
\newblock \showarticletitle{Sybil Attack Strikes Again: Denying Content Access
  in {IPFS} with a Single Computer}. In \bibinfo{booktitle}{\emph{Proc.\ 19th
  International Conference on Availability, Reliability and Security (ARES)}}.
\newblock
\href{https://doi.org/10.1145/3664476.3664482}{doi:\nolinkurl{10.1145/3664476.3664482}}


\bibitem[Coretti et~al\mbox{.}(2022)]%
        {coretti2022scuttlebutt}
\bibfield{author}{\bibinfo{person}{Sandro Coretti}, \bibinfo{person}{Aggelos
  Kiayias}, \bibinfo{person}{Cristopher Moore}, {and}
  \bibinfo{person}{Alexander Russell}.} \bibinfo{year}{2022}\natexlab{}.
\newblock \showarticletitle{The Generals' Scuttlebutt: Byzantine-Resilient
  Gossip Protocols}. In \bibinfo{booktitle}{\emph{Proceedings of the 2022 ACM
  SIGSAC Conference on Computer and Communications Security}}
  \emph{(\bibinfo{series}{CCS '22})}. \bibinfo{publisher}{ACM}.
\newblock
\href{https://doi.org/10.1145/3548606.3560638}{doi:\nolinkurl{10.1145/3548606.3560638}}


\bibitem[Delgado-Segura et~al\mbox{.}(2019)]%
        {delgadosegura2019txprobe}
\bibfield{author}{\bibinfo{person}{Sergi Delgado-Segura},
  \bibinfo{person}{Surya Bakshi}, \bibinfo{person}{Cristina
  P{\'e}rez-Sol{\`a}}, \bibinfo{person}{James Litton}, \bibinfo{person}{Andrew
  Pachulski}, \bibinfo{person}{Andrew Miller}, {and} \bibinfo{person}{Bobby
  Bhattacharjee}.} \bibinfo{year}{2019}\natexlab{}.
\newblock \showarticletitle{{TxProbe}: Discovering Bitcoin's Network Topology
  Using Orphan Transactions}. In \bibinfo{booktitle}{\emph{Financial
  Cryptography and Data Security}} \emph{(\bibinfo{series}{Lecture Notes in
  Computer Science}, Vol.~\bibinfo{volume}{11598})},
  \bibfield{editor}{\bibinfo{person}{Ian Goldberg} {and} \bibinfo{person}{Tyler
  Moore}} (Eds.). \bibinfo{publisher}{Springer}, \bibinfo{pages}{550--566}.
\newblock
\href{https://doi.org/10.1007/978-3-030-32101-7_32}{doi:\nolinkurl{10.1007/978-3-030-32101-7_32}}


\bibitem[Dingledine et~al\mbox{.}(2004)]%
        {dingledine2004tor}
\bibfield{author}{\bibinfo{person}{Roger Dingledine}, \bibinfo{person}{Nick
  Mathewson}, {and} \bibinfo{person}{Paul Syverson}.}
  \bibinfo{year}{2004}\natexlab{}.
\newblock \showarticletitle{Tor: The Second-Generation Onion Router}. In
  \bibinfo{booktitle}{\emph{Proc.\ 13th USENIX Security Symposium}}.
  \bibinfo{pages}{303--320}.
\newblock


\bibitem[Douceur(2002)]%
        {douceur2002sybil}
\bibfield{author}{\bibinfo{person}{John~R. Douceur}.}
  \bibinfo{year}{2002}\natexlab{}.
\newblock \showarticletitle{The {Sybil} Attack}. In
  \bibinfo{booktitle}{\emph{Peer-to-Peer Systems, First International Workshop,
  {IPTPS} 2002}} \emph{(\bibinfo{series}{Lecture Notes in Computer Science},
  Vol.~\bibinfo{volume}{2429})}. \bibinfo{publisher}{Springer},
  \bibinfo{pages}{251--260}.
\newblock
\href{https://doi.org/10.1007/3-540-45748-8_24}{doi:\nolinkurl{10.1007/3-540-45748-8_24}}


\bibitem[Dwork and Naor(1992)]%
        {dwork1992pricing}
\bibfield{author}{\bibinfo{person}{Cynthia Dwork} {and} \bibinfo{person}{Moni
  Naor}.} \bibinfo{year}{1992}\natexlab{}.
\newblock \showarticletitle{Pricing via Processing or Combatting Junk Mail}. In
  \bibinfo{booktitle}{\emph{Advances in Cryptology --- {CRYPTO} '92}}
  \emph{(\bibinfo{series}{Lecture Notes in Computer Science},
  Vol.~\bibinfo{volume}{740})}. \bibinfo{publisher}{Springer},
  \bibinfo{pages}{139--147}.
\newblock
\href{https://doi.org/10.1007/3-540-48071-4_10}{doi:\nolinkurl{10.1007/3-540-48071-4_10}}


\bibitem[{Ethereum Foundation}(2023)]%
        {geth}
\bibfield{author}{\bibinfo{person}{{Ethereum Foundation}}.}
  \bibinfo{year}{2023}\natexlab{}.
\newblock \bibinfo{title}{Go Ethereum (Geth)}.
\newblock
  \bibinfo{howpublished}{\url{https://github.com/ethereum/go-ethereum}}.
\newblock


\bibitem[Fanti et~al\mbox{.}(2018)]%
        {fanti2018dandelionpp}
\bibfield{author}{\bibinfo{person}{Giulia Fanti}, \bibinfo{person}{Shaileshh
  Bojja~Venkatakrishnan}, \bibinfo{person}{Surya Bakshi},
  \bibinfo{person}{Bradley Denby}, \bibinfo{person}{Shruti Bhargava},
  \bibinfo{person}{Andrew Miller}, {and} \bibinfo{person}{Pramod Viswanath}.}
  \bibinfo{year}{2018}\natexlab{}.
\newblock \showarticletitle{Dandelion++: Lightweight Cryptocurrency Networking
  with Formal Anonymity Guarantees}.
\newblock \bibinfo{journal}{\emph{Proc. ACM Meas. Anal. Comput. Syst.}}
  \bibinfo{volume}{2}, \bibinfo{number}{2}, Article \bibinfo{articleno}{29}
  (\bibinfo{year}{2018}).
\newblock
\href{https://doi.org/10.1145/3224424}{doi:\nolinkurl{10.1145/3224424}}


\bibitem[Grassi et~al\mbox{.}(2021)]%
        {poseidon}
\bibfield{author}{\bibinfo{person}{Lorenzo Grassi}, \bibinfo{person}{Dmitry
  Khovratovich}, \bibinfo{person}{Christian Rechberger}, \bibinfo{person}{Arnab
  Roy}, {and} \bibinfo{person}{Markus Schofnegger}.}
  \bibinfo{year}{2021}\natexlab{}.
\newblock \showarticletitle{Poseidon: A New Hash Function for {Zero-Knowledge}
  Proof Systems}. In \bibinfo{booktitle}{\emph{30th USENIX Security Symposium
  (USENIX Security 21)}}. \bibinfo{publisher}{USENIX Association},
  \bibinfo{pages}{519--535}.
\newblock
\showISBNx{978-1-939133-24-3}
\urldef\tempurl%
\url{https://www.usenix.org/conference/usenixsecurity21/presentation/grassi}
\showURL{%
\tempurl}


\bibitem[Groth(2016)]%
        {groth2016size}
\bibfield{author}{\bibinfo{person}{Jens Groth}.}
  \bibinfo{year}{2016}\natexlab{}.
\newblock \showarticletitle{On the size of pairing-based non-interactive
  arguments}. In \bibinfo{booktitle}{\emph{Annual International Conference on
  the Theory and Applications of Cryptographic Techniques}}. Springer,
  \bibinfo{pages}{305--326}.
\newblock


\bibitem[Heilman et~al\mbox{.}(2015)]%
        {heilman2015eclipsebitcoin}
\bibfield{author}{\bibinfo{person}{Ethan Heilman}, \bibinfo{person}{Alison
  Kendler}, \bibinfo{person}{Aviv Zohar}, {and} \bibinfo{person}{Sharon
  Goldberg}.} \bibinfo{year}{2015}\natexlab{}.
\newblock \showarticletitle{Eclipse attacks on Bitcoin's peer-to-peer network}.
  In \bibinfo{booktitle}{\emph{Proceedings of the 24th USENIX Conference on
  Security Symposium}} (Washington, D.C.) \emph{(\bibinfo{series}{SEC'15})}.
  \bibinfo{publisher}{USENIX Association}, \bibinfo{address}{USA},
  \bibinfo{pages}{129--144}.
\newblock
\showISBNx{9781931971232}


\bibitem[{iden3}(2018)]%
        {snarkjs}
\bibfield{author}{\bibinfo{person}{{iden3}}.} \bibinfo{year}{2018}\natexlab{}.
\newblock \bibinfo{title}{{snarkjs}}.
\newblock \bibinfo{howpublished}{\url{https://github.com/iden3/snarkjs}}.
\newblock


\bibitem[{iden3}(2022)]%
        {go-rapidsnark}
\bibfield{author}{\bibinfo{person}{{iden3}}.} \bibinfo{year}{2022}\natexlab{}.
\newblock \bibinfo{title}{{go-rapidsnark}}.
\newblock \bibinfo{howpublished}{\url{https://github.com/iden3/go-rapidsnark}}.
\newblock


\bibitem[Jelasity et~al\mbox{.}(2007)]%
        {jelasity2007peersampling}
\bibfield{author}{\bibinfo{person}{M\'{a}rk Jelasity}, \bibinfo{person}{Spyros
  Voulgaris}, \bibinfo{person}{Rachid Guerraoui}, \bibinfo{person}{Anne-Marie
  Kermarrec}, {and} \bibinfo{person}{Maarten van Steen}.}
  \bibinfo{year}{2007}\natexlab{}.
\newblock \showarticletitle{Gossip-based Peer Sampling}.
\newblock \bibinfo{journal}{\emph{ACM Transactions on Computer Systems}}
  \bibinfo{volume}{25}, \bibinfo{number}{3}, Article \bibinfo{articleno}{8}
  (\bibinfo{year}{2007}).
\newblock
\href{https://doi.org/10.1145/1275517.1275520}{doi:\nolinkurl{10.1145/1275517.1275520}}


\bibitem[Kiayias et~al\mbox{.}(2017)]%
        {kiayias2017ouroboros}
\bibfield{author}{\bibinfo{person}{Aggelos Kiayias}, \bibinfo{person}{Alexander
  Russell}, \bibinfo{person}{Bernardo David}, {and} \bibinfo{person}{Roman
  Oliynykov}.} \bibinfo{year}{2017}\natexlab{}.
\newblock \showarticletitle{Ouroboros: A Provably Secure Proof-of-Stake
  Blockchain Protocol}. In \bibinfo{booktitle}{\emph{Advances in Cryptology --
  CRYPTO 2017}} \emph{(\bibinfo{series}{Lecture Notes in Computer Science},
  Vol.~\bibinfo{volume}{10401})}. \bibinfo{publisher}{Springer},
  \bibinfo{pages}{357--388}.
\newblock
\href{https://doi.org/10.1007/978-3-319-63688-7_12}{doi:\nolinkurl{10.1007/978-3-319-63688-7_12}}


\bibitem[Kocher et~al\mbox{.}(2019)]%
        {kocher2019spectre}
\bibfield{author}{\bibinfo{person}{Paul Kocher}, \bibinfo{person}{Jann Horn},
  \bibinfo{person}{Anders Fogh}, \bibinfo{person}{Daniel Genkin},
  \bibinfo{person}{Daniel Gruss}, \bibinfo{person}{Werner Haas},
  \bibinfo{person}{Mike Hamburg}, \bibinfo{person}{Moritz Lipp},
  \bibinfo{person}{Stefan Mangard}, \bibinfo{person}{Thomas Prescher},
  \bibinfo{person}{Michael Schwarz}, {and} \bibinfo{person}{Yuval Yarom}.}
  \bibinfo{year}{2019}\natexlab{}.
\newblock \showarticletitle{Spectre Attacks: Exploiting Speculative Execution}.
  In \bibinfo{booktitle}{\emph{Proc.\ 40th IEEE Symposium on Security and
  Privacy (S\&P)}}. \bibinfo{pages}{1--19}.
\newblock


\bibitem[Leit\~{a}o et~al\mbox{.}(2007)]%
        {leitao2007hyparview}
\bibfield{author}{\bibinfo{person}{Jo\~{a}o Leit\~{a}o},
  \bibinfo{person}{Jos\'{e} Pereira}, {and} \bibinfo{person}{Lu\'{i}s
  Rodrigues}.} \bibinfo{year}{2007}\natexlab{}.
\newblock \showarticletitle{{HyParView}: A Membership Protocol for Reliable
  Gossip-Based Broadcast}. In \bibinfo{booktitle}{\emph{Proceedings of the 37th
  Annual IEEE/IFIP International Conference on Dependable Systems and Networks
  (DSN '07)}}. \bibinfo{publisher}{IEEE Computer Society},
  \bibinfo{address}{Edinburgh, UK}, \bibinfo{pages}{419--429}.
\newblock
\href{https://doi.org/10.1109/DSN.2007.56}{doi:\nolinkurl{10.1109/DSN.2007.56}}


\bibitem[Marcus et~al\mbox{.}(2018)]%
        {marcus2018eclipseethereum}
\bibfield{author}{\bibinfo{person}{Yuval Marcus}, \bibinfo{person}{Ethan
  Heilman}, {and} \bibinfo{person}{Sharon Goldberg}.}
  \bibinfo{year}{2018}\natexlab{}.
\newblock \bibinfo{title}{Low-Resource Eclipse Attacks on {Ethereum}'s
  Peer-to-Peer Network}.
\newblock \bibinfo{howpublished}{Cryptology ePrint Archive, Report 2018/236}.
\newblock
\urldef\tempurl%
\url{https://eprint.iacr.org/2018/236}
\showURL{%
\tempurl}


\bibitem[Maymounkov and Mazi\`{e}res(2002)]%
        {maymounkov2002kademlia}
\bibfield{author}{\bibinfo{person}{Petar Maymounkov} {and}
  \bibinfo{person}{David Mazi\`{e}res}.} \bibinfo{year}{2002}\natexlab{}.
\newblock \showarticletitle{Kademlia: A Peer-to-Peer Information System Based
  on the {XOR} Metric}. In \bibinfo{booktitle}{\emph{Revised Papers from the
  First International Workshop on Peer-to-Peer Systems (IPTPS)}}
  \emph{(\bibinfo{series}{Lecture Notes in Computer Science},
  Vol.~\bibinfo{volume}{2429})}. \bibinfo{publisher}{Springer},
  \bibinfo{address}{Berlin, Heidelberg}, \bibinfo{pages}{53--65}.
\newblock
\href{https://doi.org/10.1007/3-540-45748-8_5}{doi:\nolinkurl{10.1007/3-540-45748-8_5}}


\bibitem[{OffchainLabs}(2019)]%
        {prysm}
\bibfield{author}{\bibinfo{person}{{OffchainLabs}}.}
  \bibinfo{year}{2019}\natexlab{}.
\newblock \bibinfo{title}{{prysm: Go implementation of Ethereum proof of
  stake}}.
\newblock \bibinfo{howpublished}{\url{https://github.com/OffchainLabs/prysm/}}.
\newblock


\bibitem[{Protocol Labs}(2017)]%
        {protocollabs2017filecoin}
\bibfield{author}{\bibinfo{person}{{Protocol Labs}}.}
  \bibinfo{year}{2017}\natexlab{}.
\newblock \bibinfo{booktitle}{\emph{Filecoin: A Decentralized Storage
  Network}}.
\newblock \bibinfo{type}{{T}echnical {R}eport}.
\newblock
\urldef\tempurl%
\url{https://filecoin.io/filecoin.pdf}
\showURL{%
\tempurl}


\bibitem[Rowstron and Druschel(2001)]%
        {rowstron2001pastry}
\bibfield{author}{\bibinfo{person}{Antony Rowstron} {and}
  \bibinfo{person}{Peter Druschel}.} \bibinfo{year}{2001}\natexlab{}.
\newblock \showarticletitle{Pastry: Scalable, Decentralized Object Location,
  and Routing for Large-Scale Peer-to-Peer Systems}. In
  \bibinfo{booktitle}{\emph{Proceedings of the IFIP/ACM International
  Conference on Distributed Systems Platforms (Middleware)}}
  \emph{(\bibinfo{series}{Lecture Notes in Computer Science},
  Vol.~\bibinfo{volume}{2218})}. \bibinfo{publisher}{Springer},
  \bibinfo{address}{Berlin, Heidelberg}, \bibinfo{pages}{329--350}.
\newblock
\href{https://doi.org/10.1007/3-540-45518-3_18}{doi:\nolinkurl{10.1007/3-540-45518-3_18}}


\bibitem[Stoica et~al\mbox{.}(2001)]%
        {stoica2001chord}
\bibfield{author}{\bibinfo{person}{Ion Stoica}, \bibinfo{person}{Robert
  Morris}, \bibinfo{person}{David~R. Karger}, \bibinfo{person}{M.~Frans
  Kaashoek}, {and} \bibinfo{person}{Hari Balakrishnan}.}
  \bibinfo{year}{2001}\natexlab{}.
\newblock \showarticletitle{Chord: A Scalable Peer-to-Peer Lookup Service for
  Internet Applications}. In \bibinfo{booktitle}{\emph{Proceedings of the 2001
  Conference on Applications, Technologies, Architectures, and Protocols for
  Computer Communications}} \emph{(\bibinfo{series}{SIGCOMM '01})}.
  \bibinfo{publisher}{ACM}, \bibinfo{address}{New York, NY, USA},
  \bibinfo{pages}{149--160}.
\newblock
\href{https://doi.org/10.1145/383059.383071}{doi:\nolinkurl{10.1145/383059.383071}}


\bibitem[Van~Bulck et~al\mbox{.}(2018)]%
        {vanbulck2018foreshadow}
\bibfield{author}{\bibinfo{person}{Jo Van~Bulck}, \bibinfo{person}{Marina
  Minkin}, \bibinfo{person}{Ofir Weisse}, \bibinfo{person}{Daniel Genkin},
  \bibinfo{person}{Baris Kasikci}, \bibinfo{person}{Frank Piessens},
  \bibinfo{person}{Mark Silberstein}, \bibinfo{person}{Thomas~F. Wenisch},
  \bibinfo{person}{Yuval Yarom}, {and} \bibinfo{person}{Raoul Strackx}.}
  \bibinfo{year}{2018}\natexlab{}.
\newblock \showarticletitle{Foreshadow: Extracting the Keys to the {Intel SGX}
  Kingdom with Transient Out-of-Order Execution}. In
  \bibinfo{booktitle}{\emph{Proc.\ 27th USENIX Security Symposium}}.
  \bibinfo{pages}{991--1008}.
\newblock


\bibitem[Voulgaris et~al\mbox{.}(2005)]%
        {voulgaris2005cyclon}
\bibfield{author}{\bibinfo{person}{Spyros Voulgaris}, \bibinfo{person}{Daniela
  Gavidia}, {and} \bibinfo{person}{Maarten van Steen}.}
  \bibinfo{year}{2005}\natexlab{}.
\newblock \showarticletitle{{CYCLON}: Inexpensive Membership Management for
  Unstructured {P2P} Overlays}.
\newblock \bibinfo{journal}{\emph{Journal of Network and Systems Management}}
  \bibinfo{volume}{13}, \bibinfo{number}{2} (\bibinfo{year}{2005}),
  \bibinfo{pages}{197--217}.
\newblock
\href{https://doi.org/10.1007/s10922-005-4441-x}{doi:\nolinkurl{10.1007/s10922-005-4441-x}}


\bibitem[Winter et~al\mbox{.}(2016)]%
        {winter2016sybiltor}
\bibfield{author}{\bibinfo{person}{Philipp Winter}, \bibinfo{person}{Roya
  Ensafi}, \bibinfo{person}{Karsten Loesing}, {and} \bibinfo{person}{Nick
  Feamster}.} \bibinfo{year}{2016}\natexlab{}.
\newblock \showarticletitle{Identifying and Characterizing Sybils in the {Tor}
  Network}. In \bibinfo{booktitle}{\emph{Proc.\ 25th USENIX Security
  Symposium}}. \bibinfo{pages}{1169--1185}.
\newblock


\bibitem[Yu et~al\mbox{.}(2008)]%
        {yu2008sybillimit}
\bibfield{author}{\bibinfo{person}{Haifeng Yu}, \bibinfo{person}{Phillip~B.
  Gibbons}, \bibinfo{person}{Michael Kaminsky}, {and} \bibinfo{person}{Feng
  Xiao}.} \bibinfo{year}{2008}\natexlab{}.
\newblock \showarticletitle{{SybilLimit}: A Near-Optimal Social Network Defense
  against {Sybil} Attacks}. In \bibinfo{booktitle}{\emph{Proceedings of the
  2008 {IEEE} Symposium on Security and Privacy ({S\&P})}}.
  \bibinfo{publisher}{IEEE}, \bibinfo{pages}{3--17}.
\newblock
\href{https://doi.org/10.1109/SP.2008.13}{doi:\nolinkurl{10.1109/SP.2008.13}}


\bibitem[Yu et~al\mbox{.}(2006)]%
        {yu2006sybilguard}
\bibfield{author}{\bibinfo{person}{Haifeng Yu}, \bibinfo{person}{Michael
  Kaminsky}, \bibinfo{person}{Phillip~B. Gibbons}, {and}
  \bibinfo{person}{Abraham Flaxman}.} \bibinfo{year}{2006}\natexlab{}.
\newblock \showarticletitle{{SybilGuard}: Defending Against {Sybil} Attacks via
  Social Networks}. In \bibinfo{booktitle}{\emph{Proceedings of the 2006
  Conference on Applications, Technologies, Architectures, and Protocols for
  Computer Communications ({ACM} {SIGCOMM})}}. \bibinfo{publisher}{ACM},
  \bibinfo{pages}{267--278}.
\newblock
\href{https://doi.org/10.1145/1159913.1159945}{doi:\nolinkurl{10.1145/1159913.1159945}}


\end{thebibliography}

\appendix
\crefalias{section}{appendix}
\numberwithin{theorem}{section}
\numberwithin{lemma}{section}
\numberwithin{corollary}{section}
\numberwithin{definition}{section}
\numberwithin{proposition}{section}
\section{Deferred Proofs}
\label{appendix:proofs}

\subsection{Partition Detection Proofs}

\begin{lemma}[Restatement of \cref{lem:attack-soundness}]
  \stmtAttackSoundness
\end{lemma}
\begin{proof}
	Suppose the network is healthy and the adversary is not attacking. By
	\cref{def:network_states}, for all honest nodes $i$, $|\mathcal{R}_i| \ge
  \gamma n$ and there is no partition. The protocol
  (\alglines{alg:heartbeat}{line:hb-send-loop}{line:hb-send})
  samples from at least
	$\gamma n$ nodes. Let $X$ be the number of unique honest peer records received
	by node $i$ in round $r$. Then, $X$ is bounded below by a binomial
	distribution with parameters $\gamma n$ and $s / \sqrt{n}$. Hence,
	\(
		\Pr[X \le \theta s\sqrt{n}]
		\leq \Pr_{Y \sim \text{Binomial}(\gamma n, s/\sqrt{n})}
              [Y \leq \theta s\sqrt{n}]
		= \Pr_{Y \sim \text{Binomial}(\gamma n, s/\sqrt{n})}
    [Y \leq (1 - (1-\theta/\gamma)) \gamma s\sqrt{n}]
	\).
  Applying a standard Chernoff bound, we obtain
  $\exp\left(-\frac{s\sqrt{n}(\gamma - \theta)^2}{2\gamma}\right)$.
\end{proof}

\begin{lemma}[Restatement of \cref{lem:partition-size}]
  \stmtPartitionSize
\end{lemma}
\begin{proof}
  Assume, for sake of
  contradiction, that $|V_A| > \varphi n$. Then, since we assumed $|V_A| \le
  |V_B|$, we have $|V_B| \ge \varphi n$ as well. Thus, $|V_A| + |V_B| > 2
  \varphi n = (1 + \alpha) n$. However, $|V_A| + |V_B| = |A| + |B| + 2
  \alpha n = (1 + \alpha) n$ which is a contradiction. Therefore, $|V_A| \leq
  \varphi n$.
\end{proof}

\begin{lemma}[Restatement of \cref{lem:attack-completeness}]
  \stmtAttackCompleteness
\end{lemma}
\begin{proof}
	Suppose the network is partitioned. Without loss of generality, let node $i$
  be in the smaller partition $A$ with $|A| \le \varphi n$ by
  \cref{lem:partition-size}. Let $X$ be the number of unique peer records
  received by node $i$ in round $r$
  via the heartbeat (\cref{alg:heartbeat}).
  Since the number of honest peer records
  reachable by node $i$ is bounded above by $\varphi n$ and each is sampled with
  probability $s / \sqrt{n}$, $X$ is bounded above by a binomial distribution with
  parameters $\varphi n$ and $s / \sqrt{n}$. Hence,
  $\Pr[X \le \theta s\sqrt{n} \mid \mathcal{N}_{partitioned}] \geq \Pr_{Y \sim
  \text{Binomial}(\varphi n, s/\sqrt{n})}[Y \leq \theta s\sqrt{n}] = \Pr_{Y \sim
  \text{Binomial}(\varphi n, s/\sqrt{n})}[Y \leq (1 - (1-\theta/\varphi)) \varphi s\sqrt{n}]$.
  Applying a standard Chernoff bound yields $1 -
  \exp\left(-\frac{s\sqrt{n}(\theta - \varphi)^2}{\theta + \varphi}\right)$.
\end{proof}

\subsection{Overlay Resistance Proof}

\paragraph{Parameter choice.}
The threshold $\theta$ is derived by dividing the interval $[\varphi, 1]$ into
three equal slices of width $\varepsilon = (1-\alpha)/6$, yielding
$\theta = \varphi + 2\varepsilon = (5+\alpha)/6$.
\Cref{fig:overlay-three-slice-line} illustrates the construction.
\usetikzlibrary{calc,decorations.pathreplacing,arrows.meta}

\colorlet{axisgrayV}{black}
\colorlet{thetaCV}{blue!70!black}
\colorlet{slackCV}{orange!70!black}
\colorlet{flagCV}{red!70!black}
\colorlet{textgrayV}{black}

\tikzset{
  AxisV/.style = {line width=0.9pt, draw=axisgrayV},
  TickV/.style = {line width=0.8pt, draw=axisgrayV},
  DotV/.style n args = {1}{circle, fill=#1, draw=white, line width=0.35pt, inner sep=1.45pt},
  CallV/.style = {-{Stealth[length=2.0mm,width=1.25mm]}, line width=0.75pt, shorten >=1pt},
  BraceV/.style args = {#1}{semithick, decorate,
      decoration={brace,#1,raise=1.7pt,amplitude=2.6pt,
      pre=moveto,pre length=1pt,post=moveto,post length=1pt}},
  labV/.style = {font=\small, inner sep=0.4pt},
  slabV/.style = {font=\footnotesize, inner sep=0.4pt},
  ysV/.style = {yshift=#1}
}
\begin{figure}[h]
    \centering
\begin{tikzpicture}[every node/.style={align=center}]
  \coordinate (Z)     at (0,0);
  \coordinate (Theta) at (31mm,0);
  \coordinate (H)     at (48.5mm,0);
  \coordinate (T)     at (66mm,0);
  \coordinate (One)   at (83.5mm,0);

  \def\barH{1.8mm}   %
  \def\wedge{1.8mm}  %
  \colorlet{sliceC}{blue!25}
  \fill[sliceC]
    (Theta) --
    ($(Theta)+(\wedge,\barH)$) -- ($(H)+(-\wedge,\barH)$) --
    (H) --
    ($(H)+(-\wedge,-\barH)$) -- ($(Theta)+(\wedge,-\barH)$) -- cycle;
  \fill[sliceC]
    (H) --
    ($(H)+(\wedge,\barH)$) -- ($(T)+(-\wedge,\barH)$) --
    (T) --
    ($(T)+(-\wedge,-\barH)$) -- ($(H)+(\wedge,-\barH)$) -- cycle;
  \fill[sliceC]
    (T) --
    ($(T)+(\wedge,\barH)$) -- ($(One)+(-\wedge,\barH)$) --
    (One) --
    ($(One)+(-\wedge,-\barH)$) -- ($(T)+(\wedge,-\barH)$) -- cycle;

  \coordinate (gapL) at ($(Z)!0.38!(Theta)$);
  \coordinate (gapR) at ($(Z)!0.62!(Theta)$);
  \draw[AxisV] (Z) -- (gapL);
  \draw[AxisV] (gapR) -- (One);

  \node[font=\normalsize, text=axisgrayV, inner sep=0pt] at ($(Z)!0.5!(Theta)$) {$\cdots$};

  \foreach \p in {Z,Theta,H,T,One} {
    \draw[TickV] ([ysV=1.5mm]\p) -- ([ysV=-1.5mm]\p);
  }

  \node[DotV={black}] at (Theta) {};
  \node[DotV={black}] at (H) {};
  \node[DotV={black}]  at (T) {};

  \node[labV, text=black, below=7pt] at (Z) {$0$};
  \node[labV, text=black, below=7pt] at (Theta) {$\varphi$};
  \node[labV, text=black, below=7pt] at (H) {$\varphi+\varepsilon$};
  \node[labV, text=black, below=7pt] at (T) {$\theta$};
  \node[labV, text=black, below=7pt] at (One) {$1$};

  \node[labV, text=black] (thlab) at ([ysV=7.6mm]Theta)
    {$\tfrac{3+3\alpha}{6}$};
  \draw[CallV, draw=black] ([ysV=-0.6pt]thlab.south) -- ([ysV=0.6mm]Theta);

  \node[labV, text=black] (hlab) at ([ysV=7.6mm]H)
    {$\tfrac{4+2\alpha}{6}$};
  \draw[CallV, draw=black] ([ysV=-0.6pt]hlab.south) -- ([ysV=0.6mm]H);

  \node[labV, text=black] (tlab) at ([ysV=7.6mm]T)
    {$\tfrac{5+\alpha}{6}$};
  \draw[CallV, draw=black] ([ysV=-0.6pt]tlab.south) -- ([ysV=0.6mm]T);

  \node[slabV, text=blue!60!black] at ($([ysV=3.8mm]Theta)!0.5!([ysV=3.8mm]H)$) {$\varepsilon$};
  \node[slabV, text=blue!60!black] at ($([ysV=3.8mm]H)!0.5!([ysV=3.8mm]T)$) {$\varepsilon$};
  \node[slabV, text=blue!60!black] at ($([ysV=3.8mm]T)!0.5!([ysV=3.8mm]One)$) {$\varepsilon$};
\end{tikzpicture}
    \caption{
     Chosen parameters for the analysis. For the largest cut, the
     non-$B$ contribution to a node's table is centered at
     $\varphi=(1+\alpha)/2$. The remaining interval $[\varphi,1]$ has width
     $(1-\alpha)/2$ and is split into three equal slices of width
     $\varepsilon=(1-\alpha)/6$. This places the detection threshold at
     $\theta=\varphi+2\varepsilon=(5+\alpha)/6$.}
    \label{fig:overlay-three-slice-line}
\end{figure}
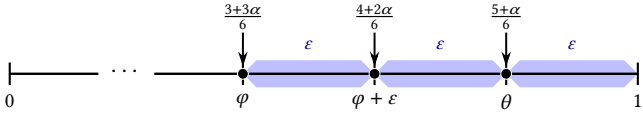

\paragraph{Proof setup.}
We now prove the relevant lemmas and theorems for overlay resistance.
The following definitions and notations will be used in the proofs.
We write $\flag(V)$ for the number of nodes in $V \subseteq [n]$ that raise
the attack-detection flag.
We quantify over all cuts $(A,B)$ of the honest nodes where $A$ is
the smaller side of the cut:
\[
A,B \subseteq \Honset,
\quad A \cup B = \Honset,
\quad A \cap B = \emptyset,
\quad 1 \leq |A| \leq \lfloor |\Honset|/2 \rfloor
\]
\paragraph{Parameters.}
In addition to our protocol parameters $n, s, \theta$ and the adversary stake
$\alpha$ there are several security parameters and notations that we will be
using in our analysis:
\begin{itemize}
  \item We use $X_{i,A} \sim \text{Binomial}(|A| + \alpha n, s/\sqrt{n})$ to
  denote the random variable representing the number of entries in $i$'s
  $\otable$ that come from nodes in $A$ and the adversary. We use $X_{i,B}$ to
  denote the number of entries in $i$'s $\otable$ that come from nodes in $B$.
  \item We say a node $i \in A$ is $A$-heavy if $X_{i,A} \geq
  (\theta-\varepsilon)s\sqrt{n}$, meaning that $i$ has many entries from nodes
  in $A$ and the adversary, and thus relatively fewer entries from nodes in $B$.
  \item $\delta$ is the overlay tolerance parameter: whenever the overlay is
  disconnected across $(A,B)$, at least a $(1-\delta)$-fraction of nodes in $A$
  raise the attack-detection flag. In the good case, we allow at most a
  $(\delta/2)$-fraction of nodes in $A$ to be $A$-heavy.
  \item $\varepsilon$ is the slack that we define the good events with such that
  we get at least $\varepsilon s \sqrt{n}$ entries on the peer tables of nodes
  in $A$ that come from nodes in $B$.
\end{itemize}

\paragraph{Local events.}
We now define cut specific events.
\begin{itemize}
  \item $\nf(A)_{\delta} := \{ \flag(A) \leq (1-\delta) |A| \}$ is the event
  that at most a $(1-\delta)$-fraction of nodes in $A$ raise the
  attack-detection flag. Equivalently, at least a $\delta$-fraction of nodes
  in $A$ failed to raise the flag.
  \item $\dc(A)$ is the event that no node in $A$ has a connection to any node
  in $B$ in the sampled overlay.
  \item $\hv(A)_{\delta} := \{|\{i \in A \mid X_{i,A} \geq
  (\theta-\varepsilon)s\sqrt{n} \}| \geq \frac{\delta}{2} |A| \}$ is the event that at least
  a $(\delta/2)$-fraction of nodes in $A$ are $A$-heavy.
\end{itemize}

\paragraph{Global events.}
Now we can define the global events that quantify over the cuts of the honest
nodes.
\begin{itemize}
  \item $\HV_{\delta} := \exists\,{(A,B)}\; \hv(A)_{\delta}$ is the event
  that for some cut $(A,B)$ of the honest nodes, at least a $(\delta/2)$-fraction
  of nodes in $A$ are $A$-heavy.
  \item $\Bad_{\delta, \varepsilon} := \exists\,{(A,B)}\;
  (\nf(A)_{\delta} \land \dc(A))$ is the event that there exists a cut $(A,B)$
  of the honest nodes such that at least a $\delta$-fraction of nodes in $A$
  fail to raise the flag, and no node in $A$ has a connection to any node in $B$
  in the sampled overlay.
\end{itemize}
We have defined $\Bad := \Bad_{\delta, \varepsilon}$ to be the event
main event of interest that captures the adversary's success in partitioning the
honest nodes in the overlay without raising the attack-detection flag for most
nodes in $A$. We show that $\Pr[\Bad]$ is negligible for reasonable choices of
$\delta$ and $\varepsilon$.

Our proof proceeds by first showing that $\Pr[\HV_{\delta}]$ is negligible. We
then condition on $\neg \HV_{\delta}$ and show that $\Pr[\Bad]$ is negligible
as well. The key intuition is that if $\neg \HV_{\delta}$ holds, then for any
cut $(A,B)$ of the honest nodes, at least a $(1-\delta/2)$-fraction of nodes in
$A$ are not $A$-heavy, meaning that they have at least $\varepsilon s\sqrt{n}$
entries from nodes in $B$. Hence, each such node has a good chance of connecting
to $B$ in the overlay, and the probability that none of them connect to $B$ is
negligible.

\paragraph{Most nodes are not heavy.}
We first show that $\Pr[\HV_{\delta}]$ is negligible. Our first step is to
show that for a single node $i\in A$, being $A$-heavy is unlikely. Then, most of
$A$ is not $A$-heavy with high probability, and thus $\Pr[\HV_{\delta}]$ is
negligible as well after a union bound over the cuts of the honest nodes.
\begin{lemma}\label{lem:heavy-node} For $\varepsilon > 0$, any cut $(A,B)$ of
  the honest nodes, and any node $i \in A$, let $\Delta_A = (\theta - \varepsilon) - (\alpha +
  |A|/n)$. If $\Delta_A > 0$, then
  \[
    \Pr[i \in A \text{ is $A$-heavy}]
    \leq
    \exp\left(-
      \frac{\Delta_A^2 s \sqrt{n}}
      {2\left(\alpha + \frac{|A|}{n}\right) + \Delta_A}
    \right)
  \]
\end{lemma}
\begin{proof}
  Let $i \in A$ and consider $X_{i,A} \sim \text{Binomial}(|A| + \alpha n,
  s/\sqrt{n})$. We know that $\E[X_{i,A}] = s\sqrt{n}(\alpha + \frac{|A|}{n})$.
  Since $(\theta-\varepsilon) = \Delta_As\sqrt{n} + \E[X_{i,A}]$, by applying a
  Chernoff bound, we can bound the probability that $X_{i,A}$ is greater than
  $(\theta-\varepsilon)s\sqrt{n}$ as follows:
  \begin{align*}
    \Pr[X_{i,A} \geq (\theta-\varepsilon)s\sqrt{n}]
    &\leq \exp\left(
      -\frac{(\Delta_A s \sqrt{n})^2}
      {2\E[X_{i,A}] + \Delta_A s \sqrt{n}}
    \right) \\
    &=
    \exp\left(
      -\frac{\Delta_A^2 s \sqrt{n}}
      {2\left(\alpha + \frac{|A|}{n}\right) + \Delta_A}
    \right)
  \end{align*}
  which proves our claim.
\end{proof}

\begin{lemma} \label{lem:heavy} For any $\delta > 0$, assume $0 < \varepsilon
  < \theta - \frac{1+\alpha}{2}$. Then,
  \[
    \Pr[\HV_{\delta}]
    \leq
    \exp \left(
      |\Honset| \left(
        \frac{2ex}{\delta}
      \right)^{\delta/2}
     \right) - 1, \qquad    x = \exp\left(
      \frac{-\varepsilon^2 s \sqrt{n}}{1+\alpha+\varepsilon}
    \right)
  \]
\end{lemma}
\begin{proof}
  We first show that the bound on $i \in A$ being $A$-heavy from
  \cref{lem:heavy-node} gives us a bound on the probability that for a
  given cut $(A,B)$, $\hv(A)_{\delta}$ holds. Fix a cut $(A,B)$ with $|A|\le
  \lfloor |\Honset|/2\rfloor$.
  Since $\alpha + \frac{|A|}{n}\le \frac{1+\alpha}{2}$, we
  have $\Delta_A \ge \varepsilon$, and the Chernoff bound in
  \cref{lem:heavy-node} implies that for each fixed $i\in A$,
  \[
    \Pr[i \in A \text{ is $A$-heavy}] \le x.
  \]
  Let $Y$ be the number of $A$-heavy nodes in $A$. Then $Y$ is stochastically
  dominated by $\text{Binomial}(|A|,x)$. Using the tail bound $\Pr[Y \geq \frac{\delta}{2}
  k] \le \left(\frac{e \E[Y]}{\frac{\delta}{2} k}\right)^{\frac{\delta}{2} k}$, we have
  \begin{align*}
    \Pr[\hv(A)_{\delta}]
    &\leq
    \Pr_{Y \sim \text{Binomial}(|A|, x)}\!\left[Y \geq \tfrac{\delta}{2} k\right]
    \\
    &\leq
    \left(
      \frac{2e \E[Y]}{\delta |A|}
    \right)^{\frac{\delta}{2} |A|}
    =
    \left(
      \frac{2e x}{\delta}
    \right)^{\frac{\delta}{2} |A|}
  \end{align*}
  We proceed by a union bound over the cuts of the honest nodes:
  \begin{align*}
    \Pr[\HV_{\delta}]
    &\leq \sum_{(A,B)} \Pr[\hv(A)_{\delta}] \\
    &\leq
    \sum_{k=1}^{\lfloor |\Honset|/2 \rfloor} \binom{|\Honset|}{k}
    \left(
      \frac{2e x}{\delta}
    \right)^{\frac{\delta}{2} k} \\
    &\leq
    \sum_{k=1}^{|\Honset|}
     \binom{|\Honset|}{k}
    \left(
      \frac{2e x}{\delta}
    \right)^{\frac{\delta}{2} k} \\
    &\leq
    \left(1 + \left(
      \frac{2e x}{\delta}
    \right)^{\delta/2}\right)^{|\Honset|} - 1 \\
    &\leq
    \exp\left(
      |\Honset| \left(
        \frac{2e x}{\delta}
      \right)^{\delta/2}
     \right) - 1
  \end{align*}
\end{proof}

\paragraph{Disconnection without flags is unlikely.}
Now that we showed that $\Pr[\HV_{\delta}]$ is negligible, we can condition on
$\neg \HV_{\delta}$ and show that $\Pr[\Bad]$ is negligible as well. We
directly show that for any cut $(A,B)$ of the honest nodes, there must either be
many nodes in $A$ that raise the attack-detection flag, or many nodes in $A$
that have good chances of connecting to $B$.
\begin{lemma} \label{lem:bad-heavy} Assume $\delta > 0$. Then,
  \[
    \Pr[\Bad \mid \neg \HV_{\delta}]
    \leq
    \exp(|\Honset|e^{-\frac{\delta}{2}\pconn\varepsilon s \sqrt{n}}) - 1
  \]
\end{lemma}
\begin{proof}
  Consider the event $\Bad$ and condition on $\neg \HV_{\delta}$. Then, for
  any cut $(A,B)$ of the honest nodes, we have $\neg \hv(A)_{\delta}$, which
  means that at most a $(\delta/2)$-fraction of nodes in $A$ are $A$-heavy. Hence,
  at most $\frac{\delta}{2} |A|$ nodes $i \in A$ satisfy:
  \[
        X_{i,A} \geq (\theta-\varepsilon)s\sqrt{n}
  \]
  We also know that since $\Bad$ holds, there exists some cut $(A,B)$ of honest
  nodes such that $\nf(A)_{\delta}$ and $\dc(A)$ hold. Since
  $\nf(A)_{\delta}$ holds, at least a $\delta$-fraction of nodes in $A$ fail
  to raise the attack-detection flag:
  \[
      X_{i,A} + X_{i,B} \geq \theta s \sqrt{n}
  \]
  Combining the two inequalities above, at least a $(\delta/2)$-fraction of nodes in
  $A$ are not $A$-heavy and fail to raise the flag, which means that there are
  at least $\frac{\delta}{2}|A|$ nodes in $A$ that satisfy:
  \[
    \begin{cases}
      X_{i,A} < (\theta-\varepsilon)s\sqrt{n} \\
      X_{i,B} \geq \theta s \sqrt{n} - X_{i,A}
    \end{cases}
  \]
  Hence, for at least a $(\delta/2)$-fraction of nodes in $A$, we
  have
  \[
    X_{i,B} \geq \varepsilon s \sqrt{n}
  \]
  Each of these entries become a connection to $B$ in the overlay with
  probability $\pconn$. Hence,
  \[
    \Pr[i \not\to_{\text{overlay}} B]
    \leq (1 - \pconn)^{\varepsilon s \sqrt{n}}
    \leq e^{-\pconn\varepsilon s \sqrt{n}}
  \]
  and therefore we have,
  \begin{align*}
    \Pr[\dc(A) \land \nf(A)_{\delta} &\mid \neg \hv(A)_{\delta}]\\
    &\leq
    \exp\!\left(-\tfrac{\delta}{2}\,p\varepsilon s \sqrt{n}|A|\right)
  \end{align*}
  Finally, we can take a union bound over the cuts of the honest nodes to get
  \begin{align*}
    \Pr[\Bad \mid &\neg \HV_{\delta}] \\
    &\leq \sum_{(A,B)} \Pr[\dc(A) \land \nf(A)_{\delta}
    \mid \neg \hv(A)_{\delta}] \\
    &=
    \sum_{k=1}^{\lfloor |\Honset|/2 \rfloor} \binom{|\Honset|}{k}
    \exp\!\left(-\tfrac{\delta}{2}\,p\varepsilon s \sqrt{n} k\right) \\
    &\leq
    \sum_{k=1}^{|\Honset|} \binom{|\Honset|}{k}
    \exp\!\left(-\tfrac{\delta}{2}\,p\varepsilon s \sqrt{n} k\right) \\
    &\leq
    (1 + e^{-\frac{\delta}{2}\pconn\varepsilon s \sqrt{n}})^{|\Honset|} - 1 \\
    &\leq
    \exp(|\Honset|e^{-\frac{\delta}{2}\pconn\varepsilon s \sqrt{n}}) - 1
  \end{align*}
\end{proof}

\begin{theorem}[Restatement of \cref{thm:overlay}]
  \stmtOverlay
\end{theorem}
\begin{proof}
  We first define the bad event that the adversary successfully partitions the
  honest nodes in the overlay without raising the attack-detection flag for most
  nodes in the smaller side of the cut:
  \[
    \Bad := \exists\,(A,B)\; (\nf(A)_\delta \land \dc(A)),
  \]
  where $\nf(A)_\delta$ is the event that at most a $(1-\delta)$-fraction of
  nodes in $A$ raise the attack-detection flag, and $\dc(A)$ is the event that
  no node in $A$ has a connection to any node in $B$ in the sampled overlay.
  Throughout the proof, we use the parameter selection in
  \cref{fig:overlay-three-slice-line}.

  Let $\rho$ be defined as:
  \[
    \rho = \left( \frac{2e x}{\delta}\right)^{\delta/2} \qquad \text{where}
    \quad x = \exp\left(
      \frac{-\varepsilon^2 s \sqrt{n}}{1+\alpha+\varepsilon}
      \right).
  \]

  By the law of total probability, we have
  \begin{align*}
    \Pr[\Bad]
    &=
    \Pr[\Bad \mid \HV_{\delta}] \Pr[\HV_{\delta}]\\
    &\qquad\qquad+
    \Pr[\Bad \mid \neg \HV_{\delta}] \Pr[\neg \HV_{\delta}] \\
    &\leq
    \Pr[\HV_{\delta}]
    +
    \Pr[\Bad \mid \neg \HV_{\delta}]
  \end{align*}
  Using the bounds from \cref{lem:heavy,lem:bad-heavy}, we
  can further bound this as
  \[
    \Pr[\Bad]
    \leq
    \left(
      \exp\left(|\Honset|\rho\right) - 1
    \right)
    +
    \left(
      \exp\left(|\Honset|e^{-\frac{\delta}{2}\pconn\varepsilon s \sqrt{n}}\right) - 1
    \right).
  \]

  Now observe that for $c = (\varepsilon^2 s)/(1+\alpha+\varepsilon)$
  \[
    \ln\rho
    =
    \frac{\delta}{2}\!\left(1+\ln 2+\ln x-\ln\delta\right)
    =
    \frac{\delta}{2}\!\left(1 + \ln 2 - c \sqrt{n} - \ln\delta\right).
  \]
  Since $\delta$ is a constant, we have $\ln\rho = -\Theta(\sqrt n)$,
  and hence $\rho = e^{-\Theta(\sqrt n)}$.
  We then get $|\Honset|\rho \le n e^{-\Theta(\sqrt n)} = e^{-\Theta(\sqrt n)}$.
  In particular, using $e^x - 1 = x + O(x^2)$, we obtain
  \[
    \exp\left(|\Honset|\rho\right)-1
    =
    |\Honset|\rho + O\left((|\Honset|\rho)^2\right)
    =
    e^{-\Theta(\sqrt n)}.
  \]

  For the second term, set $c = \frac{\delta\varepsilon C}{2}$ where $C$ is the
  constant from the degree assumption, so that
  $\frac{\delta}{2}\pconn\varepsilon s\sqrt{n} \geq c n^\kappa \log n$.
  By assumption $C > 2/(\delta\varepsilon)$, hence $c > 1$.
  Therefore
  \[
    |\Honset|e^{-\frac{\delta}{2}\pconn\varepsilon s\sqrt n}
    \le
    n e^{-c\, n^\kappa\log n}
    =
    n^{1-c\, n^\kappa}.
  \]
  When $\kappa > 0$ we have $c\, n^\kappa \to \infty$, so
  $n^{1-c\, n^\kappa} = n^{-\Theta(n^\kappa)}$.
  When $\kappa = 0$ the exponent is $1 - c < 0$, so
  $n^{1-c} = n^{-\Theta(1)}$.
  In both cases applying $e^x - 1 = x + O(x^2)$ again gives
  \[
    \exp\left(|\Honset|e^{-\frac{\delta}{2}\pconn\varepsilon s\sqrt n}\right)-1
    =
    n^{-\Theta(n^\kappa)}.
  \]
  Combining the two bounds, we get:
  \[
    \Pr[\Bad] \le e^{-\Theta(\sqrt n)} + n^{-\Theta(n^\kappa)}.
  \]
\end{proof}

\subsection{Slashing Proofs}

\begin{lemma}[Restatement of \cref{lem:stake-secret-recovery}]
  \stmtStakeSecretRecovery
\end{lemma}
\begin{proof}
  Suppose a node $i$ violates the slashing condition in some round $r$ by
  contacting $s\sqrt{n}+1$ distinct peers. Let $i$'s stake credentials be derived from
  secret $\sk_i$, resulting in $\stakeID_i$, $\stakesk_i$, and $\netpk_i$.

  We first argue that if a node violates the slashing condition, then that node
  must have sent two requests with distinct commitments. Each request sent by
  $i$ includes a commitment $C$
  (\algline{alg:heartbeat}{line:hb-commit})
  to a receiver set of size at most $s\sqrt{n}$ and a
  opening proving that the receiver is a member of $C$. By the pigeonhole
  principle, since $i$ sent $s\sqrt{n}+1$ requests to distinct peers, each with a
  commitment of size at most $s\sqrt{n}$, at least two of these requests must contain
  distinct commitments $C_A \neq C_B$.

  Next, we argue that two distinct commitments reveal the stake secret. The
  arithmetic below is over the finite field $\mathbb{F}_q$ for a large prime $q$
  (where hashes are encoded as field elements). For round $r$, the protocol
  derives the per-round slope
  \[
    a_r \gets \hashshare(\sk_i, r) \in \mathbb{F}_q
  \]
  Each request contains a share value constructed (\algline{alg:heartbeat}{line:hb-share}),
  whose well-formedness is verified by each responder
  (\algline{alg:send_response}{line:resp-verify-share}), as
  \[
    \slashshare_X = a_r \cdot C_X + \stakesk_i \pmod q
    \qquad \text{for } X \in \{A,B\}
  \]
  which is a function of the commitment $C$. Thus, the two commitments $C_A \neq
  C_B$ induce two distinct points on the line $P(x) = a_r x + \stakesk_i$:
  \[
    (x_A,y_A) = (C_A,\slashshare_A)
    \quad\text{and}\quad
    (x_B,y_B) = (C_B,\slashshare_B)
  \]
  Hence, we can recover the intercept $\stakesk_i$ by interpolation:
  \[
    a_r = \frac{y_A-y_B}{x_A-x_B},
    \qquad
    \stakesk_i = y_A - a_r x_A = P(0).
  \]
  Therefore, $\stakesk_i$ can be efficiently recovered from the two requests.
\end{proof}

\begin{lemma}[Restatement of \cref{lem:slashing-soundness}]
  \stmtSlashingSoundness
\end{lemma}
\begin{proof}
  Let $\stakeID = \hashid(\stakesk)$ be the node's stake identifier, and
  $(\netsk,\netpk)$ its network keypair all derived from the secret key $\sk$.
  We first observe that an honest node that sends at most $s\sqrt{n}$ requests in round
  $r$ reuses the same exact commitment $C$,
  slash share $\slashshare$,
  and proof
  $\shareproof$ for all its requests
  (\alglines{alg:heartbeat}{line:hb-commit}{line:hb-commitrecord};
  see also \cref{sec:too-many-reqs}).

  Note that slashing requires learning the stake secret $\stakesk$ of the node.
  Hence, consider the adversary's view in round $r$. The values in the protocol
  messages that depend on $\stakesk$ are: (a) the stake identifier
  $\stakeID=\hashid(\stakesk)$, (b) the public key $\netpk$ derived from $\sk$,
  and (c) the commitment/proof pair $(C,\shareproof)$ included in the requests.
  Therefore, one of the following must occur with non-negligible probability for
  the adversary to slash node $i$:
  \begin{itemize}
      \item The adversary inverts the hash computation $\hashid$ on input
      $\stakeID$.
      \item The adversary derives the secret $\sk$ from $\netpk$ by inverting
      key-generation.
      \item The adversary extracts the witness from the zero-knowledge proof
      $\shareproof$ included in the requests.
  \end{itemize}
  As we assumed that hash functions are random oracles, key generation
  is one-way, and the proof system is zero-knowledge, each of these
  events occurs with at most negligible probability $\negl(\lambda)$,
  and the claim follows by a union bound.
\end{proof}

\begin{corollary}[Restatement of \cref{lem:slashing-completeness}]
  \stmtSlashingCompleteness
\end{corollary}
\begin{proof}
  Suppose node $i$ violates the slashing condition in round $r$ by sending
  requests to more than $s\sqrt{n}$ distinct peers. By
  \cref{lem:stake-secret-recovery}, there exist two requests sent by $i$ in
  round $r$ with distinct commitments $C_A \neq C_B$ whose contents suffice
  to recover $\stakesk_i$.

  It remains to show that some honest node observes both commitments. Since
  $i$ must distribute $s\sqrt{n}+1$ requests across its peers, by the pigeonhole
  principle at least two honest nodes $j_1, j_2$ receive requests from $i$
  with $j_1$ receiving a request under $C_A$ and $j_2$ under $C_B$ (or a
  single node receives both). Each honest node stores the commitment record
  $(\commitrecord, \roundnum)$ inside the peer record for $i$
  (\algline{alg:send_response}{line:resp-recv}). Because the
  network is healthy ($\gamma$-healthy), the peer record of $i$ propagates.
  By the mean-field analysis (\cref{sec:mean-field}),
  an honest node's record reaches a $\Theta(s/\sqrt{n})$-fraction of honest nodes
  each round.  After the record-merging step
  (\algline{alg:receive_new_record}{line:rec-merge}), any honest node that receives peer records
  containing both $C_A$ and $C_B$ for the same $(\netpk_i, \roundnum)$
  detects the duplicate
  (\algline{alg:receive_new_record}{line:rec-dup-detect})
  and constructs a \slashproof. Since each of the two
  commitments is stored by at least one honest node and gossip propagates
  records with probability $\Theta(s/\sqrt{n})$ per hop, after one round each
  version is held by at least $\gamma s\sqrt{n}$ honest nodes. As these are
  random subsets of the $(1-\alpha)n$ honest nodes, the probability that they
  are disjoint is at most
  $\exp\left(-\frac{\gamma^2}{1-\alpha}\,s^2\right)
  = \exp(-\Omega(s^2))$. Once a \slashproof is constructed, any node can recover $\stakesk_i$
  (\algline{alg:receive_new_record}{line:rec-recover})
  and call $\textsc{Slash}(\stakesk_i, \stakeID_i)$ on-chain
  (\algline{alg:smartcontract}{line:sc-slash}).
\end{proof}

\subsection{Privacy Proofs}

\begin{lemma}[Restatement of \cref{lem:stake-anon}]
  \stmtStakeAnon
\end{lemma}
\begin{proof}[Proof (sketch)]
  Consider the transcript of messages sent by node $i$ in round $r$. Observe
  that the zero-knowledge property allows us to define an indistinguishable
  game, except with probability $\negl(\lambda)$ to the real execution where the
  challenger simulates the proof $\shareproof$ without knowing the witness
  $\stakesk_i$. Moreover, since the node sends at most $s\sqrt{n}$ requests, the
  commitments $C$ and shares $\slashshare$ are identical across all requests
  (\alglines{alg:heartbeat}{line:hb-commit}{line:hb-commitrecord}),
  and they do not leak any information about $\stakesk_i$
  (\cref{lem:slashing-soundness}). Therefore, the adversary's view is
  independent of the stake identity $\stakeID_i$ of node $i$, and the claim
  follows.
\end{proof}

\begin{lemma}[Restatement of \cref{lem:conn-privacy}]
  \stmtConnPrivacy
\end{lemma}
\begin{proof}
While transmitting full tables, the node that $i$ requests from cannot tell
which subset of entries $i$ will retain
(the private seed $\nonceoverlay$ controls insertion into $\otable$;
\algline{alg:receive_new_record}{line:rec-insert-otable})
via the messages that they exchange.
Consider the scenario in which $b=0$ and $b=1$. In both cases, the outgoing
requests from node $i$ are identical (they only request the full table of peers)
and the gossip tables of all other nodes remain identical as well.
Therefore, the entire view of $\mathcal{A}$ (the request transcript and all
tables except $i$'s newly constructed table) is identically distributed under
$b=0$ and $b=1$. Hence $\Pr[b'=b]=\frac{1}{2}$ and $\cadv = 0$.
\end{proof}

\section{Deferred Pseudocode}
\label{appendix:pseudocode}

This appendix collects the pseudocode listings deferred from \cref{sec:protocol}
for space.

\paragraph{Smart contract.}
\Cref{alg:smartcontract} details the on-chain staking contract
introduced in \cref{sec:stake-contract}.

\begin{breakablealgorithm}
\caption{\protocolname Smart Contract}
\label{alg:smartcontract}
\begin{algorithmic}[1]
    \Function{DepositAndStake}{\stakeID, funds, stakeOwner} \label{line:sc-deposit}
        \State \require $funds = 1$
        \State \require $owner[\stakeID] = \bot$
        \State \require $timeNow < epochEndTime - \stakefreeze$ \label{line:sc-freeze}
        \State add \stakeID to committed vector
        \State $deposits[\stakeID] \gets 1$
        \State $owner[\stakeID] \gets stakeOwner$
    \EndFunction

    \Statex \vspace{-4pt}

    \Statex \hspace{\algorithmicindent}\Comment{Request removal from staking}
    \Function{UnStake}{\stakeID} \label{line:sc-unstake}
        \State \require $deposits[\stakeID]$
        \State \require $sender = owner[\stakeID]$
        \State \require $timeNow < epochEndTime - \stakefreeze$
        \State $withdrawalTime[\stakeID] \gets currentEpoch$
        \State $deposits[\stakeID] \gets 0$
        \State remove \stakeID from committed vector
    \EndFunction

    \Statex \vspace{-4pt}

    \Statex \hspace{\algorithmicindent}\Comment{Withdraw un-staked funds}
    \Function{ClaimFunds}{$\stakeID$} \label{line:sc-claim}
        \State \require $sender = owner[\stakeID]$
        \State \require $0 < withdrawalTime[\stakeID] < currentEpoch$
        \State \require \stakewithdraw time passed since start of epoch \label{line:sc-withdraw-delay}
        \State $withdrawalTime[\stakeID] \gets 0$
        \State $owner[\stakeID] \gets \bot$
        \State send $StakeUnit \to sender$
    \EndFunction

    \Statex \vspace{-4pt}

    \Function{Slash}{$\stakesk,\stakeID$} \label{line:sc-slash}
        \State \require $\stakeID = \hashid(\stakesk)$
        \State \require $deposits[\stakeID]=1$
        \Statex \hspace{\algorithmicindent}\hspace{\algorithmicindent}$\lor\ withdrawalTime[\stakeID]>0$
        \State $deposits[\stakeID] \gets 0$
        \State remove \stakeID from committed vector
        \State $withdrawalTime[\stakeID] \gets 0$
        \State $owner[\stakeID] \gets \bot$
    \EndFunction

    \Statex \vspace{-4pt}

    \Function{GetProof}{$\stakeID$} \label{line:sc-getproof}
        \State \require $deposits[\stakeID]=1$
        \State \Return $\vecopen(\mathsf{aux}, \stakeID)$
    \EndFunction

    \Statex \vspace{-4pt}

    \Function{GetCommitment}{\ } \label{line:sc-getcommit}
        \State \Return the current \stakeroot
    \EndFunction
\end{algorithmic}
\end{breakablealgorithm}

\paragraph{Heartbeat and response.}
\Cref{alg:heartbeat} details the heartbeat procedure executed once per round,
and \cref{alg:send_response} describes how a node responds to incoming requests.

\begin{breakablealgorithm}
\caption{Heartbeat (executed once per round)}
\label{alg:heartbeat}
\begin{algorithmic}[1]
\Procedure{Heartbeat}{$\ptable, \roundnum$}

\Statex \hspace{\algorithmicindent}\Comment{Create a fresh \netrecord}
    \State $\timestamp \gets$ \Call{GetCurrentTime}{{}}
    \State $(\stakeroot, aux, self.\ind) \gets$ stake commitment
    \Statex \hspace{\algorithmicindent}\hspace{\algorithmicindent}from blockchain (\cref{alg:smartcontract})
    \State $\sig \gets \Sign(self.\netsk, \langle self.\addr, \timestamp \rangle)$
    \State $VecP \gets \vecopen(aux,self.\ind)$
    \State $\stakeproof \gets \Call{$\text{ZKP}_{\Rstake}$}{}$ \label{line:hb-stakeproof}
    \Statex \hspace{\algorithmicindent}\hspace{\algorithmicindent}$\bigl( (self.\netpk,\stakeroot),$
    \Statex \hspace{\algorithmicindent}\hspace{\algorithmicindent}$\phantom{\bigl(}(self.\sk, self.\stakeID, VecP, self.\ind) \bigr)$
    \State $\netrecord \gets \langle self.\netpk, \stakeroot,$ \label{line:hb-netrecord}
    \Statex \hspace{\algorithmicindent}\hspace{\algorithmicindent}$\stakeproof, \langle self.\addr, \timestamp \rangle_\sig \rangle$

    \Statex \vspace{-4pt}

    \Statex \hspace{\algorithmicindent}\Comment{Sample and sign slice seeds}
    \State $\nonce, \nonceoverlay \xleftarrow{\$} \{0,1\}^{\lambda}$ \label{line:hb-seeds}
    \State $\roundnum \gets \Call{GetCurrentRound}{{}}$
    \State $\sig \gets \Sign(\netsk, \langle \nonce, \nonceoverlay, \roundnum \rangle)$

    \Statex \vspace{-4pt}

    \Statex \hspace{\algorithmicindent}\Comment{Commit to peers and generate slashing share}
    \State $\reqcommit, ReqAux \gets$ \label{line:hb-commit}
    \Statex \hspace{\algorithmicindent}\hspace{\algorithmicindent}$\veccommit\bigl(\bigl[\ptable[\ind].\netpk \bigr]_{\ind=1}^{s\sqrt{n}}\bigr)$
    \State $\slashshare \gets \hashshare(\sk, \roundnum)\cdot \reqcommit + \stakesk$ \label{line:hb-share}
    \State $\shareproof \gets \Call{$\text{ZKP}_{\Rshare}$}{}$
    \Statex \hspace{\algorithmicindent}\hspace{\algorithmicindent}$\bigl( (\netpk,\reqcommit, \slashshare,\roundnum), (\sk,\stakesk) \bigr)$
    \State $\commitrecord \gets \langle \reqcommit, \slashshare, \shareproof \rangle$ \label{line:hb-commitrecord}

    \Statex \vspace{-4pt}

    \Statex \hspace{\algorithmicindent}\Comment{Send requests}
    \For{$\ind \in \{1 \ldots s\sqrt{n}\}$} \label{line:hb-send-loop}
        \State $\indproof \gets \vecopen(ReqAux,\ind)$
        \State $\request \gets \langle \langle \nonce, \nonceoverlay, \roundnum \rangle_\sig,$
        \Statex \hspace{\algorithmicindent}\hspace{\algorithmicindent}\hspace{\algorithmicindent}$\commitrecord, \langle \ind, \indproof \rangle, \netrecord \rangle$
        \State send $\request$ to $\ptable[\ind].\addr$ \label{line:hb-send}
    \EndFor
\EndProcedure
\end{algorithmic}
\end{breakablealgorithm}

The ZK relations \Rshare (well-formed shares) and \Rstake (proof-of-stake) are
defined formally in \cref{appendix:primitives}. Informally, \Rstake certifies
that a node's $\netpk$ is backed by stake in the current commitment, while
\Rshare binds each per-round share to the node's secret key, enabling slashing
upon equivocation (\cref{sec:too-many-reqs}).

\begin{breakablealgorithm}
\caption{Respond Upon Request}
\label{alg:send_response}
\begin{algorithmic}[1]

\Procedure{UponRequest}{\request}

    \Statex \hspace{\algorithmicindent}\Comment{Unpack request fields}
    \State $\langle \langle \nonce, \nonceoverlay, \roundnum \rangle_\sig, \commitrecord,$
    \Statex \hspace{\algorithmicindent}\hspace{\algorithmicindent}$\langle \ind, \indproof \rangle, \netrecord \rangle \gets \request$
    \State $\langle \netpk, \stakeroot, \stakeproof, \langle \addr, \timestamp \rangle_\sig \rangle \gets \netrecord$
    \State $\langle \reqcommit, \slashshare, \shareproof \rangle \gets \commitrecord$

    \Statex \vspace{-4pt}

    \Statex \hspace{\algorithmicindent}\Comment{Validate sender and store}
    \State $\peerrecord \gets \langle \netrecord, [(\commitrecord,\roundnum)]\rangle$
    \State \require \Call{ReceiveNewRecord}{$\peerrecord$} \label{line:resp-recv}

    \Statex \vspace{-4pt}

    \Statex \hspace{\algorithmicindent}\Comment{Verify freshness and authenticity}
    \State \require $\roundnum = \Call{GetCurrentRound}{{}}$
    \State \require $\CheckSig(\netpk, \langle \nonce, \nonceoverlay, \roundnum \rangle_\sig)$
    \State \require $\vecverify (\reqcommit, \ind, self.\netpk, {\indproof})$ \label{line:resp-verify-opening}
    \State \require $\ind \in \{1 \ldots s\sqrt{n}\}$
    \State \require \Call{$\text{ZKVerify}_{\Rshare}$}{} \label{line:resp-verify-share}
    \Statex \hspace{\algorithmicindent}\hspace{\algorithmicindent}${(\reqcommit, \slashshare, \netpk, \roundnum),\shareproof}$

    \Statex \vspace{-4pt}
    \Statex \hspace{\algorithmicindent}\Comment{Rate limiting}
    \State \require $\netpk \notin RecentReqs[\roundnum]$ \label{line:resp-rate-limit}
    \State \require $\netpk \notin \denylist$ \label{line:resp-denylist}
    \State $RecentReqs[\roundnum] \gets RecentReqs[\roundnum] \cup \{\netpk\}$

    \Statex \vspace{-4pt}

    \Statex \hspace{\algorithmicindent}\Comment{Respond with selected peers}
    \State $Peers \gets \{rec \in \ptable : \PRNG[\nonce][rec.\netpk]<\frac{s}{\sqrt{n}}\}$ \label{line:filter-start}
    \State $Peers \gets Peers \cup \{rec \in \ptable : \PRNG[\nonceoverlay][rec.\netpk]<\frac{s}{\sqrt{n}}\}$ \label{line:filter-end}
    \State {\bf send} $\langle Peers, \denylist \rangle$ to \addr
\EndProcedure

\end{algorithmic}
\end{breakablealgorithm}

\paragraph{Protocol messages and records.}
\Cref{fig:protocol-messages-and-records} summarizes the data structures
and message formats used throughout the protocol.

\begin{figure*}[!ht]
\fbox{%
  \begin{minipage}{0.95\textwidth}
  \small
  \renewcommand{\arraystretch}{1.4}
  \begin{tabular}{@{}p{0.47\textwidth}@{\hspace{8pt}}p{0.46\textwidth}@{}}
    \multicolumn{2}{@{}l}{\textbf{Records}} \\[2pt]
    $\netrecord = \langle \netpk, \stakeroot, \stakeproof, \langle \addr, \timestamp\rangle_\sig \rangle$
      & Network identity with proof of stake and a signed address. \\
    $\peerrecord = \langle \netrecord, [(\commitrecord_i,\roundnum_i)]_{i=1}^m \rangle$
      & A peer's identity together with its recent per-round commitments. \\
    $\commitrecord = \langle \reqcommit, \slashshare, \shareproof \rangle$
      & Per-round commitment binding the set of requests to a slashing share (with ZK proof of correctness). \\[4pt]
    \multicolumn{2}{@{}l}{\textbf{Messages}} \\[2pt]
    $\request = \langle \langle \nonce,\nonceoverlay, \roundnum \rangle_\sig, \commitrecord, \langle \ind, \indproof \rangle, \netrecord \rangle$
      & Signed seeds for slice selection, the sender's commitment to all its requests, proof of inclusion at position $\ind$, and sender's identity. \\
    $\response = \langle[\peerrecord_i]_{i=1}^k, [\slashproof_j]_{j=1}^l \rangle$
      & Heartbeat response: a list of peer records selected by the requester's seeds, plus any slashing proofs to deny access to recently slashed peers. \\[4pt]
    \multicolumn{2}{@{}l}{\textbf{Accountability}} \\[2pt]
    $\slashproof = \langle \netpk, \stakeroot, \roundnum,$
    $\commitrecord_1, \commitrecord_2 \rangle$
      & Evidence of equivocation: two distinct commitments by the same peer in a single round. \\
  \end{tabular}

  \vspace{4pt}
  \footnotesize{Notation: $[\cdot]$ denotes a list; $\langle\cdot\rangle_\sig$ marks fields covered by a signature.}
  \end{minipage}
}
\caption{\label{fig:protocol-messages-and-records}Messages and records
used by \protocolname.}
\end{figure*}

\paragraph{Record validation.}
\Cref{alg:receive_new_record} validates incoming peer records, merges them with
existing entries, and detects double-commitment violations that trigger slashing.

\begin{breakablealgorithm}
\caption{Receive New Record}
\label{alg:receive_new_record}
\begin{algorithmic}[1]

\Procedure{ReceiveNewRecord}{$\peerrecord$}
    \Statex \hspace{\algorithmicindent}\Comment{Validate and store; returns $\mathit{True}$ on success}
    \Statex \hspace{\algorithmicindent}\Comment{Unpack and validate}
    \State $\langle \netrecord, [(CR_i,\roundnum_i)]_{i=1}^m \rangle \gets \peerrecord$
    \State $\langle \netpk, \stakeroot, \stakeproof, \langle \addr, \timestamp\rangle_\sig \rangle \gets \netrecord$
    \State \require $\CheckSig(\netpk, \langle \addr, \timestamp \rangle_\sig)$
    \State \require \Call{$\text{ZKVerify}_{\Rstake}$}{$(\netpk, \stakeroot),\stakeproof$} \label{line:rec-verify-stake}
    \State \require $\stakeroot \in \mathit{AccComs}$ \Comment{set by \cref{alg:epoch_transition}} \label{line:rec-stakeroot}
    \State \require $\Call{GetCurrentTime}{} - \timestamp \leq \recordexp$ \label{line:rec-expiry}

    \Statex \vspace{-4pt}

    \Statex \hspace{\algorithmicindent}\Comment{Merge with pre-existing records}
    \If{$\netpk \in \ptable \lor \netpk \in \otable$} \label{line:rec-merge}
        \State $\peerrecord \gets \peerrecord \uplus \ptable[\netpk] \uplus \otable[\netpk]$
        \Statex \hspace{\algorithmicindent}\hspace{\algorithmicindent}\Comment{keep newest \netrecord; union $[(CR_i,\roundnum_i)]$}%
    \EndIf

    \Statex \vspace{-4pt}

    \Statex \hspace{\algorithmicindent}\Comment{Detect two distinct commitments in same round}
    \State $\mathit{dup} \gets \exists i \neq j:\ \roundnum_i = \roundnum_j$%
    \label{line:rec-dup-detect}
    \Statex \hspace{\algorithmicindent}\hspace{\algorithmicindent}$\land\ CR_i.\reqcommit \neq CR_j.\reqcommit$
    \If{$\mathit{dup}$}
        \State $\slashproof \gets \langle \netpk, \stakeroot, \roundnum_i, CR_i, CR_j \rangle$
        \Statex \hspace{\algorithmicindent}\hspace{\algorithmicindent}\Comment{Recover \stakesk and slash on-chain}
        \State $a \gets \dfrac{CR_i.\slashshare - CR_j.\slashshare}%
            {CR_i.\reqcommit - CR_j.\reqcommit}$
        \State $\stakesk \gets CR_i.\slashshare - a \cdot CR_i.\reqcommit$ \label{line:rec-recover}
        \State $\stakeID \gets \hashid(\stakesk)$
        \State submit $\textsc{Slash}(\stakesk, \stakeID)$ on-chain \label{line:rec-submit-slash}
        \Statex \hspace{\algorithmicindent}\hspace{\algorithmicindent}(\cref{alg:smartcontract})
        \State $\denylist \gets \denylist \cup \{\slashproof\}$
    \EndIf

    \Statex \vspace{-4pt}

    \Statex \hspace{\algorithmicindent}\Comment{Insert into tables}
    \If{$\PRNG[\nonce][\netpk] < \frac{s}{\sqrt{n}} \lor \netpk \in \ptable$} \label{line:rec-insert-ptable}
        \State $\ptable[\netpk] \gets \peerrecord$
    \EndIf
    \If{$\PRNG[\nonceoverlay][\netpk] < \frac{s}{\sqrt{n}} \lor \netpk \in \otable$} \label{line:rec-insert-otable}
        \State $\otable[\netpk] \gets \peerrecord$
    \EndIf
    \State \Return $\mathit{True}$
\EndProcedure

\end{algorithmic}
\end{breakablealgorithm}

\paragraph{Response handling.}
\Cref{alg:handle_response} processes heartbeat responses, and
\cref{alg:verify_slash} verifies slash proofs carried in those responses.

\begin{breakablealgorithm}
\caption{Handle response}
\label{alg:handle_response}
\begin{algorithmic}[1]

\Procedure{UponResponse}{\response}

    \Statex \hspace{\algorithmicindent}\Comment{Process and verify slash proofs}
    \State $\langle Recs, SlashProofs \rangle \gets \response$
    \For{$sp \in SlashProofs$}
        \State \require \Call{VerifySlashProof}{$sp$} (\cref{alg:verify_slash}) \label{line:handle-verify-slash}
    \EndFor
    \State $\denylist \gets \denylist \cup SlashProofs$

    \Statex \vspace{-4pt}

    \Statex \hspace{\algorithmicindent}\Comment{Validate and store received records}
    \For{$record \in Recs$}
        \State \Call{ReceiveNewRecord}{record} (\cref{alg:receive_new_record})
    \EndFor

    \Statex \vspace{-4pt}

    \Statex \hspace{\algorithmicindent}\Comment{Evict expired records}
    \State $\ptable \gets \{rec \in \ptable : \Call{GetCurrentTime}{} - rec.\timestamp \leq \recordexp\}$ \label{line:handle-evict-expired}
    \State $\otable \gets \{rec \in \otable : \Call{GetCurrentTime}{} - rec.\timestamp \leq \recordexp\}$

    \Statex \vspace{-4pt}

    \Statex \hspace{\algorithmicindent}\Comment{Evict excess entries}
    \While{$|\ptable|>(1+\epsilon) \cdot s\sqrt{n}$}
        \State $evictRec \gets \displaystyle \argmax_{\netpk\in \ptable}\PRNG[\nonce][\netpk]$
        \State $\ptable \gets \ptable \setminus \ptable[evictRec]$
    \EndWhile
    \While{$|\otable|>(1+\epsilon) \cdot s\sqrt{n}$}
        \State $evictRec \gets \displaystyle \argmax_{\netpk\in \otable}\PRNG[\nonceoverlay][\netpk]$
        \State $\otable \gets \otable \setminus \otable[evictRec]$
    \EndWhile

\EndProcedure
\end{algorithmic}
\end{breakablealgorithm}

\begin{breakablealgorithm}
\caption{Verify Slash Proof}
\label{alg:verify_slash}
\begin{algorithmic}[1]
\Function{VerifySlashProof}{$\slashproof$}
    \Statex \hspace{\algorithmicindent}\Comment{Unpack and verify both commit records}
    \State $\langle \netpk, \stakeroot, \roundnum, \commitrecord_1, \commitrecord_2 \rangle \gets \slashproof$
    \State $\langle \reqcommit_1, \slashshare_1, \pi_{shr,1} \rangle \gets \commitrecord_1$
    \State $\langle \reqcommit_2, \slashshare_2, \pi_{shr,2} \rangle \gets \commitrecord_2$

    \Statex \vspace{-4pt}

    \State \require $\stakeroot \in \mathit{AccComs}$
    \State \require $\reqcommit_1 \neq \reqcommit_2$ \label{line:vs-distinct}
    \State \require \Call{$\text{ZKVerify}_{\Rshare}$}{}
    \Statex \hspace{\algorithmicindent}\hspace{\algorithmicindent}${(\reqcommit_1, \slashshare_1, \netpk, \roundnum), \pi_{shr,1}}$
    \State \require \Call{$\text{ZKVerify}_{\Rshare}$}{}
    \Statex \hspace{\algorithmicindent}\hspace{\algorithmicindent}${(\reqcommit_2, \slashshare_2, \netpk, \roundnum), \pi_{shr,2}}$
    \State \Return $\mathit{True}$
\EndFunction
\end{algorithmic}
\end{breakablealgorithm}

\paragraph{Epoch transitions.}
\Cref{alg:epoch_transition} is executed at the start of each new epoch
to update the accepted commitment set and refresh the node's stake proof.

\begin{breakablealgorithm}
\caption{Epoch Transition (triggered at the start of each new epoch)}
\label{alg:epoch_transition}
\begin{algorithmic}[1]
\Procedure{EpochTransition}{}

    \Statex \hspace{\algorithmicindent}\Comment{Fetch new commitment from blockchain}
    \State $(\stakeroot_{\mathit{new}}, \mathit{aux}, self.\ind) \gets$ new epoch
    \Statex \hspace{\algorithmicindent}\hspace{\algorithmicindent}commitment from blockchain
    \State $\mathit{AccComs} \gets$ last $d$ epoch commitments \label{line:et-acccoms}
    \Statex \hspace{\algorithmicindent}\hspace{\algorithmicindent}incl.\ $\stakeroot_{\mathit{new}}$

    \Statex \vspace{-4pt}

    \Statex \hspace{\algorithmicindent}\Comment{Regenerate own stake proof for new epoch}
    \State $VecP \gets \vecopen(\mathit{aux}, self.\ind)$
    \State $self.\stakeproof \gets \Call{$\text{ZKP}_{\Rstake}$}{}$
    \Statex \hspace{\algorithmicindent}\hspace{\algorithmicindent}$\bigl( (self.\netpk, \stakeroot_{\mathit{new}}),$
    \Statex \hspace{\algorithmicindent}\hspace{\algorithmicindent}$\phantom{\bigl(}(self.\sk, self.\stakeID, VecP, self.\ind) \bigr)$
    \State $self.\stakeroot \gets \stakeroot_{\mathit{new}}$

    \Statex \vspace{-4pt}

    \Statex \hspace{\algorithmicindent}\Comment{Purge denylist entries whose epoch left $\mathit{AccComs}$}
    \State $\denylist \gets \{sp \in \denylist : sp.\stakeroot \in \mathit{AccComs}\}$ \label{line:et-purge-deny}

    \Statex \vspace{-4pt}

    \Statex \hspace{\algorithmicindent}\Comment{Clear stale rate-limiting state}
    \State Clear $\mathit{RecentReqs}$ for rounds in previous epochs

\EndProcedure
\end{algorithmic}
\end{breakablealgorithm}

\paragraph{Seed privacy via TEE (details).}
For clarity, the pseudocode in \cref{alg:heartbeat,alg:send_response} shows the
seeds $(\nonce,\nonceoverlay)$ sent to the responder in the clear. In practice,
the response computation
(\alglines{alg:send_response}{line:filter-start}{line:filter-end}) runs inside a
trusted execution environment (e.g., an SGX or TDX enclave) on the responder
side: the seeds are decrypted only within the enclave and never exposed to the
responder's application logic. The enclave returns its response encrypted and
padded to a fixed length, so the responder learns neither the seeds nor which
records matched. This also prevents the responder from probing the enclave by
replaying the computation against different peer tables, thereby preserving the
privacy of the requester's slice selection without altering protocol logic.
Without TEEs, deployments may choose between two extremes: sending seeds in the
clear retains the communication saving (responses of size $s^2$) but reveals the
requester's slice, while omitting seeds from the request preserves privacy at
the cost of full $s\sqrt{n}$-sized responses.

\paragraph{Slashing derivation.}
\label{sec:slashing-derivation}
The share value (\algline{alg:heartbeat}{line:hb-share}) is
\[
\slashshare \;=\; \hashshare(\sk,\roundnum)\cdot \reqcommit \;+\; \stakesk
\qquad\text{over a field } \mathbb{F},
\]
together with a ZK proof of well-formedness and consistency with the
requester's $\netpk$.
Let $a=\hashshare(\sk,\allowbreak\roundnum)\in\mathbb{F}$.
If a node sends two requests in the same round under distinct
commitments $\reqcommit_1\neq \reqcommit_2$, observers obtain two
valid pairs $(\reqcommit_1,\allowbreak\slashshare_1)$ and
$(\reqcommit_2,\allowbreak\slashshare_2)$ satisfying:
\begin{align*}
\slashshare_1 &= a\cdot \reqcommit_1 + \stakesk,\\
\slashshare_2 &= a\cdot \reqcommit_2 + \stakesk.
\end{align*}
Subtracting yields
$\slashshare_1-\slashshare_2 = a\cdot(\reqcommit_1-\reqcommit_2)$.
Since $\reqcommit_1\neq \reqcommit_2$, the difference
$(\reqcommit_1-\reqcommit_2)$ is invertible with overwhelming probability under
the encoding used for commitments.\footnote{E.g., if $\reqcommit$ is a hash
interpreted as a field element, collisions are negligible.} One therefore recovers
$a = (\slashshare_1-\slashshare_2)/(\reqcommit_1-\reqcommit_2)$
and then
$\stakesk = \slashshare_1 - a\cdot \reqcommit_1$.
Anyone can then compute $\stakeID=\hashid(\stakesk)$ and call
$\textsc{Slash}(\stakesk,\stakeID)$ on-chain
(\algline{alg:smartcontract}{line:sc-slash}). Hence, \emph{a single detected
double-commit in one round} suffices to reveal the slashing pre-image and burn
the offender's stake.

\section{Prototype Implementation Details}
\label{appendix:prototype}

\subsection{Smart contract}
\label{sec:smart_contract}
We implement the staking contract (\cref{alg:smartcontract}) in Solidity,
supporting three operations: \textbf{deposit} (stake funds against a $\stakeID$),
\textbf{withdraw} (reclaim funds after a freeze period), and
\textbf{slash} (destroy stake by presenting a recovered $\stakesk$).
Any participant can read the current $\stakeroot$ from the contract for use in
zero-knowledge proofs.

\subsection{Zero-knowledge proofs}
\label{sec:zkp}
We implement \stakeproof and \shareproof as Circom~\cite{circom} circuits
using Baby Jubjub keys (chosen to optimize circuit constraint cost) and
Poseidon~\cite{poseidon} hashes to derive \stakesk and \stakeID from \sk.
We use Groth16~\cite{groth2016size} via snarkjs~\cite{snarkjs} for
its computational efficiency and proof succinctness.

\subsection{P2P client}
We build our prototype into the Ethereum ecosystem, which provides a
realistic deployment environment: permissionless P2P networking, native
staking, a capable VM for our contract, and a modular libp2p stack.

\textit{Client architecture:}
We fork Prysm~\cite{prysm}, a Go Ethereum consensus client, and register
\protocolname as an additional libp2p protocol.
The heartbeat (\cref{alg:heartbeat}) runs as a periodic task within
Prysm's \texttt{Sync} component; each round spawns Go routines that
distribute \request messages randomly across the round window to smooth load.
The client integrates \texttt{rapidsnark}~\cite{go-rapidsnark} for ZKP
operations and communicates with the staking contract via JSON-RPC
using \texttt{abigen}-generated bindings~\cite{geth}.
Prysm's native peer-discovery is replaced by \protocolname.

\textit{Networking layer:}
We extend libp2p with Baby Jubjub key support so that the
$\netsk$--$\netpk$ pair can be used directly for networking.
The prototype retains secp256k1 keys alongside Baby Jubjub for
interoperability with unmodified Ethereum clients.

\textit{Protocol messages:}
\protocolname messages use Protocol Buffers with SSZ encoding and Snappy
compression, matching Prysm's existing serialization pipeline.

\subsection{Detailed experimental results}

\textit{Setup details:}
Experiments run on Linux kernel~5.15.
The staking contract is preloaded into the genesis block and wallets are
prefunded so that nodes can deposit stake at initialization;
client connectivity is bootstrapped via a script that exchanges
$\netrecord$ objects among nodes (in a real deployment, nodes would use
conventional mechanisms such as hardcoded bootstrap peers).
Clients are kept idle (no injected transactions) to isolate protocol
overhead.
We configure a stake proof Merkle tree depth of~32, use a single gossip
table ($T_{\mathit{gsp}}$), and do not employ trusted execution
environments for private record retrieval, measuring the raw cost of the
protocol mechanisms.
To avoid underestimating protocol costs from caching effects, we
benchmark individual functions and estimate conservative execution
times based on the expected number of invocations per round.

\textit{Instrumentation:}
We collect per-container runtime metrics using \texttt{cAdvisor},
aggregate them with Prometheus, and visualize through Grafana dashboards.
For fine-grained CPU breakdowns, we profile the client using Go's
\texttt{pprof} tool.

\textit{Per-round CPU utilization:}
\Cref{fig:cpu_period} shows per-round CPU utilization.
Usage peaks at round start (proof generation, commitment preparation)
then stabilizes as nodes process incoming requests and responses.
Compared to an idle Ethereum client, the additional overhead remains modest.

\begin{figure}
    \centering
    \begin{tikzpicture}
    \begin{axis}[
        awplot,
        width=0.95\linewidth,
        height=\awplotheight,
        xlabel={Time in AetherWeave Round (sec)},
        ylabel={Average CPU utilization (\%)},
        xmin=0, xmax=120,
        ymin=10,
        grid=major,
        legend columns=3,
        legend style={at={(0.5,0.02)}, anchor=south, fill=white, fill opacity=0.8, text opacity=1}
    ]

    \addplot[plotA, thick, name path=base_mean] table [x=seconds, y=periodic_mean, col sep=comma] {prototype/monitoring/cpu_period_baseline.csv};
    \addplot[draw=none, forget plot, name path=base_low] table [x=seconds, y=periodic_ci_low, col sep=comma] {prototype/monitoring/cpu_period_baseline.csv};
    \addplot[draw=none, forget plot, name path=base_high] table [x=seconds, y=periodic_ci_high, col sep=comma] {prototype/monitoring/cpu_period_baseline.csv};
    \addplot[plotA!20, opacity=0.5, forget plot] fill between[of=base_low and base_high];
    \addlegendentry{Baseline}

    \addplot[plotB, thick, name path=m100] table [x=seconds, y=periodic_mean, col sep=comma] {prototype/monitoring/cpu_period_n100.csv};
    \addplot[draw=none, forget plot, name path=l100] table [x=seconds, y=periodic_ci_low, col sep=comma] {prototype/monitoring/cpu_period_n100.csv};
    \addplot[draw=none, forget plot, name path=h100] table [x=seconds, y=periodic_ci_high, col sep=comma] {prototype/monitoring/cpu_period_n100.csv};
    \addplot[plotB!20, opacity=0.5, forget plot] fill between[of=l100 and h100];
    \addlegendentry{n=100}

    \addplot[plotC, thick, name path=m225] table [x=seconds, y=periodic_mean, col sep=comma] {prototype/monitoring/cpu_period_n225.csv};
    \addplot[draw=none, forget plot, name path=l225] table [x=seconds, y=periodic_ci_low, col sep=comma] {prototype/monitoring/cpu_period_n225.csv};
    \addplot[draw=none, forget plot, name path=h225] table [x=seconds, y=periodic_ci_high, col sep=comma] {prototype/monitoring/cpu_period_n225.csv};
    \addplot[plotC!20, opacity=0.5, forget plot] fill between[of=l225 and h225];
    \addlegendentry{n=225}

    \end{axis}
    \end{tikzpicture}
    \caption{\protocolname prepares to make requests (proof generation, share calculation, etc.) at the beginning of the round leading to slightly higher mean CPU utilization. Increased CPU load over the baseline is attributed to the node processing Requests and Responses (proof checks, record scoring, etc.).}
    \label{fig:cpu_period}
\end{figure}
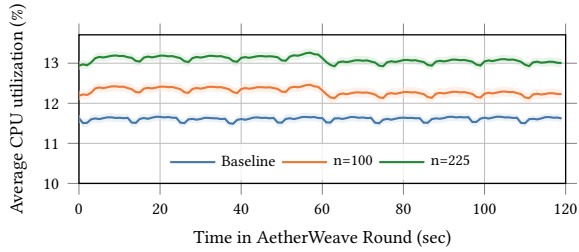

\textit{Per-round network traffic:}
\Cref{fig:net_rx_period} shows the network receive rate during a round
(transmit patterns are nearly identical).
Since requests are spread uniformly, traffic closely follows the baseline
with a constant offset.
\Cref{fig:net_rx_scaling} confirms that network overhead grows linearly
with~$\sqrt{n}$.

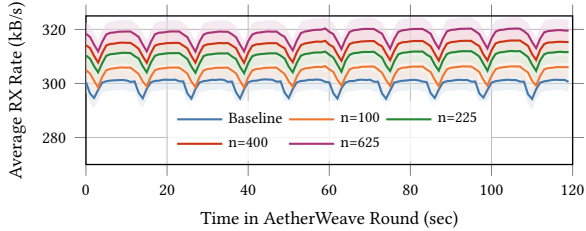
\begin{figure}
    \centering
    \begin{tikzpicture}
    \begin{axis}[
        awplot,
        width=0.95\linewidth,
        height=\awplotheight,
        xlabel={Time in AetherWeave Round (sec)},
        ylabel={Average RX Rate (kB/s)},
        xmin=0, xmax=120,
        ymin=270, ymax=325,
        grid=major,
        legend columns=3,
        legend style={at={(0.5,0.02)}, anchor=south, fill=white, fill opacity=0.8, text opacity=1}
    ]

    \addplot[plotA, thick, name path=base_mean] table [x=seconds, y=periodic_mean, col sep=comma] {prototype/monitoring/net_rx_period_baseline.csv};
    \addplot[draw=none, forget plot, name path=base_low] table [x=seconds, y=periodic_ci_low, col sep=comma] {prototype/monitoring/net_rx_period_baseline.csv};
    \addplot[draw=none, forget plot, name path=base_high] table [x=seconds, y=periodic_ci_high, col sep=comma] {prototype/monitoring/net_rx_period_baseline.csv};
    \addplot[plotA!20, opacity=0.5, forget plot] fill between[of=base_low and base_high];
    \addlegendentry{Baseline}

    \addplot[plotB, thick, name path=m100] table [x=seconds, y=periodic_mean, col sep=comma] {prototype/monitoring/net_rx_period_n100.csv};
    \addplot[draw=none, forget plot, name path=l100] table [x=seconds, y=periodic_ci_low, col sep=comma] {prototype/monitoring/net_rx_period_n100.csv};
    \addplot[draw=none, forget plot, name path=h100] table [x=seconds, y=periodic_ci_high, col sep=comma] {prototype/monitoring/net_rx_period_n100.csv};
    \addplot[plotB!20, opacity=0.5, forget plot] fill between[of=l100 and h100];
    \addlegendentry{n=100}

    \addplot[plotC, thick, name path=m225] table [x=seconds, y=periodic_mean, col sep=comma] {prototype/monitoring/net_rx_period_n225.csv};
    \addplot[draw=none, forget plot, name path=l225] table [x=seconds, y=periodic_ci_low, col sep=comma] {prototype/monitoring/net_rx_period_n225.csv};
    \addplot[draw=none, forget plot, name path=h225] table [x=seconds, y=periodic_ci_high, col sep=comma] {prototype/monitoring/net_rx_period_n225.csv};
    \addplot[plotC!20, opacity=0.5, forget plot] fill between[of=l225 and h225];
    \addlegendentry{n=225}

    \addplot[plotD, thick, name path=m400] table [x=seconds, y=periodic_mean, col sep=comma] {prototype/monitoring/net_rx_period_n400.csv};
    \addplot[draw=none, forget plot, name path=l400] table [x=seconds, y=periodic_ci_low, col sep=comma] {prototype/monitoring/net_rx_period_n400.csv};
    \addplot[draw=none, forget plot, name path=h400] table [x=seconds, y=periodic_ci_high, col sep=comma] {prototype/monitoring/net_rx_period_n400.csv};
    \addplot[plotD!20, opacity=0.5, forget plot] fill between[of=l400 and h400];
    \addlegendentry{n=400}

    \addplot[plotE, thick, name path=m625] table [x=seconds, y=periodic_mean, col sep=comma] {prototype/monitoring/net_rx_period_n625.csv};
    \addplot[draw=none, forget plot, name path=l625] table [x=seconds, y=periodic_ci_low, col sep=comma] {prototype/monitoring/net_rx_period_n625.csv};
    \addplot[draw=none, forget plot, name path=h625] table [x=seconds, y=periodic_ci_high, col sep=comma] {prototype/monitoring/net_rx_period_n625.csv};
    \addplot[plotE!20, opacity=0.5, forget plot] fill between[of=l625 and h625];
    \addlegendentry{n=625}

    \end{axis}
    \end{tikzpicture}
    \caption{AetherWeave's outbound network traffic pattern is similar to that of the baseline Ethereum client with an added offset because Requests are spread out across the round. The received traffic profile exhibits near identical patterns.}
    \label{fig:net_rx_period}
\end{figure}

\begin{figure}
    \centering
    \begin{tikzpicture}
    \begin{axis}[
        awplot,
        width=0.95\linewidth,
        height=\awplotheight,
        xlabel={$\sqrt{n}$},
        ylabel={Net RX Rate Overhead (kB/s)},
        ymin=0,
        grid=major,
        legend pos=north west,
    ]

    \addplot[
        color=plotA,
        thick,
        mark=*,
        error bars/y dir=both,
        error bars/y explicit,
        error bars/error bar style={thick, color=blue}
    ] table [
        x=nodes_sqrt,
        y=mean_norm,
        y error=error,
        col sep=comma
    ] {prototype/monitoring/net_rx_scaling_data.csv};

    \end{axis}
    \end{tikzpicture}

    \caption{AetherWeave's network RX rate overhead scales linearly with $\sqrt{n}$. Error bars are included in the plot, but are negligible, and not visible.}
    \label{fig:net_rx_scaling}
\end{figure}
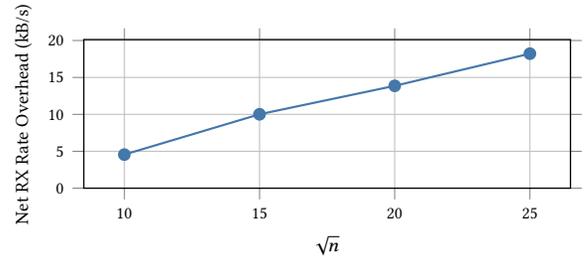

\section{Cryptographic Primitives}
\label{appendix:primitives}

\paragraph{Zero Knowledge Proofs.}
For any NP relation $\mathcal{R} \subseteq \mathcal{X} \times \mathcal{W}$,
given $x\in \mathcal{X},w \in \mathcal{W}$ that are public and private inputs to
the proof, we write: $\pi \gets \mathsf{ZKP}_{\mathcal{R}}(x; w)$ for a
non-interactive zero-knowledge proof that $(x,w) \in \mathcal{R}$, and
$\mathsf{ZKVerify}_{\mathcal{R}}(x, \pi) \in \{True, False\}$ for the
corresponding verifier.

\paragraph{Signature scheme.}
Let $(\keygen,\Sign,\CheckSig)$ be a digital signature scheme. For readability
we denote signed messages as $\langle m \rangle_\sig$.
\begin{itemize}
\item $\keygen$ is a deterministic, seeded key-generation algorithm.
\item Given a signing key $s$ and message $m$, $\Sign(s,m)$ outputs a signature
$\sig$.
\item Given a public key $pk$ and a signed message $\langle m \rangle_\sig$,\\
$\CheckSig(pk,\langle m \rangle_\sig)\in\{True, False\}$
verifies the signature.
\end{itemize}

We assume the signature scheme $(\keygen,\Sign,\CheckSig)$ is existentially
unforgeable under chosen-message attacks (EUF-CMA).

\paragraph{Hash Functions.}
Let $\lambda$ be the security parameter. Let $\hashstake$, $\hashid$,
$\hashshare$ be cryptographic hash functions of the form $H:\{0,1\}^*
\rightarrow \{0,1\}^{\lambda}$. We model these hash functions as random oracles.

\paragraph{Vector commitments.}
Let $(\mathsf{Commit},$ $\mathsf{Open},$ $\mathsf{Verify})$ be a
vector commitment scheme with implicit public parameters. In
particular, we assume access to a fixed collision-resistant hash
function, which is treated as a public parameter and omitted from the
syntax. Given a vector
$\vec{x} = (x_1,\ldots,x_n) \in \mathcal{X}^n$:
\begin{itemize}
    \item $(C, \mathsf{aux}) \gets \veccommit(\vec{x})$: commit to the vector
    $\vec{x}$
    \item $\pi \gets \vecopen(\mathsf{aux}, i)$: generate a proof for position
    $i$
    \item $\vecverify(C, i, x, \pi) \in \{True,False\}$: verify item $x$ is in
    position $i$.
\end{itemize}
The auxiliary information $\mathsf{aux}$ contains the information required to
open the commitment at any index (e.g., a Merkle tree built over $\vec{x}$). The
scheme satisfies the standard correctness and binding properties of vector
commitments. In particular, it is infeasible to open the same commitment to two
different values at the same index. In practice, our vector commitment schemes
are instantiated using Merkle trees.\footnote{While aggregating stake within the
smart contract could be implemented using a set commitment scheme (accumulator),
we instead use vector commitments to avoid introducing additional definitions
and notation. In other settings, such as commitments included in request
messages, standard accumulators are insufficient, since we require commitments
to sets of bounded size.}

\paragraph{ZK relations used by \protocolname.}
The following relations define the NP statements proved by the zero-knowledge
proofs in the protocol.

\begin{definition}[\Rshare  ; well-formed shares]
Given public inputs and witness
\[
x = (\reqcommit, \slashshare, \netpk, \roundnum),\quad
w = (\sk, \stakesk),
\]
we say $(x,w) \in \Rshare$ iff all of the following hold:

\begin{itemize}
  \item $\slashshare = \hashshare(\sk, \roundnum) \cdot \reqcommit + \stakesk$
  (when we interpret the elements in this equation as elements of field
  $\mathbb{F}$)
  \item $\stakesk = \hashstake(\sk)$
  \item $(*,\netpk) = \keygen(\sk)$
\end{itemize}
\end{definition}

\begin{definition}[\Rstake  ; proof-of-stake]
Given public inputs $x$ and witness $w$ where $$x = (\netpk,\stakeroot)$$ and
witness $$w = (\sk,\stakeID, VecP, i),$$ we say that $(x,w) \in \Rstake$ iff all
of the following hold:

\begin{itemize}
  \item $(\netsk, \netpk) = \keygen(sk)$
  \item $\stakeID = \hashid(\hashstake(\sk))$
  \item $\vecverify(\stakeroot, i, \stakeID, VecP)$ = 1
\end{itemize}
\end{definition}

\paragraph{Pseudorandom number generator (PRNG).}
A pseudorandom number generator is a deterministic function
$\PRNG: \{0,1\}^{\lambda} \times \{0,1\}^* \rightarrow [0,1)$
that, given a seed $\nonce \in \{0,1\}^{\lambda}$ and an input string $x$,
outputs a real number in $[0,1)$.

We require that for any probabilistic polynomial-time adversary,
the output of $\PRNG[\nonce][x]$ is computationally indistinguishable
from a uniformly random value in $[0,1)$, even given adaptive access
to $\PRNG[\nonce][\cdot]$ for the same seed.

\section{List of Symbols}
\label{appendix:symbols}

{\small
\renewcommand{\arraystretch}{1.1}
\topcaption{Summary of notation used throughout the paper.}
\label{tab:symbols}
\tablefirsthead{%
  \toprule
  \textbf{Symbol} & \textbf{Description} \\
  \midrule}
\tablehead{%
  \multicolumn{2}{@{}l}{\textit{\tablename\ \thetable{} -- continued}} \\
  \toprule
  \textbf{Symbol} & \textbf{Description} \\
  \midrule}
\tabletail{%
  \midrule
  \multicolumn{2}{r@{}}{\textit{continued on next page}} \\}
\tablelasttail{\bottomrule}
\begin{supertabular}{@{} p{0.35\columnwidth} p{0.58\columnwidth} @{}}
  \multicolumn{2}{@{}l}{\textbf{\textit{Keys \& Identifiers}}} \\
  \cmidrule(lr){1-2}
  \sk & Master secret key \\
  \netsk, \netpk & Network secret/public key \\
  \stakesk & Stake secret via $\hashstake(\sk)$ \\
  \stakeID & Stake identifier, $\hashid(\stakesk)$ \\
  \addr & Node network address \\[3pt]
  \multicolumn{2}{@{}l}{\textbf{\textit{Data Structures \& Tables}}} \\
  \cmidrule(lr){1-2}
  \netrecord & Network record \\
  \peerrecord & Peer record \\
  \request, \response & Heartbeat request/response \\
  \commitrecord & Commitment record \\
  \slashproof & Slashing evidence \\
  \ptable & Gossip table (global pool) \\
  \otable & Private overlay table \\
  \denylist & Deny list (slashed/quarantined keys) \\[3pt]
  \multicolumn{2}{@{}l}{\textbf{\textit{Cryptographic Primitives}}} \\
  \cmidrule(lr){1-2}
  \keygen & Key generation \\
  \Sign, \CheckSig & Signature create/verify \\
  \hashid, \hashstake, \hashshare & Hash functions \\
  \veccommit, \vecopen, \vecverify & Vector commitment ops \\
  \PRNG & Pseudorandom number generator \\
  \Rstake & ZK relation for proof-of-stake \\
  \Rshare & ZK relation for well-formed shares \\
  \stakeproof, \shareproof & ZK proofs (stake, shares) \\
  $\sig$ & Signature; $\langle\cdot\rangle_\sig$ denotes a signed payload \\[3pt]
  \multicolumn{2}{@{}l}{\textbf{\textit{Protocol Parameters}}} \\
  \cmidrule(lr){1-2}
  $n$ & Total number of nodes \\
  $s$ & Table scaling constant (table size is $s\sqrt{n}$) \\
  $\alpha$ & Adversary stake fraction \\
  $\theta$ & Detection threshold fraction \\
  \stakefreeze & Stake freeze period \\
  \stakewithdraw & Withdrawal delay \\
  \recordexp & Record expiration time \\
  \pconn & Overlay selection probability \\[3pt]
  \multicolumn{2}{@{}l}{\textbf{\textit{Protocol Variables}}} \\
  \cmidrule(lr){1-2}
  \nonce & Per-round gossip nonce \\
  \nonceoverlay & Per-round overlay nonce \\
  \reqcommit & Commitment to request recipients \\
  \slashshare & Slash share value \\
  \stakeroot & Epoch stake commitment \\
  \roundnum & Round number \\
  \indproof & Proof of inclusion at position \ind \\
  \timestamp & Record timestamp \\[3pt]
  \multicolumn{2}{@{}l}{\textbf{\textit{Mean-Field Analysis (\S5)}}} \\
  \cmidrule(lr){1-2}
  $q_r$ & Table quality at round $r$ \\
  $v_r$ & Node visibility at round $r$ \\
  $R_0 = s^2(1-\alpha)$ & Basic reproduction number \\[3pt]
  \multicolumn{2}{@{}l}{\textbf{\textit{Security Analysis (\S6)}}} \\
  \cmidrule(lr){1-2}
  $\Honset$, $\Advset$ & Honest / adversarial node sets \\
  $\gamma$ & Healthy network reachability fraction \\
  $\varphi$ & Critical partition fraction, $(1{+}\alpha)/2$ \\
  $\kappa$ & Overlay exponent parameter, $\kappa \in [0,1/2)$ \\
  $\delta$ & Overlay tolerance parameter \\
  $\varepsilon$ & Slack parameter, $(1{-}\alpha)/6$ \\
  $\mathcal{R}_i$ & Reachable set of node $i$ \\
  $\mathcal{B}$ & Blockchain \\
  $\mathsf{Game}^{\mathsf{stake\text{-}anon}}$ & Stake anonymity game \\
  $\mathsf{Adv}^{\mathsf{stake\text{-}anon}}$ & Stake anonymity advantage \\
  $\mathsf{Game}^{\mathsf{conn\text{-}priv}}$ & Connection privacy game \\
  $\mathsf{Adv}^{\mathsf{conn\text{-}priv}}$ & Connection privacy advantage \\
\end{supertabular}
}

\end{document}